\documentclass{./AIAA}
\usepackage{doi}
\usepackage{hyperref}
\usepackage{bm}
\usepackage{float}
\usepackage{amsmath}
\usepackage{amsfonts}
\usepackage{amsthm}
\usepackage{mathrsfs}
\usepackage{xcolor}
\usepackage[titletoc, title]{appendix}
\DeclareMathAlphabet\mathbfcal{OMS}{cmsy}{b}{n}

\newtheorem{theorem}{Theorem}
\newtheorem{lemma}{Lemma}

\newtheorem{definition}{Definition}
\newtheorem{assumption}{Assumption}
\newtheorem{proposition}{Proposition}

\begin{document}

\title{Robust nonlinear control of close formation flight}



\author{Qingrui Zhang\footnote{qingrui.zhang@mail.utoronto.ca} and Hugh H. T. Liu\footnote{liu@utias.utoronto.ca}}
\affiliation{University of Toronto, Toronto, Ontario  M3H 5T6, Canada}

\begin{abstract}
This paper investigates the robust nonlinear close formation control problem. It aims to achieve precise position control at dynamic flight operation for a follower aircraft under the aerodynamic impact due to the trailing vortices generated by a leader aircraft. One crucial concern is the control robustness that ensures the boundedness of position error subject to uncertainties and disturbances to be regulated with accuracy. This paper develops a robust nonlinear formation control algorithm to fulfill precise close formation tracking control. The proposed control algorithm consists of baseline control laws and disturbance observers. The baseline control laws are employed to stabilize the nonlinear dynamics of close formation flight, while the disturbance observers are introduced to compensate system uncertainties and formation-related aerodynamic disturbances. The position control performance can be guaranteed within the desired boundedness to harvest enough drag reduction for a follower aircraft in close formation using the proposed design.  The efficacy of the proposed design is demonstrated via numerical simulations of close formation flight of two aircraft.
\end{abstract}

\maketitle

\section*{Nomenclature}

\noindent\begin{tabular}{@{}lcl@{}}
\textit{$b$} &=& Wing span \\
\textit{$\mathbf{C}_{IW}$} &=& Rotation matrix from a wind frame to the inertial frame \\
\textit{$c_i$} &=& Control parameter ($i=V$, $\chi$, $p$, $q$, $r$) \\
\textit{$d_i$} &=& System uncertainties and disturbances ($i=V$, $\gamma$, $\chi$, $\mu$, $\alpha$, $\beta$,$p$, $q$, $r$) \\
\textit{$\widehat{d}_i$} &=& Estimate of system uncertainties and disturbances \\
\textit{$I_{x}$, $I_{xz}$, $I_{y}$, $I_{z}$}&=& Moments of inertia \\
\textit{$K_i$} &=& Control parameter ($i=x$, $z$, $V$, $\gamma$, $\chi$, $\mu$, $\alpha$, $\beta$,$p$, $q$, $r$) \\
\textit{$L$, $D$, $Y$, $T$}&=& Lift, drag, side force, and thrust, respectively \\
\textit{$\mathcal{L}$, $\mathcal{M}$, $\mathcal{N}$} &=& Rolling, pitching, and yawing moments \\
\textit{$l_{x}$, $l_{y}$, $l_z$} &=& The optimal relative positions in the inertial frame\\
\textit{$m_f$}&=& Mass of the follower aircraft\\
\textit{$p$, $q$, $r$}&=& Angular rates in the body frame \\
\textit{$r_{x}$, $r_{y}$}&=& The optimal relative positions in the wind frame of the leader aircraft \\
\textit{$\mathcal{T}_i$} &=& Time constant ($i=V$, $\gamma$, $\chi$, $W_x$, $W_y$, $W_z$,  $\mu$, $\alpha$, $\beta$, $p$, $q$, $r$)\\
\textit{$V_i$, $\gamma_i$, $\chi_i$}&=& Airspeed, flight path angle, and heading angle, respectively ($i=c$, $d$, $f$, $l$, $r$)\\
\textit{$\widehat{V}_f$, $\widehat{\gamma}_f$, $\widehat{\chi}_f$}&=& Resultant airspeed, flight path angle, and heading angle in trailing vortices\\
\textit{$W_x$, $W_y$, $W_z$}&=& Wake velocities induced by trailing vortices \\
\textit{$\widehat{W}_x$, $\widehat{W}_y$, $\widehat{W}_z$}&=& Estimates of wake velocities \\
\textit{$x_i$, $y_i$, $z_i$}&=& Position coordinates ($i=d$, $f$, $l$, $c$) \\
\textit{$x_e$, $y_e$, $z_e$}&=& Position errors in the inertial frame \\
\textit{$\alpha$, $\beta$, $\mu$}&=& Angle of attack, sideslip angle, and bank angle,  respectively \\
\textit{$\Delta L$, $\Delta D$, $\Delta Y$}&=& Vortex-induced lift, drag, and side force, respectively \\
\textit{$\Delta \mathcal{L}$, $\Delta \mathcal{M}$, $\Delta \mathcal{N}$} &=& Vortex-induced moments \\
\textit{$\delta_a$, $\delta_e$, $\delta_r$} &=& Aileron, elevator, and rudder  \\
\textit{$\zeta_i$} &=& Damping ratio ($i=l_x$, $l_y$, $l_z$, $V$, $\gamma$,   $\mu$, $\alpha$, $\beta$)\\
\textit{$\lambda_i$} &=& States of disturbance observer ($i=V$, $\gamma$, $\chi$, $W_x$, $W_y$, $W_z$,  $\mu$, $\alpha$, $\beta$, $p$, $q$, $r$)\\
\textit{$\xi_i$} &=& Auxiliary signals ($i=x$, $z$, $\mu$, $\alpha$, $\beta$)\\
\textit{$\omega_i$} &=& Natural frequencies ($i=l_x$, $l_y$, $l_z$, $V$, $\gamma$,  $\mu$, $\alpha$, $\beta$)\\
\end{tabular} \\

\newpage
\noindent\begin{tabular}{@{}lcl@{}}
\textit{Subscripts:}&& \\
\textit{$c$}&=& Command signal \\
\textit{$d$}&=& Desired signal \\
\end{tabular}  \\

\section{Introduction}
\begin{figure}[htbp]
\centering
\includegraphics[width=0.7\textwidth]{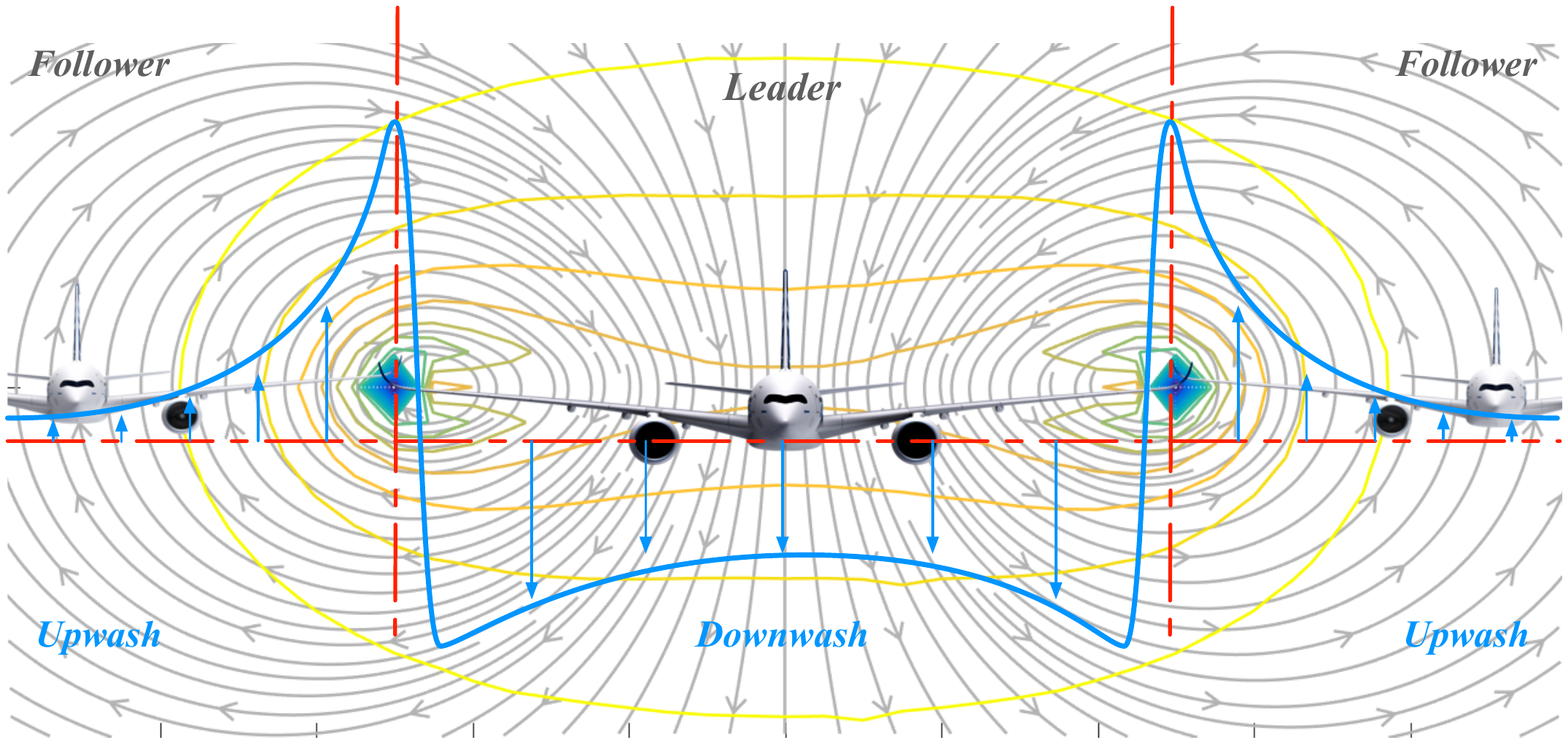}\\ \vspace{-7mm}
\caption{Close formation flight}\vspace{-6mm}
\label{Fig: CFF_2DVortex}
\end{figure}
Close formation flight is inspired by migratory birds who adopt a ``V-shape'' formation flight when migrating from one habitat to another \cite{Lissaman1970Science, Weimerskirch2001Nature, Portugal2014Nature}. In close formation, a follower aircraft, holding a close relative position to a leader aircraft, flies in the upwash wake region of the trailing vortices induced by the leader aircraft as shown in Figure \ref{Fig: CFF_2DVortex}, by which the follower aircraft reduces its drag and thus saves fuels.  Drag reduction of close formation flight has been demonstrated by simulations \cite{Halaas2014SciTech, Kent2015JGCD, Zhang2017JA}, observed by wind tunnel experiments \cite{Bangash2006JA, Cho2014JMST}, and confirmed by flight tests \cite{Ray2002AFMCE, Pahle2012AFMC, Bieniawski2014ASM,Hanson2018AFMC}. 
 
Despite its benefits, close formation flight is challenging in terms of the accuracy and robustness requirement for guidance and control \cite{Zhang2017JA}. The position control accuracy must be guaranteed within the consideration of system uncertainties and formation-related aerodynamic disturbances. 
Yet, different control algorithms have been proposed for close formation flight. Most of them are focused on level and straight flight with constant speeds \cite{Schumacher2000ACC, Pachter2001JGCD, Dogan2005JGCD, Almeida2015GNC}. Two different linear strategies were  applied, namely formation holding control and formation tracking control. Both of them are limited. Formation-holding control assumes a follower aircraft is initially well-trimmed at its optimal position in close formation, such as the proportional-integral (PI) formation control \cite{Pachter2001JGCD}, the close formation control by the linear-quadratic regulator (LQR) \cite{Dogan2005JGCD}, and the linear model predictive control (MPC)-based control \cite{Almeida2015GNC}. Linear formation-tracking control doesn't require the follower aircraft to be initially located at its desired position in close formation \cite{Binetti2003JGCD, Lavretsky2003GNC, Chichka2006JGCD, Zhang2016GNC, Zhang2017GNC, Zhang2018AESCTE}, but they are not guaranteed to address complex nonlinear aircraft dynamics at dynamic flight operation. Additionally, linear control methods will experience dramatic performance degradation or even fail to stabilize the system, when being applied to nonlinear systems. Robust nonlinear formation control is, therefore, more preferable to accommodate close formation flight at dynamic operation.

Nonlinear close formation control is challenging. Contrary to linear cases, model uncertainties and aerodynamic disturbances are less predictable and harder to be described at nonlinear scenarios, making close formation control more difficult. Early robust nonlinear close formation control was investigated using sliding model control \cite{Singh2000IJRNC}  or high order sliding mode (HOSM) control \cite{Galzi2007ACC}. Both of the two methods only focus on outer-loop design, and they require the vortex-upper bounds of induced forces or their derivatives to satisfy certain limits to guarantee the stability. The nonlinear robust formation design including both inner-loop and outer-loop control was reported in \cite{Brodecki2015JGCD} using an incremental nonlinear dynamic inversion (INDI) method, but this method is not robust to model uncertainties and external disturbances. Therefore, present nonlinear methods either rely on unknown model information to ensure both stability and robustness like the sliding mode control \cite{Galzi2007ACC}, or not robust to model uncertainties and external disturbances at general dynamic operations, such as the INDI-based control \cite{Brodecki2015JGCD}. Therefore, it is still an open issue for nonlinear robust close formation control with certain performance guaranty only using available model information.
  
In this paper, we investigate the robust nonlinear control problem for  close formation flight at dynamic operation. The control design is presented under a leader-following architecture.  The fundamental objective is to secure highly precise position control for close formation flight at dynamic flight operation design with the consideration of system uncertainties and aerodynamic impact caused by trailing vortices of leader aircraft a robust nonlinear controller. Control robustness will be one of the critical concerns which significantly affects the possible accuracy for close formation flight as it is subject to system uncertainties and aerodynamic disturbances. A  robust nonlinear  formation controller, which consists of baseline controllers and disturbance observers, is proposed in this paper. The baseline controllers are designed based on a command filtered backstepping technique to stabilize the nominal nonlinear dynamics of an aircraft in close formation \cite{Farrell2005JGCD, Farrell2009TAC}, whereas the disturbance observers could estimate and compensate for system uncertainties and formation-related aerodynamic disturbances by purely observing system inputs and available states. In the proposed design, the follower aircraft is required to track its dynamic optimal relative position to a leader aircraft in the inertial frame under different flight maneuvers. Both the inner-loop and outer-loop control will be studied in this paper, which makes the formation control design more reliably but also more difficult. The assumption on a well-designed inner-loop controller in \cite{Zhang2018TIE} is, therefore, removed in this paper. The proposed design is capable of achieving highly accurate and efficiently robust control performance without using any gradient or boundary information of formation aerodynamic disturbances. Position tracking errors will be ultimately bounded. The final boundaries could be regulated by choosing different control parameters.

The rest of the paper is organized as follows. Section \ref{Sec: Prelim} presents some preliminaries, while Section \ref{Sec: RefSignals} formulates reference trajectories. In Section \ref{Sec: NonlinCntrl}, robust nonlinear control is reported and analyzed.   Numerical simulations are given in Section \ref{Sec: NumSim}. Conclusions are  in Section \ref{Sec: Conclusion}.

\section{Preliminaries} \label{Sec: Prelim}
Some preliminaries are provided, which will be used for the design and analysis in the sequel.
\begin{definition}[Definition 4.6 \cite{Khalil2002Book}] \label{Definition: UniformAndUltBound} 
A system $\dot{\mathbf{x}}=\mathbf{f}(t\text{, }\mathbf{x})$ is uniformly ultimately bounded if there are positive constants $\mathcal{A}_b$  and $\mathcal{A}_c$, there exists $\mathcal{T}=\mathcal{T}\left(\mathcal{A}_a\text{, }\mathcal{A}_b\right)$ for any $\mathcal{A}_a\in \left(0\text{, }\mathcal{A}_c\right)$, such that
\begin{equation}
\Vert \mathbf{x}\left(t_0\right)\Vert \leq \mathcal{A}_a \Rightarrow \Vert \mathbf{x}\left(t\right)\Vert \leq \mathcal{A}_b\text{, }\qquad  \forall t\geq t_0+\mathcal{T}
\end{equation}
\end{definition}
\begin{lemma} \label{Lem: Chap5_DOB_General}
Let $d\left(t\right)$ be a bounded signal whose derivative $\dot{d}\left(t\right)$, namely $d\left(t\right)\in\mathscr{L}_\infty$ and $\dot{d}\left(t\right)\in\mathscr{L}_\infty$. Assume $\widehat{d}\left(t\right)$ is the estimation of $d\left(t\right)$ through a first-order filter as shown in (\ref{Eq: Chap5_1st-OrderFilter}).
\begin{equation}
\mathcal{T}_d\dot{\widehat{d}}=-\widehat{d}+d \label{Eq: Chap5_1st-OrderFilter}
\end{equation}
where $\mathcal{T}_d>0$ is a time constant. Define $\widetilde{d}=\widehat{d}-d$ as the estimation error. If $\widehat{d}\left(0\right)=0$,
\begin{enumerate}
\item $\widetilde{d}$  are globally bounded by $\vert\widetilde{d}\vert\leq \max\left\{\vert{d}\left(0\right)\vert\text{, }\mathcal{T}_d\Vert\dot{d}\Vert_\infty\right\} $;
\item if $\vert{d}\left(0\right)\vert\neq \mathcal{T}_d\Vert\dot{d}\Vert_\infty$, there exist any small positive constant $\epsilon$ and time $t_{\epsilon}$ such that $\vert\widetilde{d}\vert<\mathcal{T}_d\Vert\dot{d}\Vert_\infty+\epsilon$ for all $t>t_{\epsilon}$, where $t_{\epsilon}=\max\left\{0\text{, } \mathcal{T}_d\ln\left(\frac{\vert\vert{d}\left(0\right)\vert-\mathcal{T}_d\Vert\dot{d}\Vert_\infty\vert}{\epsilon}\right)\right\}$;
\item if $\lim_{t\to\infty} \dot{d} =0 $, $\lim_{t\to\infty} \widetilde{d} =0 $.
\end{enumerate}
\end{lemma}

\section{Formulation of reference trajectories at dynamic operation} \label{Sec: RefSignals}

In this section, a motion planner is designed for follower aircraft at close formation. According to \cite{Zhang2017JA}, the optimal relative position in close formation is static in the wind frame of the leader aircraft. Assume $\left[r_{x}\text{, } r_{y}\text{, } 0\right]^T$ is the static optimal relative position in the wind frame of the leader aircraft, where $r_x$ ranges from $-2b$ to $-10b$ and $r_y$ is around $\pm 0.95b$ with $b$ denoting the wing span. When flying at close formation, the reference position of a follower aircraft in the inertial frame is 
\begin{equation}
{x}_d ={x}_l+{l}_x \text{,}\quad {y}_d ={y}_l+{l}_y\text{,}\quad \text{and}\quad {z}_d ={z}_l+{l}_z  \label{Eq: RefPos}
\end{equation}
where  $x_l$, $y_l$, and $z_l$ are position coordinates of the leader aircraft in the inertial frame, and $\left[l_{x}\text{, } l_{y}\text{, } l_z\right]^T=\mathbf{C}_{IW}\left(\mu_l, \gamma_l, \chi_l\right)\left[r_{x}\text{, } r_{y}\text{, } 0\right]^T$ where $\mu_l$, $\gamma_l$, and $\chi_l$ are the bank, flight path, and heading angles of the leader aircraft. 
Differentiating (\ref{Eq: RefPos}) yields
\begin{equation}
\dot{x}_d =V_l\cos{\gamma_l}\cos{\chi_l}+\dot{l}_x \text{, }\quad 
\dot{y}_d =V_l\cos{\gamma_l}\sin{\chi_l}+\dot{l}_y \text{,}\quad \text{and}\quad 
\dot{z}_d =-V_l\sin{\gamma_l}+\dot{l}_z  \label{Eq: RefKinematics0}
\end{equation}
where $V_l$, $\gamma_l$, and $\chi_l$ are the airspeed, flight path angle, and heading angle of leader aircraft, respectively.  At dynamic operation,  $l_{x}$, $l_{y}$, and $l_z$ are time-varying, but their derivatives cannot be accurately computed. Hence, in the design, we introduce a command filter  (\ref{Eq: RefEstimator}) to get the command signals $l_{ci}$ and $ \dot{l}_{ci}$ ($i\in\left\{x\text{, } y\text{, }z\right\}$). Let $\mathscr{S}\left(t\right)$ be a smooth signal, so the command filter is 
\begin{equation}
\left[
\begin{array}{c}
\dot{\mathscr{S}}_c \\
\ddot{\mathscr{S}}_c
\end{array}\right]=
\underbrace{\left[\begin{array}{cc}
0\text{,}& 1 \\
-\omega_\mathscr{S}^2\text{,} &-2\zeta_\mathscr{S}\omega_\mathscr{S}
\end{array}\right]}_{\mathbf{A}_{\mathscr{S}}}\left[
\begin{array}{c}
{\mathscr{S}}_c \\
\dot{\mathscr{S}}_c
\end{array}\right]+\left[
\begin{array}{c}
0 \\
\omega_\mathscr{S}^2
\end{array}\right]\mathscr{S}\left(t\right) \label{Eq: RefEstimator}
\end{equation} 
where $\omega_\mathscr{S}>0$ is the natural frequency and $\zeta_\mathscr{S}>0$ is the damping ratio.  Let $\tilde{\mathscr{S}}=\mathscr{S}_c-\mathscr{S}$ and $\mathbf{e}_{\mathscr{S}}=\left[\tilde{\mathscr{S}}\text{, }\dot{\tilde{\mathscr{S}}}\right]^T$. If $\mathscr{S}_c\left(0\right)=\mathscr{S}\left(0\right)$ and $\dot{\mathscr{S}}_c\left(0\right)=0$, Lemma \ref{Lem: 2ndOrderEstProperties} exists for any bounded signal $\mathscr{S}\left(t\right)$.
\begin{lemma} \label{Lem: 2ndOrderEstProperties}
The estimator (\ref{Eq: RefEstimator}) is input-to-state stable with respect to $\mathscr{S}\left(t\right)$. If both ${ \dot{\mathscr{S}}}\left(t\right)$  and ${ \ddot{\mathscr{S}}}\left(t\right)$ are bounded,  $\tilde{\mathscr{S}}$ is uniformly and ultimately bounded, and the following inequality exists for $\mathbf{e}_{\mathscr{S}}\left(t\right)$.
\begin{equation*}\footnotesize
\Vert\mathbf{e}_{\mathscr{S}}\Vert_2\leq \sqrt{\frac{\lambda_{max}\left(\mathbf{P_{\mathscr{S}}}\right)}{\lambda_{min}\left(\mathbf{P_{\mathscr{S}}}\right)}}e^{-\frac{t}{2\lambda_{max}\left(\mathbf{P_{\mathscr{S}}}\right)}}\vert \dot{\mathscr{S}}\left(0\right)\vert+\left(1-e^{-\frac{t}{2\lambda_{max}\left(\mathbf{P_{\mathscr{S}}}\right)}}\right) \frac{2\lambda_{max}^2\left(\mathbf{P}_{\mathscr{S}}\right)\Vert\ddot{\mathscr{S}}+2\zeta_{\mathscr{S}}\omega_{\mathscr{S}}\dot{\mathscr{S}}\Vert_{\infty}}{\lambda_{min}\left(\mathbf{P}_{\mathscr{S}}\right)}
\end{equation*}
where $\lambda_{max}\left(\cdot\right)$ and $\lambda_{min}\left(\cdot\right)$ are the maximum and minimum eigenvalues of a matrix, respectively, $\mathbf{P}_{\mathscr{S}}$ is positive definite such that $\mathbf{P}_{\mathscr{S}}\mathbf{A}_{\mathscr{S}}+\mathbf{A}^T_{\mathscr{S}}\mathbf{P}_{\mathscr{S}}=-\mathbf{I}$. 
Furthermore, there exist 
\begin{equation*}\small
\tilde{\mathscr{S}}=\mathcal{O}\left({1}/{\omega_{\mathscr{S}}}\right) \qquad \text{ and }\qquad \dot{\tilde{\mathscr{S}}}/{\omega_ {\mathscr{S}}}=\mathcal{O}\left({1}/{\omega_ {\mathscr{S}}}\right)\vspace{-3mm}
\end{equation*}
where $\mathcal{O}\left(\cdot\right)$ is an order of magnitude notation \cite{Khalil2002Book}. 
\end{lemma}
In real flight, $\dot{\gamma}_l$,  $\dot{\chi}_l$, and $\dot{\mu}_l$ and their derivatives are all bounded, so (\ref{Eq: RefEstimator}) is a valid choice to obtain $l_{ci}$ and $ \dot{l}_{ci}$ ($i\in\left\{x\text{, } y\text{, }z\right\}$). In addition, Lemma \ref{Lem: 2ndOrderEstProperties} indicates that the command signals $l_{ci}$ and $ \dot{l}_{ci}$ ($i\in\left\{x\text{, } y\text{, }z\right\}$) could be ensured to be arbitrarily close to their corresponding desired values $l_{i}$ and $ \dot{l}_{i}$ by choosing proper command filter parameters. Note that $\mathscr{S}_c\left(0\right)=\mathscr{S}\left(0\right)$  is needed to avoid the peaking phenomenon (Page 613,   \cite{Khalil2002Book}). If $\mathscr{S}_c\left(0\right)\neq\mathscr{S}\left(0\right)$, $\mathscr{S}_c$, the signal $\dot{\mathscr{S}}_c$ will transiently peak to $\mathcal{O}\left({\omega_\mathscr{S}}\right)$ before it exponentially decays, resulting in the so-called peaking phenomenon due to $\mathscr{S}_c\left(0\right)\neq\mathscr{S}\left(0\right)$. Without loss of generality, the following assumption is introduced. 
\begin{assumption} \label{Assump: Chap5_BoundedSignal}
The attitude signals $\gamma_l$,  $\chi_l$, and $\mu_l$ and their derivatives are all bounded.
\end{assumption}
In light of (\ref{Eq: RefEstimator}), the command position for a follower aircraft in close formation is ${x}_r=x_l+{l}_{cx}$, ${y}_r=y_l+{l}_{cy}$, and ${z}_r=z_l+{l}_{cz}$, and accordingly,
\begin{equation}
\left\{
\begin{array}{l}
\dot{{x}}_r=\dot{x}_l+\dot{l}_{cx} \text{,}\qquad
\dot{{y}}_r=\dot{y}_l+\dot{l}_{cy}\sin{\chi_r} \text{,}\qquad
\dot{{z}}_r=\dot{z}_l+\dot{l}_{cz}\\
V_r 		  = \sqrt{\dot{x}_r^2+\dot{y}_r^2+\dot{z}_r^2} \text{, }\quad
\gamma_r = -\sin^{-1}\left(\frac{\dot{z}_r}{V_r}\right) \text{, }\quad
\chi_r	  =\chi_l+\sin^{-1}\left(\frac{-\dot{l}_{cx}\sin{\chi_l}+\dot{l}_{cy}\cos{\chi_l}}{V_r\cos{\gamma_r}}\right) 
\end{array}\right. \label{Eq: MotionPlan}
\end{equation}

\section{Robust nonlinear formation control design} \label{Sec: NonlinCntrl}
\begin{figure}[htbp]
\centering
\includegraphics[width=0.8\textwidth]{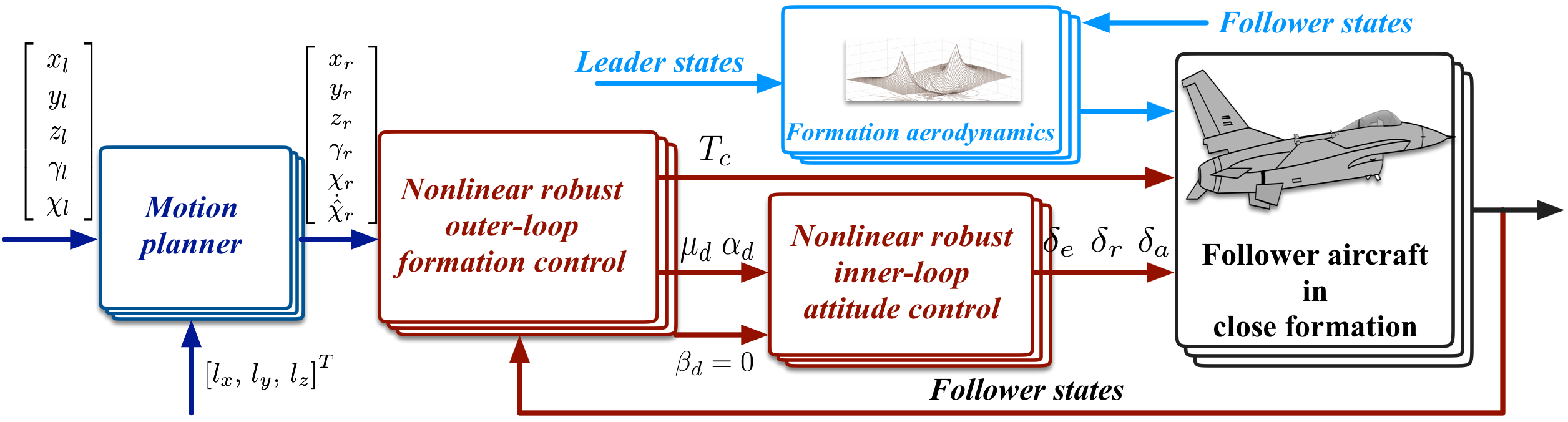} \\ \vspace{-5mm}
\caption{The entire formation control structure}\vspace{-4mm}
\label{Fig: Chap5_EntireCntrl}
\end{figure}

The proposed design in this section can be easily extended to the case with more than three aircraft, even though it is discussed under the leader-follower architecture with two aircraft. In the proposed design, command filtered backstepping technique is employed, which avoids the analytic calculation of time derivatives of intermediate virtual inputs \cite{Farrell2005JGCD, Sonneveldt2007JGCD, Farrell2009TAC, Sonneveldt2009JGCD}. As shown in Figure \ref{Fig: Chap5_EntireCntrl}, the entire design consists of two major loops: an outer loop for formation position control and an inner loop for attitude control. The outer-loop control allows a follower aircraft to track the planned motion by (\ref{Eq: MotionPlan}), and generates command thrust $T_c$, desired bank angle $\mu_d$, and desired angle of attack $\alpha_d$. The inner-loop control stabilizes follower aircraft's attitudes to their desired values $\mu_d$ and $\alpha_d$ from the outer-loop control, while holding zero sideslip angle $\beta_f$.

\subsection{Outer-loop formation position control} \label{Subsec: Chap5_OuterloopForm}
Let  $\overline{D}$ and $\overline{L}$ be the nominal values of the drag $D$ and lift $L$, respectively. They are obtained by either available aerodynamic data or certain analytical models \cite{Morelli1998ACC}.  The sideslip angle $\beta_f$ is negligibly small, as it is always stabilized to be zero. Accordingly, the side force $Y$ is small and taken as a model uncertainty. The outer-loop dynamics used for control design are  
\begin{equation}
\left\{\begin{array}{lcl}
\dot{x}_f &=&V_f\cos{\gamma_f}\cos{\chi_f}+W_{x}  \\
\dot{y}_f &=&V_f\cos{\gamma_f}\cos{\chi_f}+W_{y}\\
\dot{z}_f &=&-V_f\sin{\gamma_f}+W_{z}
\end{array}\right.  \qquad
\left\{\begin{array}{lcl}
\dot{V}_f&=& \frac{1}{m_f}\left(T\cos{\alpha_f}\cos{\beta_f}-\overline{D}\right)-g\sin{\gamma_f}+d_{V}\\
\dot{\gamma}_f&=&\frac{\left(\overline{L}+T\sin{\alpha_f}\right)\cos{\mu_f}}{m_fV_f}-\frac{g}{V_f}\cos{\gamma_f}+d_{\gamma}\\
\dot{\chi}_f&=&\frac{\left(\overline{L}+T\sin{\alpha_f}\right)\sin{\mu_f}}{m_fV_f\cos{\gamma_f}}+d_{\chi}
\end{array}\right.\label{Eq: Chap5_OuterDesignModel}
\end{equation}
where $x_f$, ${y}_f$, and ${z}_f$ are follower position coordinates in the inertial frame, $V_f$ is the airspeed, $\gamma_f$ and $\chi_f$ are the flight path and heading angles,  $T$ is the thrust, $W_{x}$, $W_{y}$, and $W_{z}$ are induced wake velocities, and $d_{V}$, $d_{\gamma}$, and $d_{\chi}$ are the augmentation of system uncertainties and disturbances.
\begin{equation}
\left\{\begin{array}{lcl}
d_{V}&=&-\frac{D-\overline{D}+\Delta D}{m_f}-\frac{\dot{W}_{W_{x}}}{m_f}\\
d_{\gamma}&=&\frac{\left(L-\overline{L}+\Delta L\right)\cos{\mu_f}-\left(Y+\Delta Y-T\cos{\alpha_f}\sin{\beta_f}\right)\sin{\mu_f}}{m_fV_f}-\frac{\dot{W}_{W_{y}}\sin{\mu_f}+\dot{W}_{W_{z}}\cos{\mu_f}}{m_fV_f}\\
d_{\chi}&=&\frac{\left(L-\overline{L}+\Delta L\right)\sin{\mu_f}+\left(Y+\Delta Y-T\cos{\alpha_f}\sin{\beta_f}\right)\cos{\mu_f}}{m_fV_f\cos{\gamma_f}}+\frac{\dot{W}_{W_{y}}\cos{\mu_f}-\dot{W}_{W_{z}}\sin{\mu_f}}{m_fV_f\cos{\gamma_f}}\\ 
\end{array}\right. \label{Eq: Chap5_Uncer+Disturb}
\end{equation}
where $\dot{W}_{W_{x}}$, $\dot{W}_{W_{y}}$, and $\dot{W}_{W_{z}}$ are the wake velocity derivatives denoted in the wind frame of follower aircraft, $\Delta L$, $\Delta D$, and $\Delta Y$ are the vortex-induced forces. According to \cite{Zhang2017JA}, $W_{x}$, $W_{y}$, and $W_{z}$  are bounded, and have much slower dynamics in comparison with aircraft speed and attitudes,  so their derivatives are relatively small. Furthermore, the following assumption is introduced.
\begin{assumption} \label{Assump: Chap5_BoundW}
Induced wake velocities $W_{x}$, $W_{y}$, and $W_{z}$  are all bounded, and furthermore, they are piecewise constant, namely $\dot{W}_{{x}}=0$, $\dot{W}_{{y}}=0$, and $\dot{W}_{{z}}=0$.
\end{assumption}
The following nonlinear disturbance observer is employed. 
\begin{equation}
\widehat{\mathbf{W}}  =  \boldsymbol{\lambda}_W+\boldsymbol{\mathcal{T}}^{-1}_W\mathbf{X_P}\text{, }\qquad
\boldsymbol{\dot{\lambda}}_W =-\boldsymbol{\mathcal{T}^{-1}}_W\boldsymbol{\lambda}_W-\boldsymbol{\mathcal{T}^{-1}}_W\left(\boldsymbol{\mathcal{T}}^{-1}_W\mathbf{X_P}+\mathbf{{U}_P}\right)  \label{Eq: Chap5_WakeEst}
\end{equation} 
where $\boldsymbol{\mathcal{T}}_W=diag\left\{\mathcal{T}_{Wx}\text{, }\mathcal{T}_{Wy}\text{, }\mathcal{T}_{Wz}\right\}>0$, $\boldsymbol{\lambda}_W=\left[\lambda_{Wx}\text{, }\lambda_{Wy}\text{, }\lambda_{Wz}\right]^T$, $\mathbf{X_P}=\left[x_f\text{, }y_f\text{, }z_f\right]^T$, $\mathbf{{U}_P}=\left[V_f\cos{\gamma_f}\cos{\chi_f}\text{, }V_f\cos{\gamma_f}\sin{\chi_f}\text{, }-V_f\sin{\gamma_f}\right]^T$, $\widehat{\mathbf{W}}=\left[\widehat{{W}}_x\text{, }\widehat{{W}}_y\text{, }\widehat{{W}}_z\right]^T$ where $\widehat{{W}}_x$, $\widehat{{W}}_y$, and $\widehat{{W}}_z$ are estimates of $W_{x}$, $W_{y}$, and $W_{z}$, respectively.
It is chosen that $\boldsymbol{\lambda}_W\left(0\right)=-\boldsymbol{\mathcal{T}}_W^{-1}\mathbf{X_P}\left(0\right)$.
Let $\widetilde{W}_x=\widehat{W}_x-W_x$,  $\widetilde{W}_y=\widehat{W}_y-W_y$,  and $\widetilde{W}_z=\widehat{W}_z-W_z$.  Under Assumption \ref{Assump: Chap5_BoundW}, one has 
\begin{equation}
\dot{\widetilde{W}}_x=-\widetilde{W}_x /\mathcal{T}_{Wx}\text{, }\qquad \dot{\widetilde{W}}_y=-\widetilde{W}_y/\mathcal{T}_{Wy}\text{,} \qquad\text{and }\quad\dot{\widetilde{W}}_z=-\widetilde{W}_z/\mathcal{T}_{Wz} \label{Eq: EstErrDyn}
\end{equation}
According to Lemma \ref{Lem: Chap5_DOB_General}, $\widetilde{W}_x$, $\widetilde{W}_y$, and $\widetilde{W}_z$ can be made arbitrarily small by choosing sufficiently small time constants, even if $\dot{W}_{{x}}\neq0$, $\dot{W}_{{y}}\neq0$, and $\dot{W}_{{z}}\neq0$. To simplify the analysis, it is assumed that  $\widetilde{W}_x =\widetilde{W}_y =\widetilde{W}_z=0 $.
In light of (\ref{Eq: Chap5_WakeEst}), one has 
\begin{equation}
\dot{x}_f =\widehat{V}_f\cos{\widehat{\gamma}_f}\cos{\widehat{\chi}_f}  \text{,}\quad 
\dot{y}_f =\widehat{V}_f\cos{\widehat{\gamma}_f}\sin{\widehat{\chi}_f} \text{,} \quad \text{and} \quad
\dot{z}_f =-\widehat{V}_f\sin{\widehat{\gamma}_f}
 \label{Eq: NewKinF2}
\end{equation}
where 
\begin{equation*}\small
\left\{\begin{array}{ll}
\widehat{V}_f &= \sqrt{\left({V}_f\cos{{\gamma}_f}\cos{{\chi}_f}+\widehat{W}_x\right)^2+\left({V}_f\cos{{\gamma}_f}\cos{{\chi}_f}+\widehat{W}_y\right)^2+\left(-{V}_f\sin{{\gamma}_f}+\widehat{W}_z \right)^2}\\
\widehat{\gamma}_f &= -\sin^{-1}\left(\frac{\widehat{W}_{z}-{V}_f\sin{\gamma_f}}{\widehat{V}_f}\right) \\
 \widehat{\chi}_f &=\chi_f+\sin^{-1}\left(\frac{\widehat{W}_y\cos{\chi_f}-\widehat{W}_x\sin{\chi_f}}{\widehat{V}_f\cos{\widehat{\gamma}_f}}\right)
 \end{array}\right.
 \end{equation*}  
Let $x_e=x_f-{x}_r$, $y_e=y_f-{y}_r$, and $z_e=z_f-{z}_r$. Transform $x_e$, $y_e$, and $z_e$ into a new frame.
\begin{equation*}
e_x =\cos{\widehat{\chi}_f}x_e + \sin{\widehat{\chi}_f}y_e\text{, }\quad 
e_y =-\sin{\widehat{\chi}_f}x_e+\cos{\widehat{\chi}_f}y_e\text{, }\quad 
e_z=z_e
\end{equation*}
We have
\begin{equation}
\left\{\begin{array}{lcl}
\dot{e}_x &=&\widehat{V}_f\cos{\widehat{\gamma}_f}-{V}_r\cos{{\gamma}_r}\cos{e_\chi}+\dot{\widehat{\chi}}_f e_y  \\
\dot{e}_y &=&{V}_r\cos{{\gamma}_r}\sin{e_\chi}-\dot{\widehat{\chi}}_f e_x\\
\dot{e}_z &=&-\widehat{V}_f\sin{\widehat{\gamma}_f}+{V}_r\sin{{\gamma}_r}
\end{array}\right. \label{Eq: NewKinF}
\end{equation}
where $e_\chi=\widehat{\chi}_f-\chi_r$. The desired velocity and flight path angle are shown in (\ref{Eq: Chap5_Desired_V_Gamma}).
\begin{equation}
V_d 		  = \left(-K_xe_x+V_r\cos\gamma_r\cos{e_\chi}\right)/\cos{\widehat{\gamma}_f}-\delta V\text{, }\quad
\gamma_d  = \sin^{-1}\left(\left({K_ze_z+V_r\sin\gamma_r+\widehat{W}_z}\right)/{V_f}\right) \label{Eq: Chap5_Desired_V_Gamma}
\end{equation}
where $K_x\text{, }K_z>0$ are control parameters, and $\delta V=\widehat{V}_f-V_f$. The desired signals $V_d$ and $\gamma_d$ are passed through a command filter to obtain $V_c$, $\gamma_c$, and their rates. Hence, 
\begin{equation}
\left[
\begin{array}{c}
\dot{\mathscr{S}}_c \\
\ddot{\mathscr{S}}_c
\end{array}\right]=
\left[\begin{array}{cc}
0\text{,}& 1 \\
-\omega_\mathscr{S}^2\text{,} &-2\zeta_\mathscr{S}\omega_\mathscr{S}
\end{array}\right]\left[
\begin{array}{c}
\mathscr{S}_c \\
\dot{\mathscr{S}}_c
\end{array}\right]+\left[
\begin{array}{c}
0 \\
\omega_\mathscr{S}^2
\end{array}\right]\mathscr{S}_d \text{,  } \quad\mathscr{S}\in\left\{V\text{, } \gamma\right\} \label{Eq: Chap5_VGa_CMD}
\end{equation} 
Let $e_V=V_f-V_c$ and $e_\gamma=\gamma_f-\gamma_c$. Note that $\overline{L}=\overline{L}_0+\overline{L}_\alpha\alpha_f$ where $\overline{L}_0$ is the lift at $\alpha_f=0$ and $\overline{L}_\alpha$ is the lift derivative with respect to the angle of attack. According to (\ref{Eq: Chap5_OuterDesignModel}), one has
\begin{equation}
\left\{\begin{array}{lcl}
\dot{e}_V &=& \underbrace{\frac{T\cos{\alpha_f}\cos{\beta_f}-\overline{D}}{m_f}-g\sin{\gamma_f}}_{u_V}-\dot{V}_c+d_{V}\\
\dot{e}_\gamma&=&\underbrace{\frac{\left(\overline{L}_0+\overline{L}_\alpha \alpha_f+T\sin{\alpha_f}\right)\cos{\mu_f}-m_fg\cos{\gamma_f}}{m_fV_f}}_{u_\gamma}-\dot{\gamma}_c+d_{\gamma}\\
\dot{e}_\chi&=&\underbrace{\frac{\left(\overline{L}_0+\overline{L}_\alpha \alpha_f+T\sin{\alpha_f}\right)\sin{\mu_f}}{m_fV_f\cos{\gamma_f}}}_{u_\chi}-\dot{\widehat{\chi}}_r+{d}_\chi \end{array}\right. \label{Eq: Chap5_ErrDyn_VGammaChi}
\end{equation}
where $u_V$, $u_\gamma$, and $u_{\chi}$ are intermediate control inputs, $\dot{\widehat{\chi}}_r$ is the estimation of $\dot{\chi}_r$ by passing $\chi_r$ through a 2nd-order filter similar to (\ref{Eq: RefEstimator}).
Two more uncertainty terms ${\dot{\widetilde{\chi}}}_r$ and $\dot{\widetilde{\chi}}_f$ are included in ${d}_\chi$ in (\ref{Eq: Chap5_ErrDyn_VGammaChi}), where  ${\dot{\widetilde{\chi}}}_r={\dot{\widehat{\chi}}}_r-\dot{\chi}_r$ and $\dot{\widetilde{\chi}}_f=\frac{d}{dt}\sin^{-1}\left(\frac{{\widehat{W}}_y\cos{\chi_f}-\widehat{W}_x\sin{\chi_f}}{\widehat{V}_f\cos{\widehat{\gamma}_f}}\right)$. Hence, ${d}_\chi$ in (\ref{Eq: Chap5_ErrDyn_VGammaChi})  is re-defined to be ${d}_\chi=\frac{\left(L-\overline{L}+\Delta L\right)\sin{\mu_f}+\left(Y+\Delta Y-\cos{\alpha_f}\sin{\beta_f}\right)\cos{\mu_f}}{m_fV_f\cos{\gamma_f}}+\dot{\widetilde{\chi}}_r+\dot{\widetilde{\chi}}_f+\frac{\dot{W}_{W_{y}}\cos{\mu_f}-\dot{W}_{W_{z}}\sin{\mu_f}}{m_fV_f\cos{\gamma_f}}$.
Note that $\chi_r$ is a smooth signal with bounded derivatives. According to Lemma \ref{Lem: 2ndOrderEstProperties}, ${\dot{\widetilde{\chi}}}_r$ and its derivative are bounded, so ${\dot{\widetilde{\chi}}}_f$ and its derivative are also bounded. 
\begin{assumption}\label{Assump: Chap5_Uncer+Distur_Bound}
The uncertainties and disturbances $d_{V}$, $d_{\gamma}$, $d_{\chi}$, and their derivatives are bounded.
\end{assumption}
The following virtual inputs $u_V^d$, $u_\gamma^d$, and $u_\chi^d$ are choosing for the error systems (\ref{Eq: NewKinF}) and (\ref{Eq: Chap5_ErrDyn_VGammaChi}). \vspace{-3mm}
\begin{equation}
u_V^d = u_{V 0}^d -\widehat{d}_V+\dot{V}_c \text{,}\quad u_\gamma^d = u_{\gamma 0}^d -\widehat{d}_\gamma+\dot{\gamma}_c \text{,} \quad \text{and} \quad u_\chi^d = u_{\chi 0}^d -\widehat{d}_\chi+\dot{\widehat{\chi}}_r \label{Eq: Chap5_FormCntrl_Entire} \vspace{-3mm}
\end{equation}
where $\widehat{d}_V$, $\widehat{d}_\gamma$, and $\widehat{{d}}_\chi$ are estimates of ${d}_V$, ${d}_\gamma$, and ${{d}}_\chi$, respectively, and $u_{V0}^d$, $u_{\gamma0}^d$, and $u_{\chi0}^d$ are 
\begin{equation}
{u}_{V0}^d	 = -K_Ve_V-\frac{c_V\varepsilon_x\cos\widehat{\gamma}_f}{H}\text{,}\quad
u_{\gamma0}^d = -K_\gamma e_\gamma \text{,} \quad\text{and}\quad
u_{\chi0}^d	 = -K_\chi e_{\chi}-\frac{c_\chi e_y V_r\cos{\gamma_r}\cos\left(\frac{e_\chi}{2}\right)}{H}  \label{Eq: Chap5_NonmCntrl_VChiGa}
\end{equation}
where $K_V$, $K_\gamma$, $K_\chi$, $c_V$, $c_\chi>0$ are control parameters,  $H=\sqrt{\varepsilon_x^2+e_y^2+1}$, $\varepsilon_x=e_x-\xi_x$, and $\varepsilon_z=e_z-\xi_z$, where $\xi_x$ and $\xi_z$ in (\ref{Eq: Chap5_Auxili_X_Z}) are used to counteract the estimation errors in the  filter (\ref{Eq: Chap5_VGa_CMD}). \vspace{-3mm}
\begin{equation}
\dot{\xi}_x = -K_x\xi_x+\left(V_c-V_d\right)\cos\widehat{\gamma}_f\text{,}\qquad
\dot{\xi}_z = -K_z\xi_z+V_f\left(\sin\gamma_d-\sin\gamma_f\right)\label{Eq: Chap5_Auxili_X_Z} \vspace{-3mm}
\end{equation}

\begin{figure}[tbp]
\centering
\includegraphics[width=0.45\textwidth]{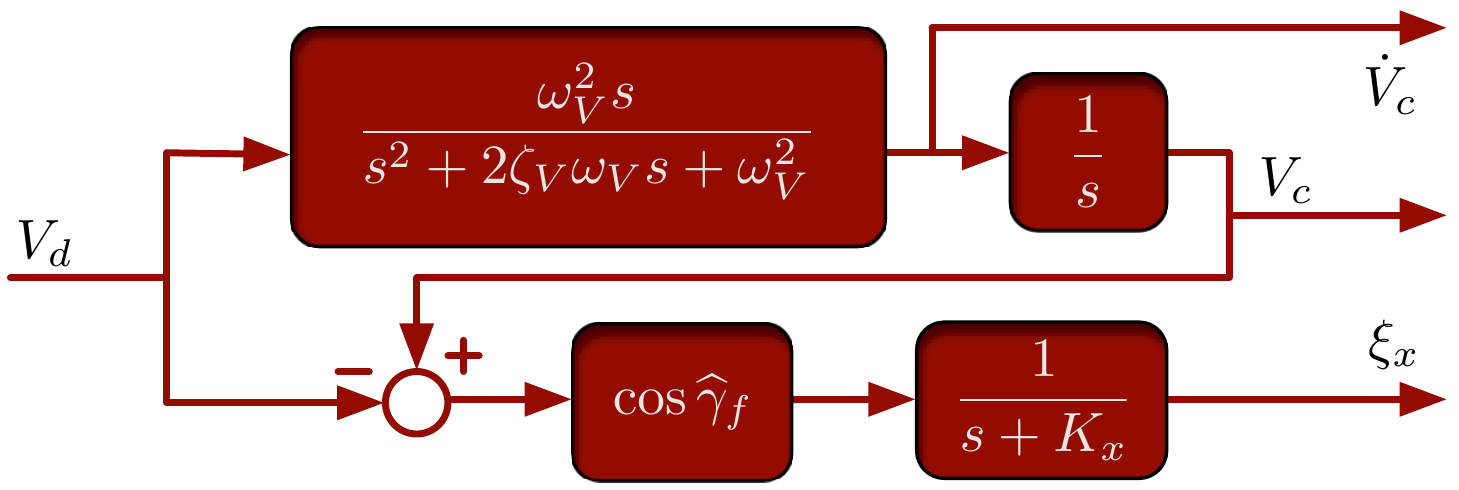}\\ \vspace{-6mm}
\caption{Command filter and auxiliary system for speed control}\vspace{-4mm}
\label{Fig: Speed_CF}
\end{figure}
Shown in Figure \ref{Fig: Speed_CF} is the command filter and auxiliary system for speed control.  If one chooses $u_V={u}_{V}^d$, $u_\gamma={u}_{\gamma}^d$, and $u_\chi={u}_{\chi}^d$, there exists
\begin{equation}
\left\{\begin{array}{lcl}
\dot{V}_f &=& -K_Ve_V-c_V\frac{\varepsilon_x\cos\widehat{\gamma}_f}{H}+\dot{V}_c-\widehat{d}_V+d_{V}\\
\dot{\gamma}_f&=&-K_\gamma e_\gamma+\dot{\gamma}_c-\widehat{d}_\gamma+d_{\gamma}\\
\dot{\chi}_f&=& -K_\chi \sin\left(\frac{e_\chi}{2}\right)-\frac{c_\chi e_y V_r\cos{\gamma_r}\cos\left(\frac{e_\chi}{2}\right)}{H}+\dot{\widehat{\chi}}_r-\widehat{{d}}_\chi+{d}_\chi \end{array}\right. \label{Eq: Chap5_SteadyErrDyn_VGammaChi}
\end{equation}
Based on (\ref{Eq: Chap5_SteadyErrDyn_VGammaChi}), the nonlinear disturbance observer  is
\begin{equation}
\mathbf{\widehat{D}_{D}}  		=  \boldsymbol{\lambda_D}+\boldsymbol{\mathcal{T}_D^{-1}}\mathbf{X_D} \text{, }\qquad
\boldsymbol{\dot{\lambda}_{D}}	=-\boldsymbol{\mathcal{T}_D^{-1}}\boldsymbol{\lambda_D}-\boldsymbol{\mathcal{T}_D^{-1}}\left(\boldsymbol{\mathcal{T}_D^{-1}}\mathbf{X_D} +\mathbf{U_D}-\mathbf{\widehat{D}_{D}}   \right) \label{Eq: Chap5_DO_GaChiV_Mat}
\end{equation}
where $\boldsymbol{\mathcal{T}_D}=diag\left\{\mathcal{T}_V\text{, }\mathcal{T}_\gamma\text{, }\mathcal{T}_\chi\right\}>0$ is a positive definite constant matrix, $\boldsymbol{\lambda_D}=\left[\lambda_V\text{, }\lambda_\gamma\text{, }\lambda_\chi\right]^T$, $\mathbf{X_D}=\left[V_f\text{, }\gamma_f\text{, }\chi_f\right]^T$, $\mathbf{U_D}=\left[{u}_{V0}^d\text{, }u_{\gamma0}^d\text{, }u_{\chi0}^d\right]^T$, and $\mathbf{\widehat{D}_{D}}=\left[\widehat{d}_V\text{, }\widehat{d}_\gamma\text{, }\widehat{d}_\chi\right]^T$. The uncertainty and disturbance estimates $\widehat{d}_V$, $\widehat{d}_\gamma$, and $\widehat{{d}}_\chi$ need to be fed back to the estimator for the next estimation updates. Combining (\ref{Eq: Chap5_NonmCntrl_VChiGa}) and (\ref{Eq: Chap5_DO_GaChiV_Mat}), one is able to get $u_V^d$, $u_\gamma^d$, and $u_\chi^d$.  Hence,
 \begin{equation}
\left\{\begin{array}{lcl}
T_c &=& \frac{m_f\left(u_V^d+g\sin{\gamma_f}\right)+\overline{D}}{\cos{\alpha_ f}\cos{\beta_f}}\\
\alpha_d &=&\frac{m_fV_f\sqrt{\left(u_\gamma^d+\frac{g}{V}\cos{\gamma_f}\right)^2+\left(u_\chi^d\right)^2\cos^2{\gamma_f}}-T\sin{\alpha_f}-\overline{L}_0}{\overline{L}_\alpha} \\
\mu_d &=&\tan^{-1}\left(\frac{m_fV_fu_\chi^d\cos{\gamma_f}}{m_fV_fu_\gamma+m_fg\cos{\gamma_f}}\right)
\end{array}\right. \label{Eq: Chap5_Desired_TcAoAMu}
\end{equation}
\begin{figure}[tbp]
\centering
\includegraphics[width=0.9\textwidth]{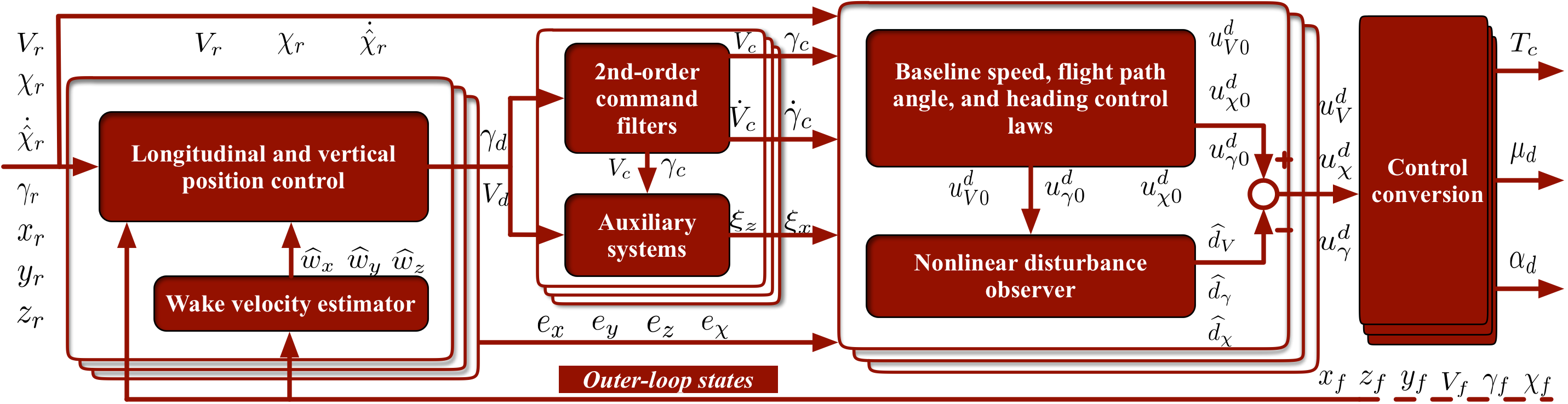} \\ \vspace{-6mm}
\caption{The outer-loop formation position control structure}\vspace{-4mm}
\label{Fig: Chap5_Outerloop}
\end{figure}
The entire outer-loop control structure is illustrated in Figure \ref{Fig: Chap5_Outerloop}. The following assumption is introduced for ${d_V}$, ${d_\gamma}$, and ${d_\chi}$ for the stability analysis.
\begin{assumption} \label{Assump: Chap5_SlowDisturb}
${d_V}$, ${d_\gamma}$, and ${d_\chi}$ have slow dynamics, and furthermore, $\dot{d}_V\simeq 0$, $\dot{d}_\gamma \simeq 0$, and $\dot{d}_\chi\simeq 0$.
\end{assumption}
The following error dynamics will be obtained.
\begin{equation}
\left\{\begin{array}{lcl}
\dot{\varepsilon}_x &=&\widehat{V}_f\cos{\widehat{\gamma}_f}-{V}_r\cos{{\gamma}_r}\cos{e_\chi}-\dot{\xi}_x+\dot{\widehat{\chi}}_f e_y  \\
\dot{e}_y &=&{V}_r\cos{{\gamma}_r}\sin{e_\chi}-\dot{\widehat{\chi}}_f \varepsilon_x\\
\dot{\varepsilon}_z &=&-\widehat{V}_f\sin{\widehat{\gamma}_f}+{V}_r\sin{{\gamma}_r}-\dot{\xi}_z
\end{array}\right. \label{Eq: NewKinF2}
\end{equation}
Assume $\alpha_f$ and  $\mu_f$ are able to be rapidly stabilized by its inner-loop attitude controller to their desired values $\alpha_d$ and $\mu_d$, respectively.  The following theorem holds.

\begin{theorem} \label{Thm: OuterForm}
If Assumptions \ref{Assump: Chap5_Uncer+Distur_Bound} and \ref{Assump: Chap5_SlowDisturb} hold, and $K_V$, $K_\gamma$, $K_\chi$,  $\mathcal{T}_V$, $\mathcal{T}_\gamma$, $\mathcal{T}_\chi$, $c_V$,   $c_\chi>0$, $0<K_x<2\zeta_V\omega_V$, and $ 0<K_z<2\zeta_\gamma\omega_\gamma$, the proposed outer-loop formation controller given by (\ref{Eq: Chap5_FormCntrl_Entire}), (\ref{Eq: Chap5_NonmCntrl_VChiGa}), and (\ref{Eq: Chap5_DO_GaChiV_Mat}) will stabilize the outer-loop formation error system composed of (\ref{Eq: Chap5_ErrDyn_VGammaChi}) and (\ref{Eq: NewKinF2}), and \vspace{-3mm}
\begin{equation*}
\lim_{t\to\infty} e_x \to \xi_x\text{, }\quad \lim_{t\to\infty} e_z \to \xi_z \text{, }\quad \text{ and } \quad \lim_{t\to\infty} e_y \to 0  \vspace{-3mm}
\end{equation*}
where $\xi_x$ and $\xi_z$ will be ultimately arbitrarily bounded by control parameters $\omega_V$, $K_x$, and $\omega_\gamma$, $K_z$, respectively. If $\xi_x\left(0\right)=\xi_z\left(0\right)=0$, there exist \vspace{-3mm}
 \begin{equation}
 \lim_{t\to\infty}{\xi}_x \leq\frac{1}{K_x}\Vert V_c-V_d\Vert_\infty \text{ }\quad\text{and }\quad
 \lim_{t\to\infty} {\xi}_x \leq \frac{1}{K_z}\Vert V_f\Vert_\infty\Vert \sin{\gamma_d}-\sin{\gamma_f}\Vert_\infty \vspace{-3mm}
 \end{equation} 
 where $V_c-V_d$ will exponentially converge to $\mathcal{O}\left(\frac{1}{\omega_V}\right)$ and $\lim_{t\to\infty}\gamma_f-\gamma_d\to\mathcal{O}\left(\frac{1}{\omega_\gamma}\right)$. 
 Therefore, by tuning  $\omega_V$, $K_x$, and $\omega_\gamma$, $K_z$, the ultimate boundaries of ${e}_x$ and ${e}_z$ could be regulated accordingly.
\end{theorem}
\begin{proof}
Let $\widetilde{d}_i=\widehat{d}_i-d_i$ ($i\in\left\{V\text{, }\gamma\text{, }\chi\right\}$). If Assumption \ref{Assump: Chap5_SlowDisturb} holds, and $\alpha_f$ and $\mu_f$ are rapidly stabilized to be $\alpha_d$ and $\mu_d$, respectively, it is easily to obtain \vspace{-3mm}
\begin{equation}
\dot{\widetilde{{d}}_V}=-\frac{\widetilde{d}_V}{\mathcal{T}_{V}}\text{,}\quad\dot{\widetilde{{d}}_\gamma}=-\frac{\widetilde{d}_\gamma}{\mathcal{T}_{\gamma}}\text{,}\quad\text{and}\quad\dot{\widetilde{{d}}_V}=-\frac{\widetilde{d}_\chi}{\mathcal{T}_{\chi}} \label{Eq: Chap5_VGaChi_EstErr} \vspace{-3mm}
\end{equation}
Based on (\ref{Eq: Chap5_ErrDyn_VGammaChi}), (\ref{Eq: NewKinF2}), and (\ref{Eq: Chap5_VGaChi_EstErr}), Theorem \ref{Thm: OuterForm} is proven in two steps. The first step shows $\lim_{t\to\infty}\varepsilon_x\to0$, $\lim_{t\to\infty}e_y\to0$, and $\lim_{t\to\infty}\varepsilon_z\to0$, and thus $\lim_{t\to\infty}e_x\to\xi_x$ and $\lim_{t\to\infty}e_z\to\xi_z$. The second step demonstrates that $\xi_x$ and $\xi_z$ are uniformly ultimately bounded, which implies $e_x$ and $e_z$ are ultimately bounded. The ultimate boundaries of $e_x$ and $e_z$ are related to control gains.

To show $\lim_{t\to\infty}\varepsilon_x\to0$, $\lim_{t\to\infty}e_y\to0$, and $\lim_{t\to\infty}\varepsilon_z\to0$, we choose
\begin{equation*}
\mathbb{V}=H+\frac{\varepsilon_z^2}{2}+8c_\chi^{-1}\sin^2{\left(\frac{e_\chi}{4}\right)}+c_V^{-1}\frac{e_V^2}{2}+\frac{e_\gamma^2}{2}+c_\chi^{-1}\frac{\widetilde{d}_\chi^2}{2}+c_V^{-1}\frac{{\widetilde{d}_V}^2}{2}+\frac{\widetilde{d}_\gamma^2}{2}-1
\end{equation*}
where  $H=\sqrt{\varepsilon_x^2+e_y^2+1}$.  
Differentiating $\mathbb{V}$  yields
\begin{eqnarray}
\dot{\mathbb{V}}_1&=& \frac{\varepsilon_x\dot{\varepsilon}_x}{H}+ \frac{e_y\dot{e}_y}{H}+\varepsilon_z\dot{\varepsilon}_z+2c_\chi^{-1}\sin{\left(\frac{e_\chi}{2}\right)}
					\dot{e}_\chi+c_V^{-1}{e_V}{\dot{e}_V}+{e_\gamma}\dot{e}_\gamma+c_\chi^{-1}\widetilde{d}_\chi\dot{\widetilde{d}_\chi}+\widetilde{d}_\gamma\dot{\widetilde{d}_\gamma}+c_V^{-1}\widetilde{d}_V\dot{\widetilde{d}_V}\nonumber \\
				&=&-K_x\frac{\varepsilon_x^2}{H}+c_V^{-1}{e_V}\left({u_V}-\dot{V}_c+d_{V}+\frac{{c_V}{\varepsilon_x\cos{\widehat{\gamma}_f}}}{H}\right)-
					K_z\varepsilon_z^2+ {e_\gamma}\left({u_\gamma}-\dot{\gamma}_c+d_{\gamma}\right)
					  \nonumber \\
				& &+2c_\chi^{-1}\sin{\left(\frac{e_\chi}{2}\right)}\left({u_\chi}-\dot{\widehat{\chi}}_r+d_\chi+\frac{{c_\chi}e_y{V}_r\cos{{\gamma}_r}\cos{\frac{e_\chi}{2}}}{H}
					\right)-c_\chi^{-1}\frac{\widetilde{d}_\chi^2}{\mathcal{T}_\chi}
					-c_V^{-1}\frac{{\widetilde{d}_V}^2}{\mathcal{T}_V}-\frac{\widetilde{d}_\gamma^2}{\mathcal{T}_\gamma} \label{Eq: Lyap1_Der}
\end{eqnarray}
Substituting (\ref{Eq: Chap5_FormCntrl_Entire}) and (\ref{Eq: Chap5_NonmCntrl_VChiGa}) for corresponding terms in (\ref{Eq: Lyap1_Der}) yields
\begin{equation*}\small
\begin{array}{ll}
\dot{\mathbb{V}}_1 
				&=-c_\chi^{-1}\left[\sin{\left(\frac{e_\chi}{2}\right)}\text{, }\widetilde{d}_\chi\right]\left[\begin{array}{cc}2K_\chi& -1 \\-1&\frac{1}{\mathcal{T}_\chi}\end{array}\right]\left[\begin{array}{c}\sin{\left(\frac{e_\chi}{2}\right)} \\ \widetilde{d}_\chi\end{array}\right]-c_V^{-1}\left[e_V\text{, }\widetilde{d}_V\right]\left[\begin{array}{cc}K_V& -0.5 \\-0.5&\frac{1}{\mathcal{T}_V}\end{array}\right]\left[\begin{array}{c}e_V \\ \widetilde{d}_V\end{array}\right]\nonumber \\
				&-\left[e_\gamma\text{, }\widetilde{d}_\gamma\right]\left[\begin{array}{cc}K_\gamma& -0.5 \\-0.5&\frac{1}{\mathcal{T}_\gamma}\end{array}\right]\left[\begin{array}{c}e_\gamma \\ \widetilde{d}_\gamma\end{array}\right]-K_x\frac{\varepsilon_x^2}{H}-K_z\varepsilon_z^2
\end{array}\label{Eq: Chap5_Lyap1_Der}
\end{equation*}
If $K_V\text{, }K_\chi\text{, }K_\gamma\text{, }\mathcal{T}_V\text{, }\mathcal{T}_\gamma\text{, }\mathcal{T}_\chi>0$, one has
\begin{equation*}\small
\left[\begin{array}{cc}2K_\chi& -1 \\-1&\frac{1}{\mathcal{T}_\chi}\end{array}\right]>0 \text{,}\qquad\left[\begin{array}{cc}K_V& -0.5 \\-0.5&\frac{1}{\mathcal{T}_V}\end{array}\right]>0\text{,}\qquad\text{and}\qquad\left[\begin{array}{cc}K_\gamma& -0.5 \\-0.5&\frac{1}{\mathcal{T}_\gamma}\end{array}\right]>0
\end{equation*}
Hence, there exists $\dot{\mathbb{V}}_1\leq0$ by choosing $K_x>0$, $K_z>0$, $c_V>0$, and  $c_\chi>0$.
Therefore, $\mathbb{V}$ is a non-increasing positive definite function, which implies $\mathbb{V}(\infty)\leq\mathbb{V}(0)$ is a finite constant and $\lim_{t\to\infty}\int^{t}_{0}\dot{\mathbb{V}}_1(\tau)d\tau=\mathbb{V}(\infty)-\mathbb{V}(0)$ exists and is finite.  Thus, $\mathbb{V}$ is bounded and has a finite limit as $t\to \infty$,  and $\varepsilon_x$, $e_y$,  $\varepsilon_z$, $e_\chi$, $e_V$, $e_\gamma$, $\widetilde{d}_V$, $\widetilde{d}_\chi$, and $\widetilde{d}_\gamma$ are all bounded as well.  Furthermore, $\ddot{\mathbb{V}}_1$ is also bounded, as it is a function of $\varepsilon_x$, $e_y$,  $\varepsilon_z$, $e_\chi$, $e_V$, $e_\gamma$, $\widetilde{d}_V$, $\widetilde{d}_\chi$, and $\widetilde{d}_\gamma$. Hence,  $\dot{\mathbb{V}}_1$ is uniformly continuous.  According to Barb\v{a}lat's lemma \cite{Ioannou1995Book}, $\lim_{t\to\infty}\dot{\mathbb{V}}_1\to0$,  so $\lim_{t\to\infty}\varepsilon_x\to0\text{, }\varepsilon_z\to0\text{, }\lim_{t\to\infty}e_\chi\to0\text{, } \lim_{t\to\infty}e_V\to0\text{, } \lim_{t\to\infty}e_\gamma\to0$,  $\lim_{t\to\infty}\widetilde{d}_V\to0$,  $\lim_{t\to\infty}\widetilde{d}_\chi\to0$,  $\lim_{t\to\infty}\widetilde{d}_\gamma\to0$. It is obviously that  $\lim_{t\to\infty}e_x\to\xi_x$ and  $\lim_{t\to\infty}e_z\to\xi_z$. One could apply Barb\v{a}lat's lemma to $\dot{e}_{\chi}$ to show $\lim_{t\to\infty}e_y\to0$.  Differentiating $\dot{e}_{\chi}$ by time will yield
\begin{eqnarray}
\ddot{e}_{\chi} &=& -K_{\psi}\frac{\dot{e}_\chi}{2}\cos{\frac{e_\chi}{2}}-\frac{c_\chi\dot{H}e_yV_r\cos{\gamma_r}\cos{\frac{e_\chi}{2}}}{H^2}+\frac{c_\chi}{H}\dot{e}_yV_r\cos{\gamma_r}\cos{\frac{e_\chi}{2}}+\frac{c_\chi}{H}e_y\dot{V}_r\cos{\gamma_r}\cos{\frac{e_\chi}{2}} \nonumber \\
			&&-\frac{c_\chi}{H}\dot{\gamma}_r{e}_yV_r\sin{\gamma_r}\cos{\frac{e_\chi}{2}}-\frac{c_\chi}{H}\frac{\dot{e}_\chi}{2}{e}_yV_r\cos{\gamma_r}\sin{\frac{e_\chi}{2}}-\frac{\widetilde{d}_\chi}{\mathcal{T}_\chi} \nonumber
\end{eqnarray}
Obviously, $\ddot{\chi}_{e}$ is bounded according to Lemma \ref{Lem: Chap5_DOB_General} and Assumptions \ref{Assump: Chap5_BoundedSignal} and \ref{Assump: Chap5_Uncer+Distur_Bound}. Considering $\int \dot{e}_\chi dt=e_\chi$ and $\lim_{t\to\infty}e_\chi\to0$, we have $\dot{\psi}_e\to0$ as $t\to\infty$ in accordance with Barb\v{a}lat's lemma.  According to (\ref{Eq: Chap5_ErrDyn_VGammaChi}), (\ref{Eq: Chap5_FormCntrl_Entire}), and (\ref{Eq: Chap5_NonmCntrl_VChiGa}),  $e_y=-\frac{\dot{e}_\chi+K_{\chi}e_\chi+\widetilde{d}_\chi}{c_\chi V_r\cos{\gamma_r}\cos{\frac{e_\chi}{2}}}H$. Since $\lim_{t\to\infty}\dot{e}_\chi\to0$, $\lim_{t\to\infty}e_\chi\to0$, $\lim_{t\to\infty}\widetilde{d}_\chi\to0$, $c_\chi V_r\cos{\gamma_r}\cos{\frac{e_\chi}{2}}>0$, and $H>0$, it is readily to conclude that  $\lim_{t\to\infty}e_y\to0$.

The second step shows that $\xi_x$ and $\xi_z$ are uniformly ultimately bounded. From (\ref{Eq: Chap5_Desired_V_Gamma}), one has
\begin{equation}
V_d 		  = {\left(-K_xe_x+V_r\cos\gamma_r\ \right)}/{\cos{\widehat{\gamma}_f}}-\delta V= {\left(-K_x\varepsilon_x-K_x\xi_x+V_r\cos\gamma_r\cos{e_\chi}\right)}/{\cos{\widehat{\gamma}_f}}-\delta V \label{Eq: Chap5_Vd_Convert}
\end{equation}
Substituting the equation above for $V_d$ in (\ref{Eq: Chap5_VGa_CMD}) and (\ref{Eq: Chap5_Auxili_X_Z}) will generate
\begin{equation}
\left[\begin{array}{c}
\dot{\xi}_x  \\
\dot{V}_c  \\
\ddot{V}_c
\end{array}
\right]=\underbrace{\left[\begin{array}{ccc}
0 & \cos{\widehat{\gamma}_f}  & 0\\
0 &  			0  			  & 1\\
-\frac{K_x\omega_V^2}{\cos{\widehat{\gamma}_f}} & -\omega_V^2  & -2\zeta_V\omega_V
\end{array}
\right]}_{\mathbf{A}_v}\left[\begin{array}{c}
{\xi}_x  \\
{V}_c  \\
\dot{V}_c
\end{array}
\right]-\left[\begin{array}{c}
\cos{\widehat{\gamma}_f} \\
0  \\
\omega_V^2
\end{array}
\right]\left(\frac{V_r\cos{\gamma_r}\cos{e_\chi}-K_x\varepsilon_x }{\cos{\widehat{\gamma}_f}}+\delta V\right)\label{Eq: Chap5_Dyn_XixVcVcdot}
\end{equation} 
The characteristic equation of (\ref{Eq: Chap5_Dyn_XixVcVcdot}) is $s^3+2\zeta_V\omega_Vs^2+\omega_V^2s+K_x\omega_V^2=0$. According to the Routh--Hurwitz stability criterion, $\mathbf{A}_v$ is ensured to be Hurwitz all the time, if $0< K_x<2\zeta_V\omega_V$. Therefore, the system (\ref{Eq: Chap5_Dyn_XixVcVcdot}) is input-to-state stable with respect to $\frac{V_r\cos{\gamma_r}\cos{e_\chi}-K_x\varepsilon_x }{\cos{\widehat{\gamma}_f}}+\delta V$. According to the previous analysis, $\varepsilon_x$ will asymptotically converge to $0$. By Assumptions \ref{Assump: Chap5_BoundedSignal} and \ref{Assump: Chap5_BoundW}, both $\delta V$ and ${V_r\cos{\gamma_r}\cos{\xi_e}}$ are uniformly bounded. Since $\widehat{\gamma}_f\in\left(-\frac{\pi}{2}\text{, }\frac{\pi}{2}\right)$, $\frac{V_r\cos{\gamma_r}\cos{e_\chi}-K_x\varepsilon_x }{\cos{\widehat{\gamma}_f}}+\delta V$ is uniformly bounded, which implies $\xi_x$ is uniformly bounded. 

The final boundary of $\xi_x$ is related to $\omega_V$ and $K_x$. If $\xi_x\left(0\right)=0$, there exists ${\xi}_x \leq \frac{1}{K_x}\left(1-e^{-K_xt}\right)\Vert V_c-V_d\Vert_\infty$, so $\lim_{t\to\infty}{\xi}_x \leq\frac{1}{K_x}\Vert V_c-V_d\Vert_\infty$. Therefore, by tuning $K_x$, the ultimate boundary of ${\xi}_x$ could be changed. In addition,  $\xi_x$ is bounded-input-bounded-output with respect to $V_c-V_d$, and Lemma \ref{Lem: 2ndOrderEstProperties} indicates that $V_c-V_d$ will exponentially converge to $\mathcal{O}\left(\frac{1}{\omega_V}\right)$.  The ultimate boundary of $\xi_x$ can be, therefore, altered by changing $\omega_V$ as well. Since $\lim_{t\to\infty}e_x\to\xi_x$, the ultimate boundary of $e_x$ could thus be regulated by changing $K_x$ and $\omega_V$.

According to (\ref{Eq: Chap5_Auxili_X_Z}), the dynamics of $\xi_z$ are input-to-state stable with respect to $V_f\left(\sin\gamma_d-\sin\gamma_f\right)$. Since $V_f\left(\sin\gamma_d-\sin\gamma_f\right)$ is bounded, $\xi_z$ is ultimately bounded. According to our previous analysis, one knows that $\lim_{t\to\infty}\gamma_f\to\gamma_c$, while $\gamma_c$ is from a 2nd-order command filter (\ref{Eq: Chap5_VGa_CMD}). As $\gamma_c-\gamma_d=\mathcal{O}\left(\frac{1}{\omega_\gamma}\right)$, $\lim_{t\to\infty}\gamma_f-\gamma_d\to\mathcal{O}\left(\frac{1}{\omega_\gamma}\right)$, and thus $\sin\gamma_d-\sin\gamma_f$ will eventually be limited by certain small boundaries related to $\omega_\gamma$. Therefore, $\xi_z$ is uniformly ultimately bounded. 
Choosing $\xi_z\left(0\right)=0$ implies $ {\xi}_z \leq \frac{1}{K_z}\left(1-e^{-K_zt}\right)\Vert V_f\Vert_\infty\Vert \sin{\gamma_d}-\sin{\gamma_f}\Vert_\infty$, so $\lim_{t\to\infty} {\xi}_z \leq \frac{1}{K_z}\Vert V_f\Vert_\infty\Vert \sin{\gamma_d}-\sin{\gamma_f}\Vert_\infty$.
As $\lim_{t\to\infty}e_x\to\xi_x$, $\lim_{t\to\infty}e_z\to\xi_z$, and both $\xi_x$ and $\xi_z$ are uniformly ultimately bounded,  $e_x$ and $e_z$  will be ultimately bounded, respectively.
\end{proof}

\subsection{Inner-loop attitude control}
The inner-loop dynamics of a follower aircraft in close formation is 
\begin{equation}
\left\{\begin{array}{lcl}
\dot{\mu}_f	  &=&p\frac{\cos{\alpha_f}}{\cos{\beta_f}}+r\frac{\sin{\alpha_f}}{\cos{\beta_f}}+\dot{\gamma}_f\cos{\mu_f}\tan{\beta_f}+\dot{\chi}_f\left(\sin{\gamma_f}+\sin{\mu_f}\cos{\gamma_f}\tan{\beta_f}\right)\\
\dot{\alpha}_f   &=&q-p\tan{\beta_f}\cos{\alpha_f}-r\sin{\alpha_f}\tan{\beta_f}-\dot{\gamma}_f\frac{\cos{\mu_f}}{\cos{\beta_f}}-\dot{\chi}_f\frac{\sin{\mu_f}\cos{\gamma_f}}{\cos{\beta_f}}\\
\dot{\beta}_f    &=&p\sin{\alpha_f}-r\cos{\alpha_f}-\dot{\gamma}_f\sin{\mu_f}+\dot{\chi}_f\cos{\mu_f}\cos{\gamma_f} \\
\dot{p} &=& \frac{\left(I_y-I_z\right)I_z-I_{xz}^2}{I_xI_z-I_{xz}^2}pq+\frac{\left(I_x-I_y+I_z\right)I_{xz}}{I_xI_z-I_{xz}^2}pq+\frac{I_z}{I_xI_z-I_{xz}^2}\left(\mathcal{L}+\Delta \mathcal{L}\right)+\frac{I_{xz}}{I_xI_z-I_{xz}^2}\left(\mathcal{N}+\Delta \mathcal{N}\right)\\
\dot{q} &=& \frac{I_z-I_x}{\left(I_xI_z-I_{xz}^2\right)I_y}pr-\frac{I_{xz}}{\left(I_xI_z-I_{xz}^2\right)I_y}\left(p^2-r^2\right)+\frac{1}{\left(I_xI_z-I_{xz}^2\right)I_y}\left(\mathcal{M}+\Delta \mathcal{M}\right)\\
\dot{r}  &=& \frac{\left(I_x-I_y\right)I_x+I_{xz}^2}{I_xI_z-I_{xz}^2}pq-\frac{\left(I_x-I_y+I_z\right)I_{xz}}{I_xI_z-I_{xz}^2}rq+\frac{I_{xz}}{I_xI_z-I_{xz}^2}\left(\mathcal{L}+\Delta \mathcal{L}\right)+\frac{I_{x}}{I_xI_z-I_{xz}^2}\left(\mathcal{N}+\Delta \mathcal{N}\right)
\end{array}
\right. \label{Eq: Innerloop-Dyn}
\end{equation} 
where $\mu_f$ is the bank angle, $\alpha_f$ is the angle of attack, $\beta_f$ is the sideslip angle, $p$, $q$, and $r$ are angular rates in the body frame, $\mathcal{L}$, $\mathcal{M}$, and $\mathcal{N}$ are moments, while $\Delta \mathcal{L}$, $\Delta \mathcal{M}$, and $\Delta \mathcal{N}$ are the moment disturbances induced by trailing vortices. The presented attitude controller will rapidly stabilize $\alpha_f$ and $\mu_f$ to their desired values $\alpha_d$ and $\mu_d$, respectively. Meanwhile, the sideslip angle $\beta_f$ will be kept to be $0$. Regarding the first objective, a command filter is employed again to estimate the derivatives of $\alpha_d$ and $\mu_d$, respectively.
\begin{equation}
\left[
\begin{array}{c}
\dot{\mathscr{S}}_c \\
\ddot{\mathscr{S}}_c
\end{array}\right]=
\left[\begin{array}{cc}
0\text{,}& 1 \\
-\omega_\mathscr{S}^2\text{,} &-2\zeta_\mathscr{S}\omega_\mathscr{S}
\end{array}\right]\left[
\begin{array}{c}
\mathscr{S}_c \\
\dot{\mathscr{S}}_c
\end{array}\right]+\left[
\begin{array}{c}
0 \\
\omega_\mathscr{S}^2
\end{array}\right]\mathscr{S}_d \text{,  }\quad \mathscr{S}\in\left\{\alpha\text{, } \mu\right\} \label{Eq: Chap5_AoAMu_CMD}
\end{equation} 
Define $e_\alpha=\alpha_f-\alpha_d$ and $e_\mu=\mu_f-\mu_d$. Let $\mathbf{e}_\Theta=\left[e_\mu\text{, }e_\alpha\text{, }\beta_f\right]^T$, $\boldsymbol{\Theta}_d=\left[\mu_d\text{, }\alpha_d\text{, }0\right]^T$, $\boldsymbol{\Theta}_c=\left[\mu_c\text{, }\alpha_c\text{, }0\right]^T$, $\boldsymbol{\Psi}=\left[\gamma_f\text{, }\chi_f\right]^T$, and $\boldsymbol{\Omega}=\left[p\text{, }q\text{, }r\right]^T$. According to (\ref{Eq: Innerloop-Dyn}), we have 
\begin{equation}
\mathbf{\dot{e}}_\Theta = \mathbf{G}\boldsymbol{\Omega}+\mathbf{H}\boldsymbol{\dot{\Psi}}-\boldsymbol{\dot{\Theta}}_d \label{Eq: Chap5_Errdyn_Theta}
\end{equation}
where  $\boldsymbol{\dot{\Psi}}$ will be estimated by $\boldsymbol{\widehat{\dot{\Psi}}}=\left[u_\gamma+\widehat{d}_\gamma\text{, }u_\chi+\widehat{d}_\chi \right]^T$, and
\begin{equation*}\footnotesize
\mathbf{G}=\left[\begin{array}{ccc}
\cos{\alpha_f}\sec{\beta_f}  & 0 & \sin{\alpha_f}\sec{\beta_f}\\
-\cos{\alpha_f}\tan{\beta_f} & 1 & -\sin{\alpha_f}\tan{\beta_f} \\
\sin{\alpha_f}			  & 0 & -\cos{\alpha_f}
\end{array}\right]\text{,}\quad
\mathbf{H}=\left[\begin{array}{cc}
 \cos{\mu_f}\tan{\beta_f}   &  \sin{\gamma_f}+\sin{\mu_f}\cos{\gamma_f}\tan{\beta_f} \\
-\cos{\mu_f}\sec{\beta_f}  & -\sin{\mu_f}\cos{\gamma_f}\sec{\beta_f}\\
-\sin{\mu_f}			&    \cos{\mu_f}\cos{\gamma_f}
\end{array}\right] 
\end{equation*}
When Assumptions \ref{Assump: Chap5_Uncer+Distur_Bound} and \ref{Assump: Chap5_SlowDisturb} hold, $\lim_{t\to\infty}\left(u_\gamma+\widehat{d}_\gamma\right)\to\dot{\gamma}$,
$\lim_{t\to\infty}\left(u_\chi+\widehat{d}_\chi \right)\to\dot{\chi} $ as $\lim_{t\to\infty}\widehat{d}_\gamma\to{d}_\gamma$ and $\lim_{t\to\infty}\widehat{d}_\chi\to{d}_\chi$. 

Since Assumption \ref{Assump: Chap5_SlowDisturb} is barely met, $\mathbf{H}\left(\boldsymbol{{\dot{\Psi}}}-\boldsymbol{\widehat{\dot{\Psi}}}\right)$ is treated as model uncertainties. Additionally, $\boldsymbol{\dot{\Theta}}_d$ is replaced by $\boldsymbol{\dot{\Theta}}_c$ in terms of (\ref{Eq: Chap5_AoAMu_CMD}), as it is unavailable. Eventually, the model uncertainties of (\ref{Eq: Chap5_Errdyn_Theta}) are $\mathbf{d}_\Theta = \mathbf{H}\left(\boldsymbol{{\dot{\Psi}}}-\boldsymbol{\widehat{\dot{\Psi}}}\right)+\left(\boldsymbol{\dot{\Theta}}_c-\boldsymbol{\dot{\Theta}}_d\right)$. Define  $\mathbf{u}_\Theta= \mathbf{G}\boldsymbol{\Omega}+\mathbf{H}-\boldsymbol{\widehat{\dot{\Psi}}}$, so\vspace{-3mm}
\begin{equation}
\mathbf{\dot{e}}_\Theta =\mathbf{u}_\Theta+\mathbf{d}_\Theta-\boldsymbol{\dot{\Theta}}_c \label{Eq: Chap5_Errdyn_Theta2}\vspace{-4mm}
\end{equation}
The desired intermediate virtual inputs to stabilize (\ref{Eq: Chap5_Errdyn_Theta2}) are proposed in (\ref{Eq: Chap5_Desired_u_Theta}).\vspace{-3mm}
\begin{equation}
\mathbf{u}_\Theta^d=-\mathbf{K}_\Theta\mathbf{e}_\Theta+\boldsymbol{\dot{\Theta}}_c-\mathbf{\widehat{d}}_\Theta  \label{Eq: Chap5_Desired_u_Theta}\vspace{-4mm}
\end{equation}
where $\mathbf{K}_\Theta=diag\left\{K_\mu\text{, }K_\alpha\text{, }K_\beta\right\}>0$ is a gain matrix and $\mathbf{\widehat{d}}_\Theta$ is the estimation of $\mathbf{d}_\Theta$, which is from \vspace{-10mm}
\begin{equation}
\mathbf{\widehat{d}}_\Theta		    =   \boldsymbol{\lambda}_\Theta+\boldsymbol{\mathcal{T}}_\Theta^{-1}\mathbf{{e}}_\Theta \text{, }\qquad
\boldsymbol{\dot{\lambda}}_\Theta    =  -\boldsymbol{\mathcal{T}}_\Theta^{-1} \boldsymbol{\lambda}_\Theta-\boldsymbol{\mathcal{T}}_\Theta^{-1}\left(\boldsymbol{\mathcal{T}}_\Theta^{-1}\mathbf{{e}}_\Theta+\mathbf{u}_\Theta-\boldsymbol{\dot{\Theta}}_c\right)
\label{Eq: Chap5_D-Theta_Est} \vspace{-3mm}
\end{equation}
where$\boldsymbol{\mathcal{T}}_\Theta=diag\left\{\mathcal{T}_\mu\text{, }\mathcal{T}_\alpha\text{, }\mathcal{T}_\beta\right\}>0$ is a diagonal time constant matrix, and $\boldsymbol{\lambda}_\Theta=\left[\lambda_\mu\text{, }\lambda_\alpha\text{, }\lambda_\beta\right]^T$. 
Since $\mathbf{u}_\Theta= \mathbf{G}\boldsymbol{\Omega}+\mathbf{H}-\boldsymbol{\widehat{\dot{\Psi}}}$, the desired angular rates are given by \vspace{-3mm}
\begin{equation}
\boldsymbol{\Omega}_d= \mathbf{G}^{-1}\left(\mathbf{u}_\Theta^d-\mathbf{H}-\boldsymbol{\widehat{\dot{\Psi}}}\right) \label{Eq: Chap5_CntrlTransform} \vspace{-3mm}
\end{equation}
where $\boldsymbol{\Omega}_d=\left[p_d\text{, }q_d\text{, }r_d\right]^T$. The commanded angular rates $p_c$, $q_c$, and $r_c$ are obtained by (\ref{Eq: Chap5_Cmd_PQR}).
\begin{equation}
\left[
\begin{array}{c}
\dot{\mathscr{S}}_c \\
\ddot{\mathscr{S}}_c
\end{array}\right]=
\left[\begin{array}{cc}
0\text{,}& 1 \\
-\omega_\mathscr{S}^2\text{,} &-2\zeta_\mathscr{S}\omega_\mathscr{S}
\end{array}\right]\left[
\begin{array}{c}
\mathscr{S}_c \\
\dot{\mathscr{S}}_c
\end{array}\right]+\left[
\begin{array}{c}
0 \\
\omega_\mathscr{S}^2
\end{array}\right]\mathscr{S}_d \text{,  } \qquad \mathscr{S}\in\left\{p\text{, } q\text{, } r\right\} \label{Eq: Chap5_Cmd_PQR}
\end{equation} 
where $\mathscr{S}_c\left(0\right)=\mathscr{S}_d\left(0\right)$ for $\mathscr{S}\in\left\{p\text{, } q\text{, } r\right\}$.  Define $\mathbf{e}_\Omega=\boldsymbol{\Omega}-\boldsymbol{\Omega}_c$. According to (\ref{Eq: Innerloop-Dyn}), one has\vspace{-3mm}
\begin{equation}
\mathbf{\dot{e}}_\Omega=-\mathbf{I}^{-1}\boldsymbol{\Omega}\times\mathbf{I}\boldsymbol{\Omega}+\mathbf{I}^{-1}\left(\boldsymbol{\tau}+\Delta{\boldsymbol{\tau}}\right)-\boldsymbol{\dot{\Omega}}_c \label{Eq: Chap5_Omega_ErrDyn}\vspace{-3mm}
\end{equation}
where $\boldsymbol{{\tau}}=\left[\mathcal{L}\text{, }\mathcal{M}\text{, }\mathcal{N}\right]^T$, $\boldsymbol{\Delta\tau}=\left[\Delta\mathcal{L}\text{, }\Delta\mathcal{M}\text{, }\Delta\mathcal{N}\right]^T$, and $\mathbf{I}$ is the inertia matrix of the aircraft.
The control inputs are control surface deflections, including aileron deflection $\delta_a$, elevator deflection $\delta_e$, and rudder deflection $\delta_r$. Let $\boldsymbol{\delta_u}=\left[\delta_a\text{, }\delta_e\text{, }\delta_r\right]^T$, and $\boldsymbol{\tau}=\boldsymbol{\tau}_0+\mathbf{M_\tau}\boldsymbol{\delta_u}$, where $\mathbf{M_\tau}$ is the control derivative matrix and $\boldsymbol{\tau}_0$ is the torque vector at $\boldsymbol{\delta_u}=0$. Both $\boldsymbol{\tau}_0$ and $\mathbf{M_\tau}$ cannot be accurately obtained, so they are approximated using available aerodynamic data from wind tunnel tests. Let  $\boldsymbol{\overline{\tau}}_0$ and $\mathbf{\overline{M}_\tau}$ be the approximate results of $\boldsymbol{\tau}_0$ and $\mathbf{M_\tau}$, respectively, so
\vspace{-4mm}
\begin{equation}
\boldsymbol{\overline{\tau}}=\boldsymbol{\overline{\tau}}_0+\mathbf{\overline{M}_\tau}\boldsymbol{\delta_u} \vspace{-4mm}
\end{equation}
Let $\mathbf{u}_\tau=-\mathbf{I}^{-1}\boldsymbol{\Omega}\times\mathbf{I}\boldsymbol{\Omega}+\mathbf{I}^{-1}\left(\boldsymbol{\overline{\tau}}_0+\mathbf{\overline{M}_\tau}\boldsymbol{\delta_u}\right)=\left[u_p\text{, }u_q\text{, }u_r\right]^T$. The error model (\ref{Eq: Chap5_Omega_ErrDyn}) is rewritten as 
\vspace{-3mm}
\begin{equation}
\mathbf{\dot{e}}_\Omega=\mathbf{u}_\tau+\mathbf{d_\tau}-\boldsymbol{\dot{\Omega}}_c\label{Eq: Chap5_Omega_ErrDyn2}\vspace{-3mm}
\end{equation}
where $\mathbf{d}_{\tau}=\mathbf{I}^{-1}\left(\boldsymbol{\tau}-\boldsymbol{\overline{\tau}}+\Delta\boldsymbol{\tau}\right)$ is the sum of model uncertainties and formation aerodynamic disturbances.  Eventually, the control law for $\mathbf{u}_\tau$ is proposed to be \vspace{-3mm}
\begin{equation}
\mathbf{u}_\tau^d=-\mathbf{K}_{\Omega}{\mathbf{e}}_\Omega-\mathbf{C}_\Omega\mathbf{G}^T\boldsymbol{\varepsilon}_\Theta-\mathbf{\widehat{d}_\tau}+\boldsymbol{\dot{\Omega}}_c  \label{Eq: Chap5_DesiredTorque} \vspace{-3mm}
\end{equation}
where  $\mathbf{K}_{\Omega}=diag\left\{K_p\text{, }K_q\text{, }K_r\right\}>0$ is a diagonal constant gain matrix, $\mathbf{C}_\Omega=diag\left\{c_p\text{, }c_q\text{, }c_r\right\}>0$ is a constant matrix, $\mathbf{\widehat{d}_\tau}$ is the estimation of $\mathbf{d_\tau}$, and $\boldsymbol{\varepsilon}_\Theta=\left[\varepsilon_\mu\text{, }\varepsilon_\alpha\text{, }\varepsilon_\beta\right]^T=\mathbf{e}_\Theta-\boldsymbol{\xi}_\Theta$ where  $\boldsymbol{\xi}_\Theta$ is from\vspace{-10mm}
\begin{equation}
\boldsymbol{\dot{\xi}}_\Theta=-\mathbf{K}_{\Theta}\boldsymbol{{\xi}}_{{\Theta}}+\mathbf{G}\left(\boldsymbol{\Omega}_c-\boldsymbol{\Omega}_d\right) \label{Eq: Chap5_Auxil_Xi_Theta}\vspace{-3mm}
\end{equation}
\begin{figure}[tbp]
\centering
\includegraphics[width=0.9\textwidth]{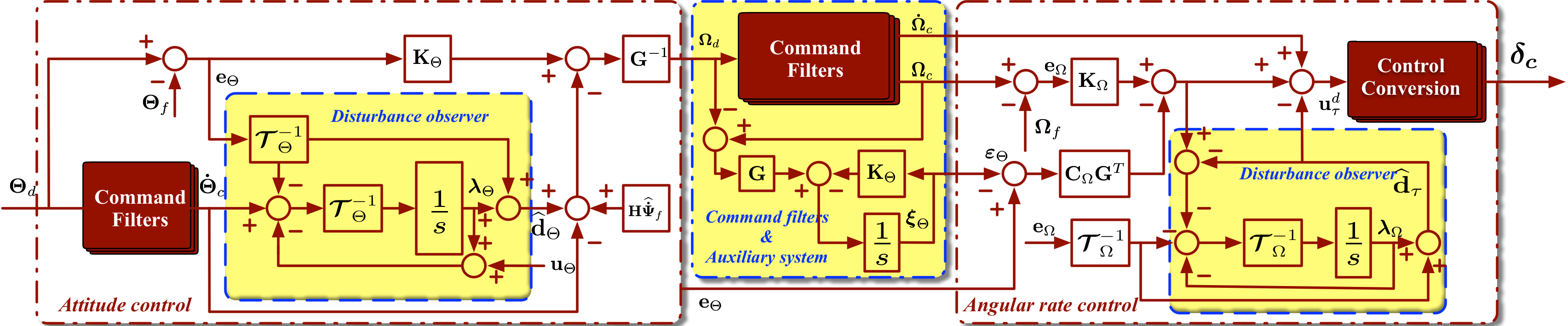}\\ \vspace{-5mm}
\caption{The inner-loop attitude control structure}\vspace{-3mm}
\label{Fig: Chap5_Innerloop}
\end{figure}
The uncertainty and disturbance estimator for  $\mathbf{d_\tau}$ is given by
\begin{equation}
\mathbf{\widehat{d}_\tau}= \boldsymbol{\lambda}_\Omega+\boldsymbol{\mathcal{T}}_\Omega^{-1}\mathbf{{e}}_\Omega \text{, }\quad
\boldsymbol{\dot{\lambda}}_\Omega=-\boldsymbol{\mathcal{T}}_\Omega^{-1}\boldsymbol{\lambda}_\Omega-\boldsymbol{\mathcal{T}}_\Omega^{-1}\left(\boldsymbol{\mathcal{T}}_\Omega^{-1}\mathbf{{e}}_\Omega-\mathbf{K}_{\Omega}{\mathbf{e}}_\Omega-\mathbf{C}_\Omega\mathbf{G}^T\boldsymbol{\varepsilon}_\Theta-\mathbf{\widehat{d}_\tau}\right) \label{Eq: Chap5_Omega_DistEst}
\end{equation}
where $\boldsymbol{\mathcal{T}}_\Omega=diag\left\{\mathcal{T}_p\text{, }\mathcal{T}_q\text{, }\mathcal{T}_r\right\}>0$ is the time constant matrix. Let $\boldsymbol{\delta_{c}}=\left[\delta_{ac}\text{, }\delta_{ec}\text{, }\delta_{rc}\right]^T$ be the commanded control surface deflection vector. Eventually,  we have \vspace{-3mm}
\begin{equation}
\boldsymbol{\delta_{c}}=\mathbf{\overline{M}_\tau}^{-1}\left(\mathbf{I}\mathbf{u}_\tau^d+\boldsymbol{\Omega}\times\mathbf{I}\boldsymbol{\Omega}-\boldsymbol{\overline{\tau}}_0\right) \vspace{-3mm}
\end{equation}
The inner-loop controller is shown in Figure \ref{Fig: Chap5_Innerloop}. The following assumption is introduced for the sake of stability analysis.
\begin{assumption} \label{Assump: Chap5_D_Theta-Tau}
Both $\mathbf{d}_\Theta$ and $\boldsymbol{d}_\tau$ are bounded with slow dynamics, namely $\mathbf{\dot{d}}_\Theta \simeq 0$ and $\mathbf{\dot{d}}_\tau \simeq 0$.
\end{assumption}
Define $\mathbf{\widetilde{d}}_\Theta=\mathbf{\widehat{d}}_\Theta-\mathbf{d}_\Theta$ and $\mathbf{\widetilde{d}}_\tau=\mathbf{\widehat{d}}_\tau-\mathbf{d}_\tau$. Under Assumption \ref{Assump: Chap5_D_Theta-Tau}, we have 
\begin{equation}
\mathbf{\dot{\widetilde{d}}}_\Theta=-\boldsymbol{\mathcal{T}}_\Theta^{-1}\mathbf{\widetilde{d}}_\Theta  \text{, }\qquad
\mathbf{\dot{\widetilde{d}}}_\tau=-\boldsymbol{\mathcal{T}}_\Omega^{-1}\mathbf{\widetilde{d}_\tau}\label{Eq: Chap5_D-Theta-Tau_EstErr} 
\end{equation}
\begin{lemma} \label{Lem: Chap5_InnerStab}
If Assumption \ref{Assump: Chap5_D_Theta-Tau} holds, and $K_\mu$, $K_\alpha$, $K_\beta$, $K_p$, $K_q$, $K_r$, $\mathcal{T}_\mu$, $\mathcal{T}_\alpha$, $\mathcal{T}_\beta$, $\mathcal{T}_p$, $\mathcal{T}_q$, $\mathcal{T}_r$, $c_p$, $c_q$, and $c_r>0$, $\boldsymbol{\varepsilon}_\Theta$ and ${\mathbf{e}}_\Omega$ will exponentially converge to zero, namely $\lim_{t\to\infty}{\mathbf{e}}_\Theta\to\boldsymbol{\xi}_\Theta$, and  $\sigma_c-\sigma_d=\mathcal{O}\left(\frac{1}{\omega_\sigma}\right)$ and ${\xi}_\sigma=\mathcal{O}\left(\frac{1}{\omega_\sigma}\right)$ with $\sigma\in\left\{p\text{, } q\text{, } r\right\}$, so ${\mathbf{e}}_\Theta$ is ultimately bounded.
\end{lemma}
In real implementations, it is only required that $\mathbf{d}_\Theta$, $\boldsymbol{d}_\tau$, and their derivatives are bounded. If $\mathbf{\dot{d}}_\Theta \simeq 0$ and $\boldsymbol{\dot{d}}_\tau \simeq 0$ fail to exist,  ${\xi}_\sigma=\mathcal{O}\left(\frac{1}{\omega_\sigma}\right)$ still holds uniformly. However, instead of achieving $\lim_{t\to\infty}\boldsymbol{\varepsilon}_\Omega\to0$, we can only ensure $\Vert{\boldsymbol{\varepsilon}}_\Theta-\boldsymbol{\bar{e}}_\Theta\Vert= \mathcal{O}\left(\epsilon_1\right)$ and $\Vert\boldsymbol{e}_\Omega-\boldsymbol{\bar{e}}_\Omega\Vert= \mathcal{O}\left(\epsilon_1\right)$ where $\epsilon_1$ is a certain positive small value related to the time constants of the disturbance observers (\ref{Eq: Chap5_D-Theta_Est}) and (\ref{Eq: Chap5_Omega_DistEst}), and $\boldsymbol{\bar{e}}_\Theta$ and $\boldsymbol{\bar{e}}_\Omega$ are the attitude tracking errors from the standard backstepping design. The dynamics of  $\boldsymbol{\bar{e}}_\Theta$ and $\boldsymbol{\bar{e}}_\Omega$  are shown in (\ref{Eq: Chap5_ClosedInnerLoopStandard}). Obviously, there exist $\lim_{t\to\infty}\boldsymbol{\bar{e}}_\Theta\to0$ and $\lim_{t\to\infty}\boldsymbol{\bar{e}}_\Omega\to0$. Since $\boldsymbol{\varepsilon}_\Theta=\boldsymbol{e}_\Theta-\boldsymbol{\xi}_\Theta$, $\boldsymbol{{e}}_\Theta$ will be ultimately bounded. This conclusion is summarized in Proposition \ref{Prop: Chap5_InnerStab_NoZeroDer}.
\begin{equation}
\boldsymbol{\dot{\bar{e}}}_\Theta = -\mathbf{K}_{\Theta}\boldsymbol{\bar{e}}_\Theta +\mathbf{G}\mathbf{\bar{e}}_\Omega\text{, }\qquad
\mathbf{\dot{\bar{e}}}_\Omega = -\mathbf{K}_{\Omega}{\mathbf{\bar{e}}}_\Omega-\mathbf{C}_\Omega\mathbf{G}^T\boldsymbol{\bar{e}}_\Theta\label{Eq: Chap5_ClosedInnerLoopStandard}
\end{equation}

\begin{proposition} \label{Prop: Chap5_InnerStab_NoZeroDer}
Assume $\mathbf{d}_\Theta$, $\boldsymbol{d}_\tau$, and their derivatives are bounded. The proposed inner-loop attitude control law composed of (\ref{Eq: Chap5_Desired_u_Theta}), (\ref{Eq: Chap5_D-Theta_Est}), (\ref{Eq: Chap5_CntrlTransform}), (\ref{Eq: Chap5_Cmd_PQR}), (\ref{Eq: Chap5_DesiredTorque}), and (\ref{Eq: Chap5_Omega_DistEst}) will make $\Vert{\boldsymbol{\varepsilon}}_\Theta-\boldsymbol{\bar{e}}_\Theta\Vert= \mathcal{O}\left(\epsilon_1\right)$ and $\Vert\boldsymbol{e}_\Omega-\boldsymbol{\bar{e}}_\Omega\Vert= \mathcal{O}\left(\epsilon_1\right)$ uniformly hold, where $\epsilon_1=\max\left\{\mathcal{T}_\mu\text{, }\mathcal{T}_\alpha\text{, }\mathcal{T}_\beta\text{, }\mathcal{T}_p\text{, }\mathcal{T}_q\text{, }\mathcal{T}_r\right\}$. Furthermore, $\boldsymbol{e}_\Theta$ will be uniformly ultimately bounded, and $\Vert{\boldsymbol{e}}_\Theta-\boldsymbol{\bar{e}}_\Theta\Vert= \mathcal{O}\left(\epsilon_2\right)$ where $\epsilon_2=\max\left\{\epsilon_1\text{, }\frac{1}{\omega_p}\text{, }\frac{1}{\omega_q}\text{, }\frac{1}{\omega_r}\right\}$.
\end{proposition}

\section{Simulation verification} \label{Sec: NumSim}
The proposed robust nonlinear close formation controller is validated based on two F-16 aircraft.  A high-fidelity model presented in \cite{Russel2003Tech} will be employed to simulate the nonlinear dynamics of a F-16 aircraft. The aerodynamic effects by close formation flight are characterized using the aerodynamic model developed in \cite{Zhang2017JA}. Note that the aerodynamic data used to build up the nonlinear aircraft model are assumed to be unavailable for the control design. Instead, a global nonlinear parameter modeling technique in \cite{Morelli1998ACC} is employed to calculate the necessary aerodynamic coefficients and $\overline{D}$, $\overline{L}_{0}$, $\overline{L}_\alpha$, $\boldsymbol{\overline{\tau}}_0$, and $\mathbf{\overline{M}_\tau}$.   The aircraft geometry and mass parameters are listed in Table \ref{Tab: AircraftGeoMassPara}, while the necessary aerodynamic parameters are given in Table \ref{Tab: AircraftAeroCoeff}. The numerical simulations are carried out at two scenarios. In the first scenario, the robustness of the disturbance observer-based controller is verified by being compared with the control without disturbance observers (DO). In the second scenario, close formation flight is conducted at different velocities with the same group of control parameters to further confirm the efficacy of the proposed design. 
\begin{table}[tbph]
  \centering
  \caption{\footnotesize Aircraft geometry and mass parameters}\footnotesize\label{Tab: AircraftGeoMassPara}
  \begin{tabular}{llll|llll}
    \toprule \toprule 
    Parameter & Symbol &Value  & Unit &  Parameter & Symbol &Value  & Unit  \\  \toprule
    Wing area &  $S$ &$27.87$  & $m^2$  & Vertical tail height &$h_t$& $3.05$& $m$ \\  
    Vertical tail area & $S_v$ & $5.09$& $m^2$ & Quarter-chord sweep angle &$\Lambda_s$& $0.57$& $rad$\\ 
    Horizontal tail area & $S_{h}$ & $10.034$& $m^2$& Dihedral angle &$\Lambda_d$&$0$& $rad$ \\ 
    Wing span  & $b$ & $9.14$ & $m$ &Gross mass & $m$ & $9295.44$ & $kg$\\ 
    Tail wing span &$b_t$& $5.49$ & $m$ & Roll moment of inertia &$I_x$&$12874.8$&$kg\cdot m^2$\\ 
    Mean aerodynamic chord &$\bar{c}$& $3.45$& $m$ & Pitch moment of inertia &$I_y$&$75673.6$&$kg\cdot m^2$ \\ 
    Root chord &$c_r$& $5.02$ & $m$ & Yaw moment of inertia &$I_z$&$85552.1$&$kg\cdot m^2$ \\ 
    Tip chord &$c_t$& $1.07$ & $m$ & Product moment of inertia &$I_{xz}$&$1331.4$&$kg\cdot m^2$ \\ \toprule 
    \end{tabular}\vspace{-5mm}%
\end{table}
\begin{table}[tbph]
  \centering
  \caption{\footnotesize Aircraft aerodynamic coefficients}\footnotesize\label{Tab: AircraftAeroCoeff}
  \begin{tabular}{lll|lll}
    \toprule 
    Coefficient	 		           & 	     Symbol            &     Value   	& Coefficient	 		           & 	      	Symbol        		&     Value \\   \hline 
    Zero-lift drag coefficient         &  	   $C_{D_0}$         &    $0.02$ 	 	& Rolling moment coefficient   &    $C_{\mathcal{L}_{\delta_{r}}}$   &  $0.02636$\\  
    Oswald efficiency number 	   &	    $e_{o}$	     &   $0.663$ 	 & Pitching moment coefficient    &    $C_{\mathcal{M}_{0}}$   &  $-0.02029$   \\ 
    Section lift curve slope      	   & 	$C_{l_{\alpha}}$   &     $5.3$ 		&  Pitch stiffness  				&    $C_{\mathcal{M}_{\alpha}}$   &  $0.0466$ \\ 
    Lift coefficient 			& 	$C_{L_{0}}$ 	    &    $0.05$ 		&  Pitch damping			&    $C_{\mathcal{M}_{q}}$   &  $-5.159$  \\ 
    Wing lift curve slope        	   & $C_{L_{\alpha}}$    &     $5.3$ 		 & Pitching moment coefficient  &    $C_{\mathcal{M}_{\delta_{e}}}$   &  $-0.60123$	\\  
    Vertical tail efficiency factor   & 		$c_{\eta}$    	&     $0.95$ 	  & Weathercock stability derivative     		& 	$C_{\mathcal{N}_{\beta}}$  	&    $0.2993$ 	\\  
    Dihedral effect       	& $C_{\mathcal{L}_{\beta}}$ &  $-0.1059$ 	&  Cross coupling derivative     			&   	$C_{\mathcal{N}_{p}}$		&    $0.02678$    \\
    Roll damping coefficient     &   $C_{\mathcal{L}_{p}}$&    $-0.4127$   & Yaw damping coefficient       			&   	$C_{\mathcal{N}_{r}}$		&     $-0.36988$ \\ 
    Vertical tail effect coefficient   &  $C_{\mathcal{L}_{r}}$&   $0.0625$ &  Yawing moment coefficient   &   $C_{\mathcal{N}_{\delta_{a}}}$	&     $-0.03349$ \\ 
    Rolling moment coefficient   &  $C_{\mathcal{L}_{\delta_{a}}}$&  $-0.1463$ &      Yawing moment coefficient   &    $C_{\mathcal{N}_{\delta_{r}}}$   &  $-0.081159$ \\ \toprule 
    \end{tabular}%
\end{table}

\begin{table}[tbph]
  \centering
  \caption{\footnotesize Command filter parameters for $\mathbf{L}_c$ and $\dot{\chi}_r$} \footnotesize\label{Tab: ComGen}
  \begin{tabular}{cc|cc|cc|cc}
    \toprule 
    Natural frequency & Value   & Damping ratio &Value &  Natural frequency & Value   & Damping ratio &Value\\  \hline
    $\omega_{l_{x}}$  & 	5 	& $\zeta_{l_{x}}$ & 1	&$\omega_{l_{y}}$  & 	5 	& $\zeta_{l_{y}}$ & 1	\\  
    $\omega_{l_{z}}$  & 	5 	& $\zeta_{l_{z}}$ & 1	&$\omega_{\chi_{r}}$  & 	5 	& $\zeta_{\chi_{r}}$ & 1	\\
    \toprule
    \end{tabular}%
\end{table}

\begin{table}[tbph]
  \centering
  \caption{\footnotesize Outerloop control  parameters} \footnotesize\label{Tab: OuterPara}
  \begin{tabular}{ccc|ccc|ccc}
    \toprule 
    Parameter 	& 	Symbol   	 &Value		  &	Parameter  	& 	Symbol 	& 	Value   	&	   Parameter 		& 		Symbol 			& 	Value\\  \hline
Time constant  & $\mathcal{T}_{Wx}$ & $0.8$&	Control gain  	&   $K_x$ 	& 	$0.3	$	& Control gain		&  		$C_y$ 			&  $10^{-4}$ \\
Time constant  & $\mathcal{T}_{Wy}$ & $0.5$ &	Control gain  	&   $K_z$ 	& 	$0.2	$ 	& Control gain		&  		$C_z$ 			&  $5\times10^{-7}$ \\
Time constant  & $\mathcal{T}_{Wz}$ & $0.4$ &   Control gain  	&   $K_V$ 	& 	$1.75$	 & Natural frequency	&  	  $\omega_V$ 			& 	$8$\\
Time constant  & $\mathcal{T}_{V}$ & $0.25$ &   Control gain  	&   $K_\gamma$ & 	$0.75$  	& Natural frequency	&   $\omega_\gamma$ 		& 	$8$\\
Time constant  & $\mathcal{T}_{\gamma}$&$0.2$ &Control gain  &  $K_\chi$ 	& 	$1.75$	 &Damping ratio	&  	  $\zeta_V$ 			& 	$1$\\
Time constant & $\mathcal{T}_{\chi}$  & $0.2$&Control gain  	&  $C_x$ 		&  $10^{-5}$	 & Damping ratio	&   $\zeta_\gamma$ 		& 	$1$\\  \toprule
    \end{tabular}%
\end{table}

\subsection{Scenario 1: With/without disturbance observers}
The initial conditions of the leader aircraft are $x_l\left(0\right)=45$ $m$, $y_l=-15$ $m$, $z_l=-5015$ $m$, $V_l=200$ $m/s$, $\gamma_l=\chi_l=0$ $deg$.
According to the aerodynamic analysis in \cite{Zhang2017JA}, the optimal relative position vector is selected to be $\left[-36\text{, } 9\text{, } 0\right]^T$ $m$. The parameters of the command filters for generating $\mathbf{L}_c$ and estimating are given in Table \ref{Tab: ComGen}. The initial conditions for the follower aircraft are  $x_f\left(0\right)=45$ $m$, $y_f\left(0\right)=-15$ $m$, $z_f\left(0\right)=-5015$ $m$, $V_f\left(0\right)=200$ $m/s$, $\gamma_f\left(0\right)=\chi_f\left(0\right)=\mu_f\left(0\right)=\beta_f\left(0\right)=0$ $deg$,  $\alpha_f\left(0\right)=2.774$ $deg$, $p=q=r=0$ $rad/s$.  The outerloop and innerloop control parameters are presented in Table \ref{Tab: OuterPara} and \ref{Tab: InnerPara}, respectively. The formation trajectory in the inertial frame is illustrated in Figures \ref{Fig: Traj3D}. From $0$ to $35$ seconds, the formation trajectory is at level and straight flight. From $35$ to $145$ seconds, aircraft in close formation are required to make a turn, and meanwhile reduce their altitudes. After $145$ seconds, level and straight trajectory is recovered. Shown in Figure \ref{Fig: RelPosPos_T} are the postures and relative positions of the leader and follower aircraft at four different time instants under the proposed robust nonlinear control. The follower aircraft is initially far away from its optimal position relative to the leader aircraft. Under the proposed robust nonlinear controller, the follower aircraft is able to quickly catch up the optimal relative position. Shown in Figure \ref{Fig: TopFront_T180} are highlights of the top and front views of the relative positions between the leader and follower aircraft at $180$ seconds. 
\begin{table}[tbph]
  \centering
  \caption{Innerloop control  parameters} \label{Tab: InnerPara}
  \begin{tabular}{ccc|ccc|ccc}
    \toprule 
    Parameter 	& 	Symbol   	  		& Value    &Parameter & 	Symbol 	& Value   			&	  Parameter 	& 	Symbol 		& 	Value \\  \hline
Time constant &$\mathcal{T}_{\mu}$ 	& $0.8$    &Control gain& $K_p$ 	&  $12 $			&Natural frequency	&  	  $\omega_q$ 			& 	$5	$	 	\\
Time constant &$\mathcal{T}_{\alpha}$ 	& $0.8$  & Control gain	& $K_q$ 	&  $7.5 $   		&Natural frequency	&  	  $\omega_r$ 			& 	$5 	$	 	\\
Time constant &$\mathcal{T}_{\beta}$ 	& $0.8$  & Control gain	& $K_r$ 	&   $7.5$   		&Damping ratio		&  	  $\zeta_\mu$ 			& 	$1 $		\\
Time constant &$\mathcal{T}_{p}$ 		&$0.02$	&  Control gain	&  $c_p$ 	&  	$10^{-5}$	  	&   Damping ratio	&         $\zeta_\alpha$ 		& 	$1$\\
Time constant &$\mathcal{T}_{q}$ 		&$0.02$	& Natural frequency&  $\omega_\alpha$& $8$&Damping ratio		&  	  $\zeta_p$ 			& 	$1$ 		\\
Time constant &$\mathcal{T}_{r}$ 		&$0.02$	&Natural frequency	&  $\omega_\mu$ 	& $8$&Damping ratio		&  	  $\zeta_q$ 			& 	$1$ 	\\
Control gain  	&$K_\mu$ 			& $ 5$	&	Control gain	&  $c_r$ 		&  	$10^{-5}$	&Damping ratio		&  	   $\zeta_r$ 			& 	$1 $\\
Control gain  	& $K_\alpha$ 			& $5	$&	 	 Control gain	&  $c_q$ 		&  	$10^{-5}$	 &&&\\
Control gain  	& $K_\beta$			& $5	$ &	  Natural frequency	&  $\omega_p$ 	& 	$25$		&&&\\ \toprule 
    \end{tabular}%
\end{table}
\begin{figure}[htbp]
\centering
\includegraphics[width=0.8\textwidth]{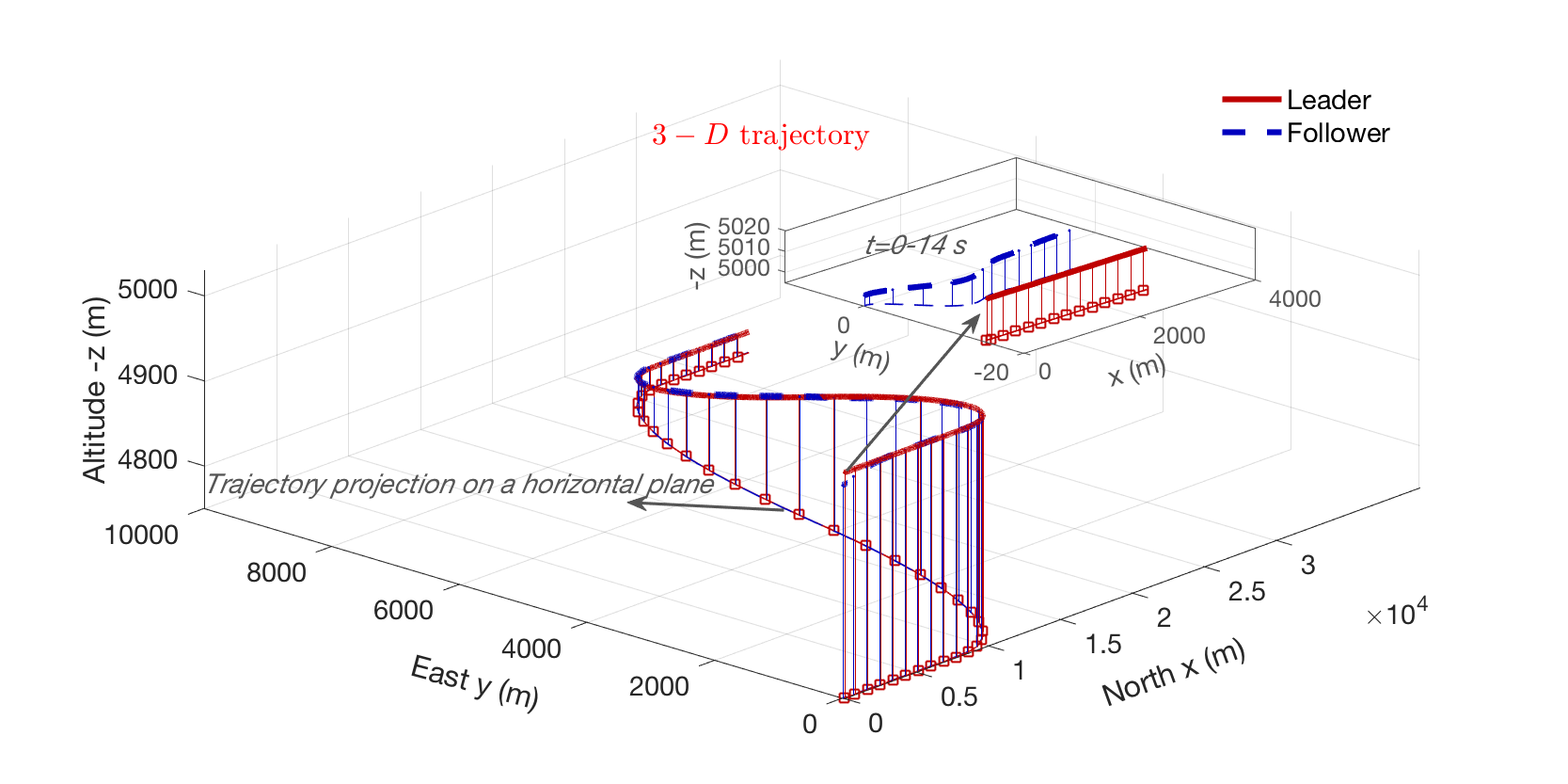}\\ \vspace{-4mm}
\caption{Formation trajectory}
\label{Fig: Traj3D}
\end{figure}
\begin{figure}[h]
\centering
\begin{tabular}{llll}
\includegraphics[width=0.465\textwidth]{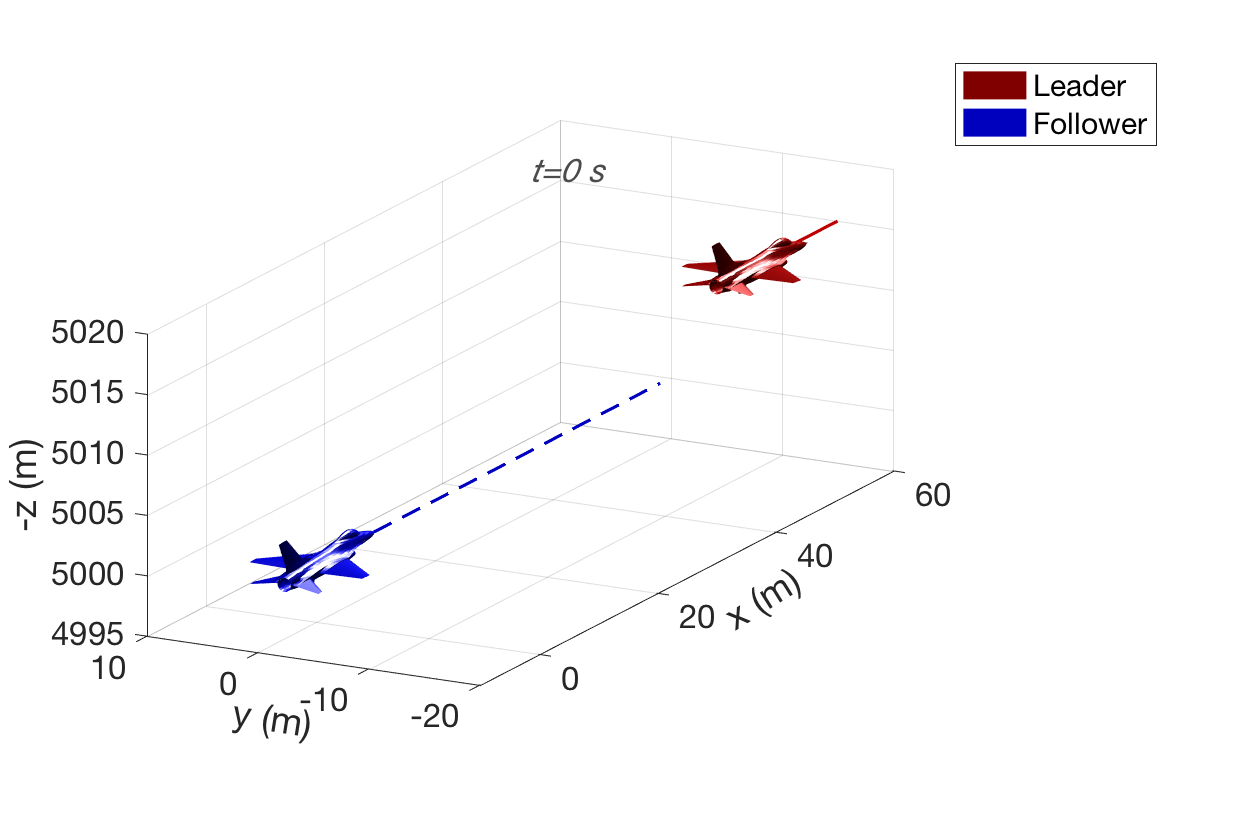}
& &
\includegraphics[width=0.465\textwidth]{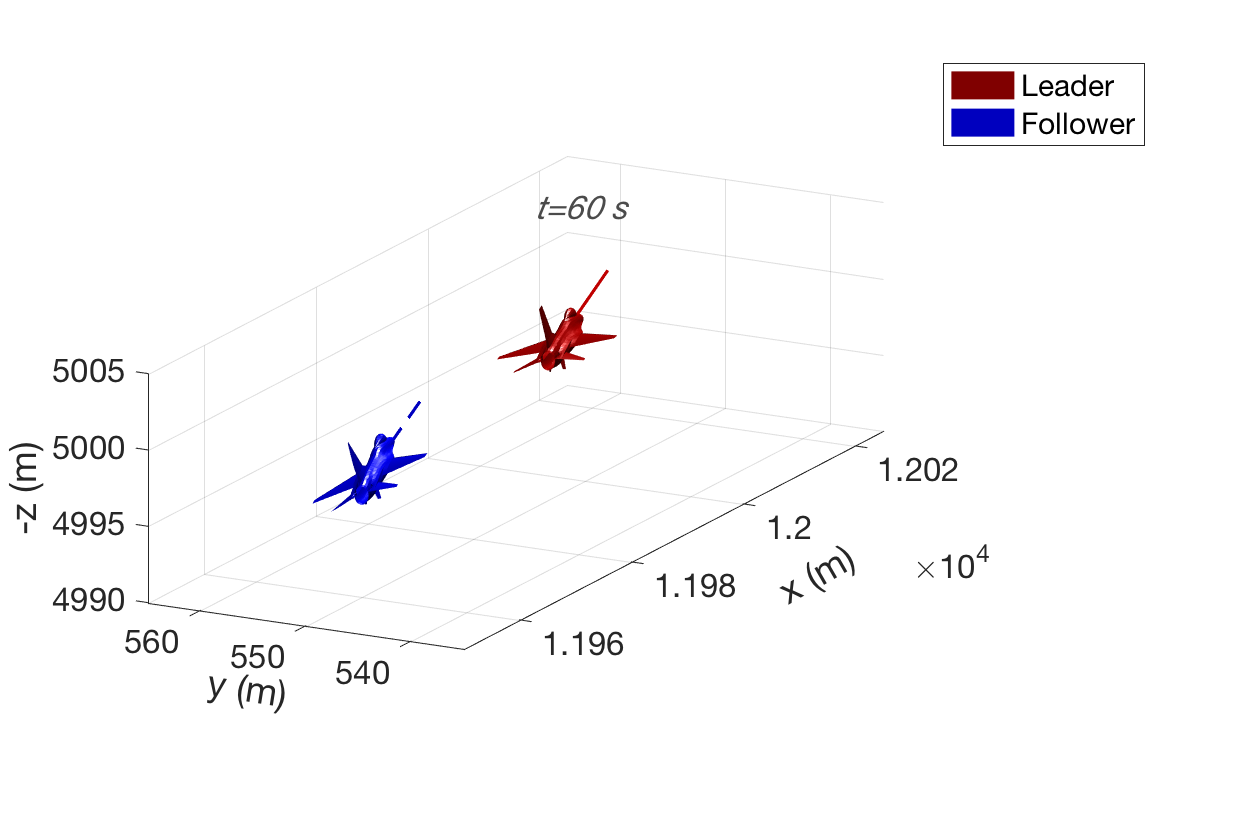} \vspace{-5mm} \\ 
\includegraphics[width=0.465\textwidth]{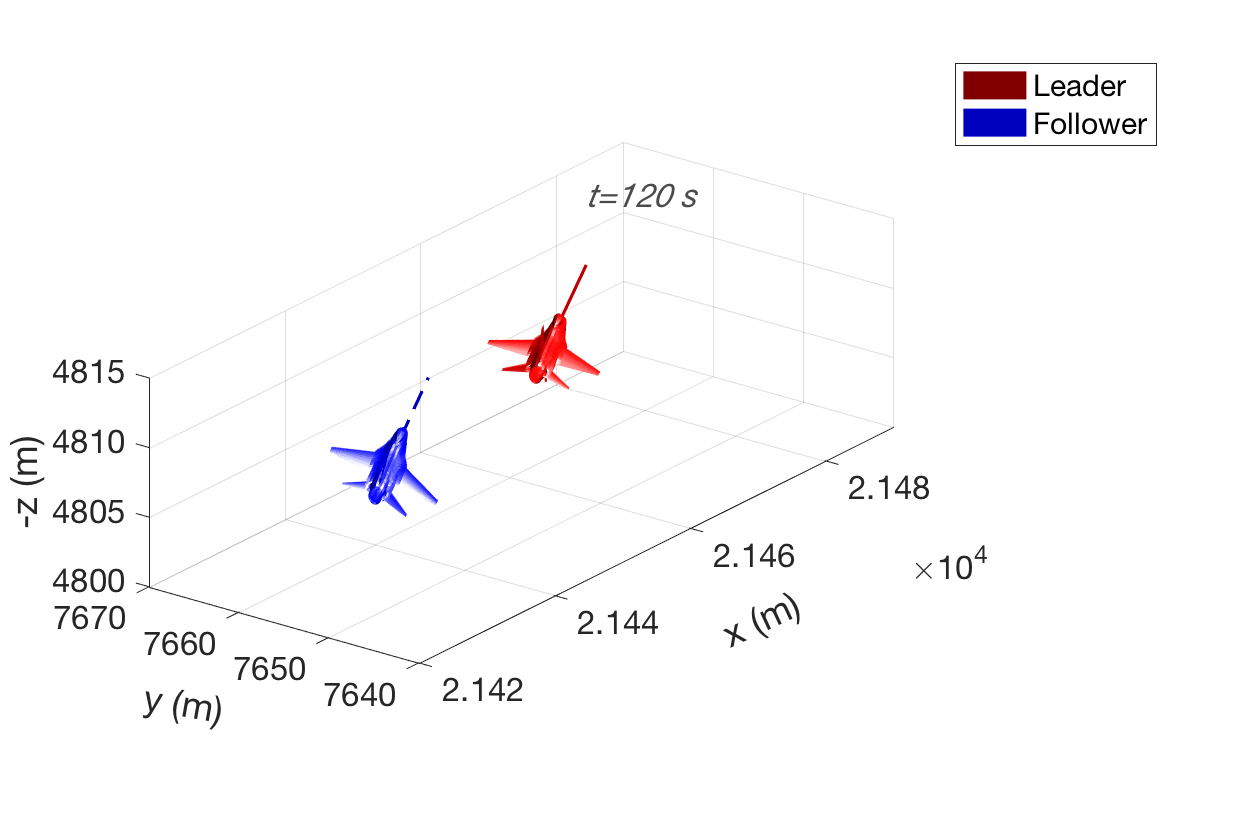}
& &
\includegraphics[width=0.465\textwidth]{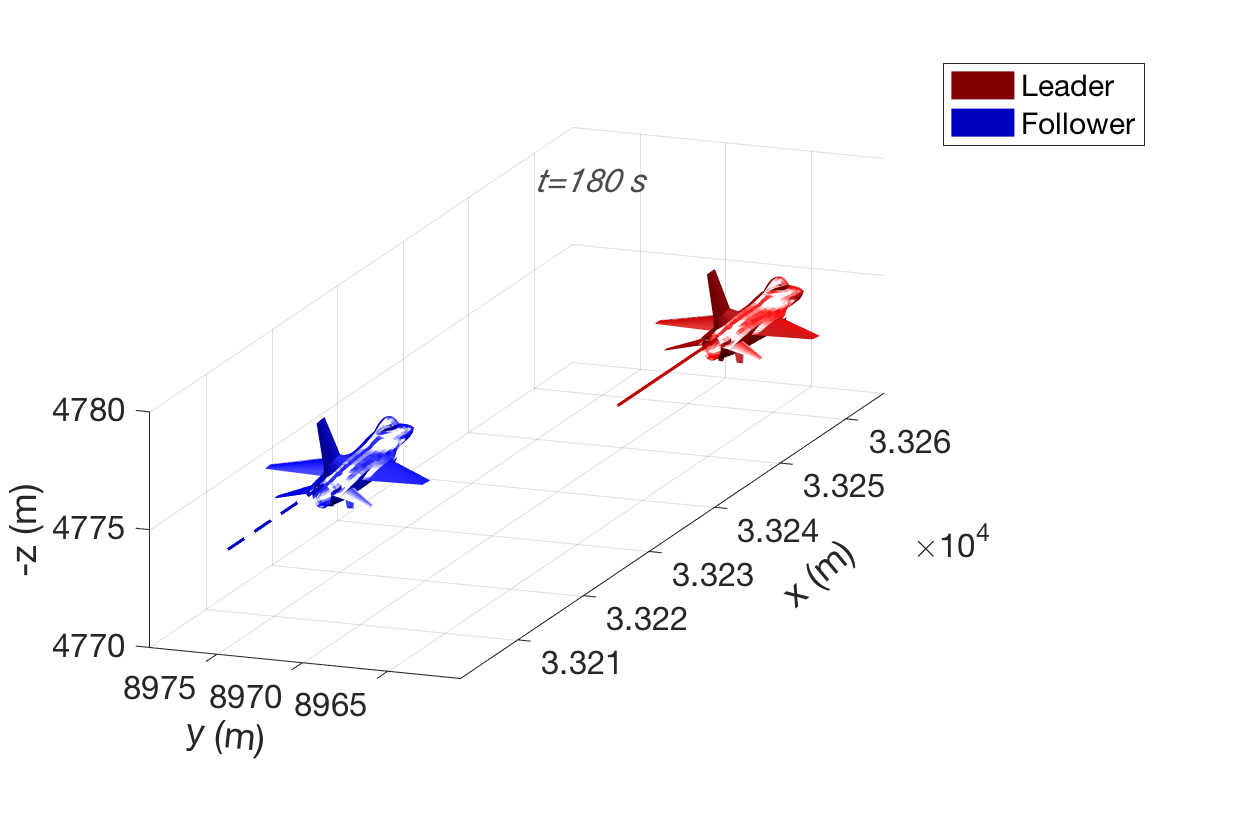}
\end{tabular}\vspace{-6mm}
\caption{Aircraft relative positions and postures at different time instants}
\label{Fig: RelPosPos_T}
\end{figure}
\begin{figure}[h]
\centering
\begin{tabular}{lll}
\includegraphics[width=0.51\textwidth]{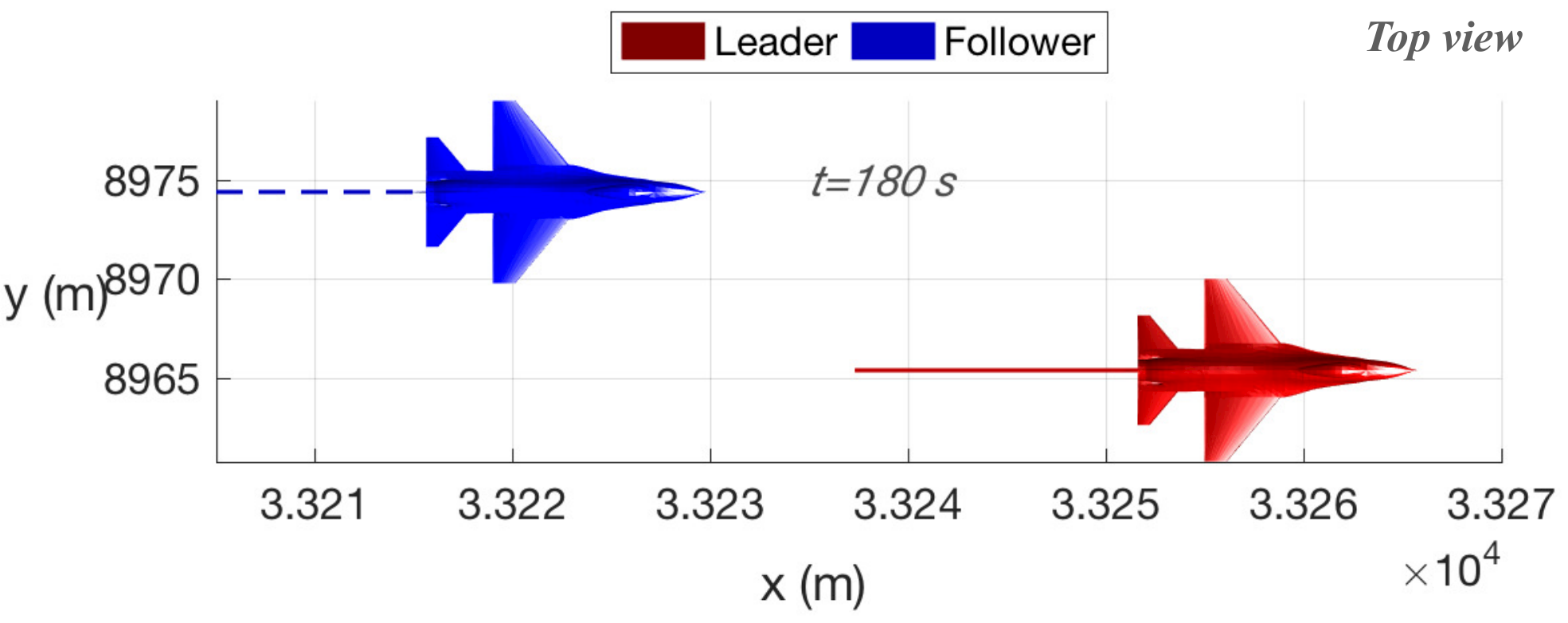}
& &
\includegraphics[width=0.46\textwidth]{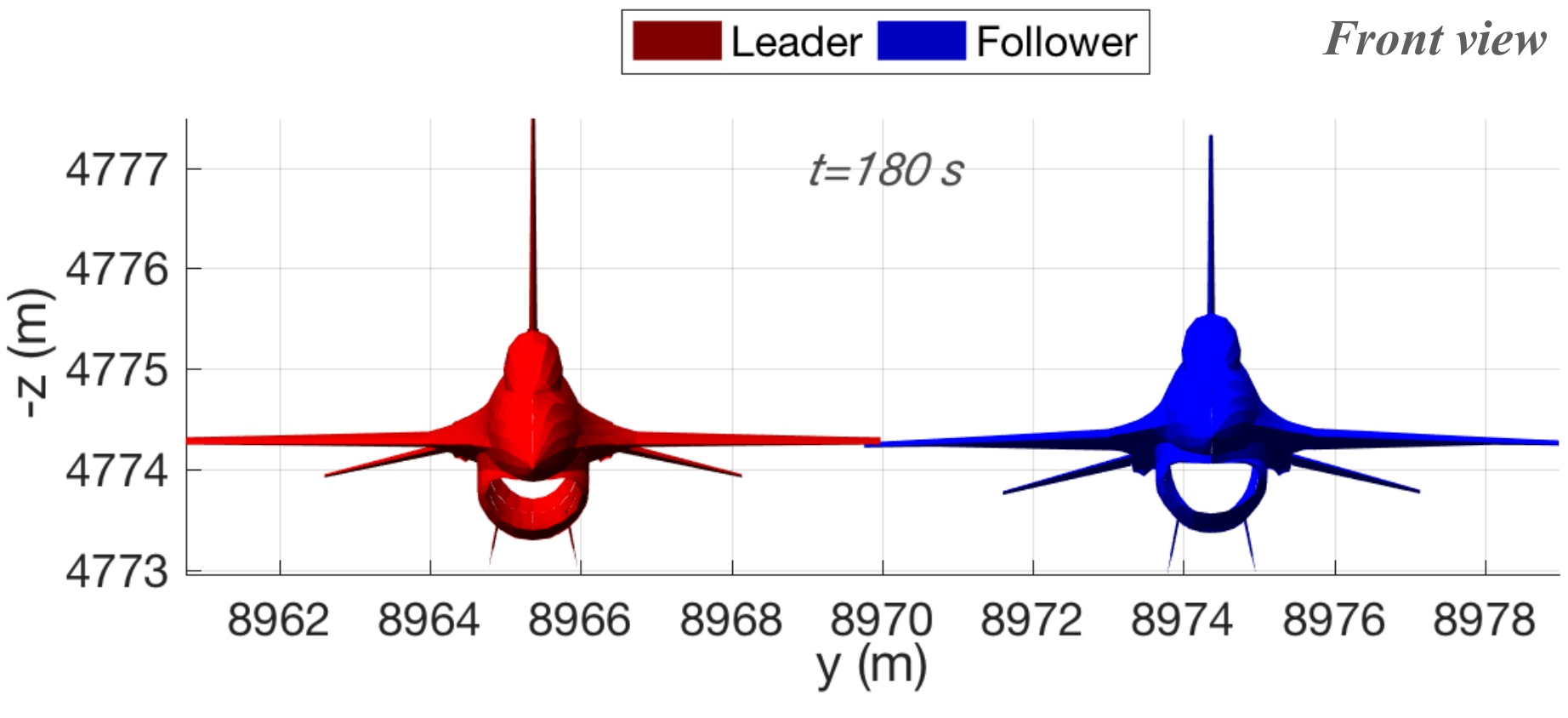}
\end{tabular}\vspace{-5mm}
\caption{Top and front view of close formation flight at $t=180$ $s$}\vspace{-3mm}
\label{Fig: TopFront_T180}
\end{figure}

Two different nonlinear controllers are implemented in the first scenario. One is the proposed robust nonlinear formation controller, the other one is purely the baseline nonlinear formation controller without including disturbance observers (DO). The position tracking errors under the two controllers are shown in Figures  \ref{Fig: ErrX}, \ref{Fig: ErrY}, and \ref{Fig: ErrZ}. According to \cite{Zhang2017JA},  close formation flight will fail if the optimal lateral and vertical relative positions cannot be tracked with in $10\%$ wing span, while $90\%$ of the maximum drag reduction could be retained if the position tracking error is kept under $5\%$ wing span. For efficient close formation flight, it is required that  both the lateral and vertical tracking errors are smaller than at least $5\%$ wing span.  To show the validness of the proposed design,  the regions covered by $5\%$ and $10\%$ wing span are highlighted in  Figures \ref{Fig: ErrY} and \ref{Fig: ErrZ}. Obviously, the nonlinear control without including disturbance observers failed to achieve reasonable close formation flight, as the vertical position tracking errors are much far away from the region of interest. After incorporating disturbance observers (DO), both the lateral and vertical position tracking errors are confined to be smaller than $5\%$ wing span, so accurate close formation flight is fulfilled. Furthermore,  position tracking errors under the proposed robust nonlinear controller would converge to zero when close formation is at level and straight flight. When the leader aircraft is taking maneuvers ($35$ to $145$ $s$), steady tracking errors are observed for lateral position tracking. The steady tracking errors are under $5\%$ wing span, which implies efficient close formation flight is still guaranteed. Shown in Figures \ref{Fig: ErrV}, \ref{Fig: ErrGamma}, and \ref{Fig: ErrChi} are tracking errors of speed, flight path angle, and heading angle, respectively.  The non-zero lateral steady tracking errors from $35$ to $145$ $s$ result from the non-zero tracking errors in heading angle control as shown in Figure \ref{Fig: ErrChi}. When leader aircraft is taking maneuvers, the heading angle will have non-zero tracking errors due to the difference between $\dot{\chi}_r$ and $\dot{\hat{\chi}}_r$. In real implementations, the derivative of the reference heading angle $\chi_r$ is always unavailable, so a second order filter is introduced to approximate $\dot{\chi}_r$. When the leader aircraft is under level and straight flight, $\chi_r$ is constant, which implies $\dot{\chi}_r\equiv0$. In this case, $\dot{\hat{\chi}}_r$ could be ensured to be equal to $\dot{\chi}_r$, so asymptotical stability is able to be obtained as shown in Figures \ref{Fig: ErrX}, \ref{Fig: ErrY},  \ref{Fig: ErrZ}, and  \ref{Fig: ErrChi}. However, when the leader aircraft is taking maneuvers, $\chi_r$ is not constant, and $\dot{\hat{\chi}}_r$ can only be guaranteed to converge to a certain value close to $\dot{\chi}_r$. The difference between $\dot{\hat{\chi}}_r$ and $\dot{\chi}_r$ might lead to the steady heading tracking errors  from $35$ to $145$ $s$ as shown in Figure \ref{Fig: ErrChi}, which is reflected in lateral position tracking control as given in Figure \ref{Fig: ErrY}. 
\begin{figure}[htbp]
\centering
\includegraphics[width=0.8\textwidth]{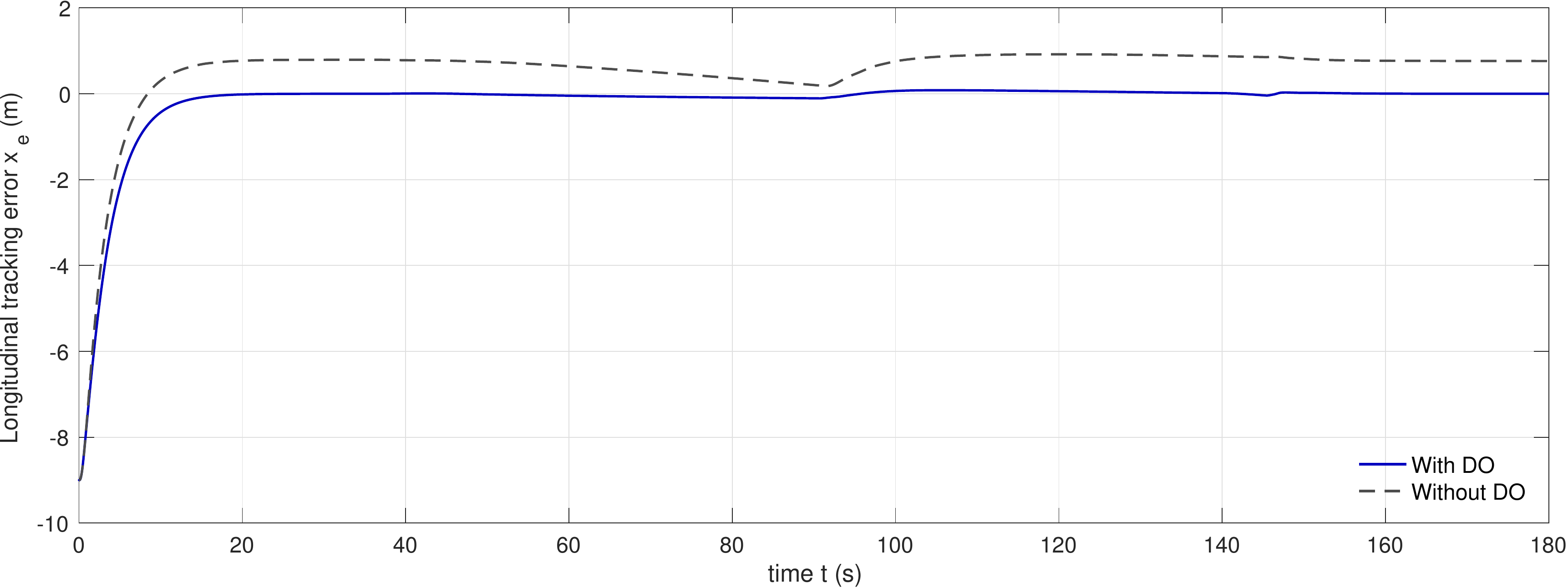} \\ \vspace{-5mm}
\caption{Longitudinal tracking errors, $x_e=x_f-x_r$ (Scenario 1)}\vspace{-0mm}
\label{Fig: ErrX}
\end{figure}
\begin{figure}[htbp]
\centering
\includegraphics[width=0.8\textwidth]{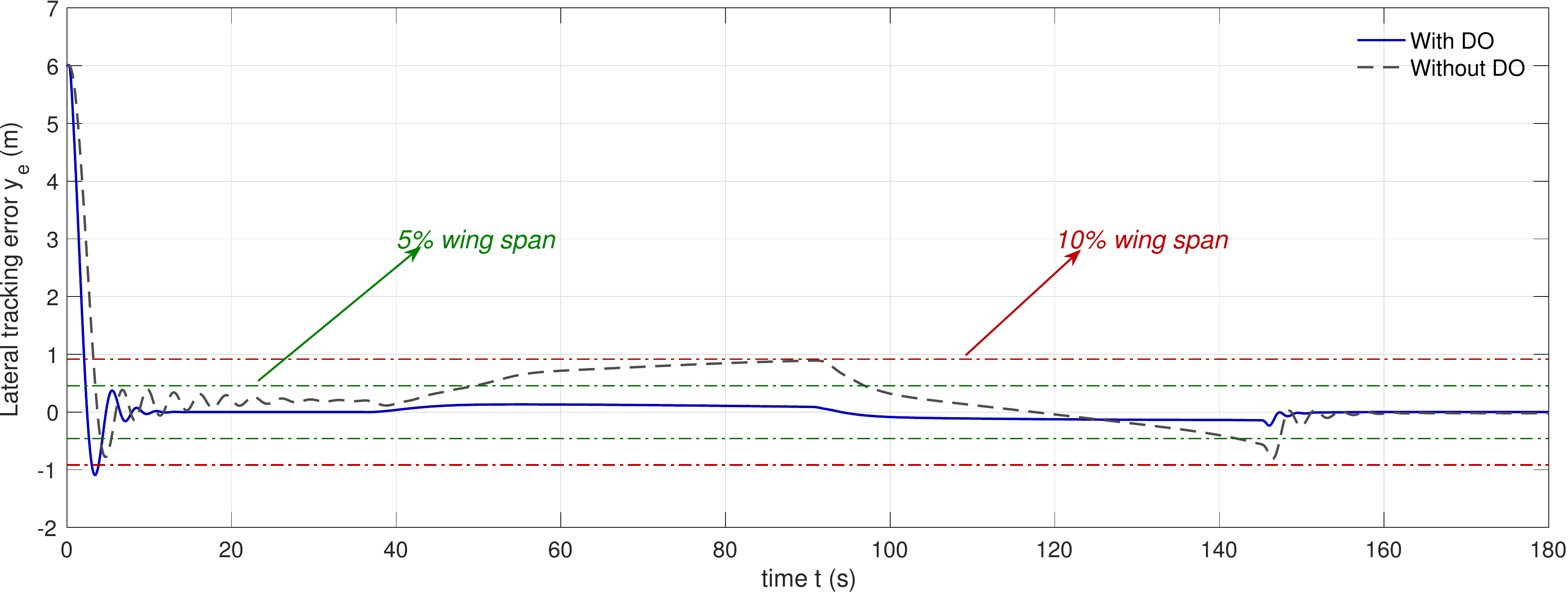}\\ \vspace{-5mm}
\caption{Lateral tracking errors, $y_e=y_f-y_r$ (Scenario 1)}\vspace{-0mm}
\label{Fig: ErrY}
\end{figure}
\begin{figure}[htbp]
\centering
\includegraphics[width=0.8\textwidth]{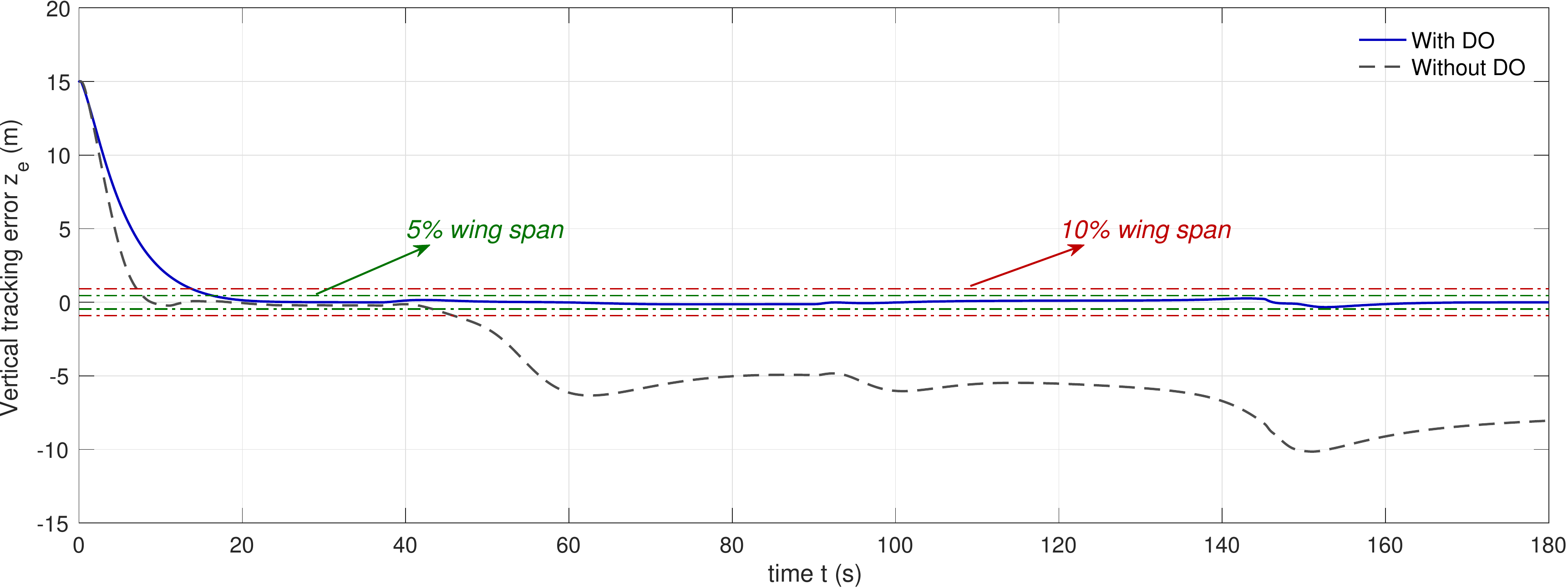}\\ \vspace{-5mm}
\caption{Vertical tracking errors, $z_e=z_f-z_r$ (Scenario 1)}\vspace{-3mm}
\label{Fig: ErrZ}
\end{figure}

The inner-loop state responses are summarized in Figure \ref{Fig: InnerLoopStateResp}. The sideslip angle of the follower aircraft is always kept to be zero by the proposed robust nonlinear formation controller.  Shown in Figure \ref{Fig: ActuatorResp} are responses of control inputs. As mentioned before, the baseline formation controller without disturbance observers cannot achieve successful close formation flight. The follower aircraft under the baseline controller is eventually one span away from its optimal relative position to the leader aircraft, in which case the influence of the trailing vortices is quite small. Therefore, the steady performance of the follower aircraft by the baseline formation controller is similar to that of an aircraft at solo flight. Compared with the baseline controller, the robust nonlinear controller will eventually have $13.876\%$ decrease in throttle inputs, which implies that around $13.876\%$ energy saving could be obtained by close formation flight.
\begin{figure}[htbp]
\centering
\includegraphics[width=0.8\textwidth]{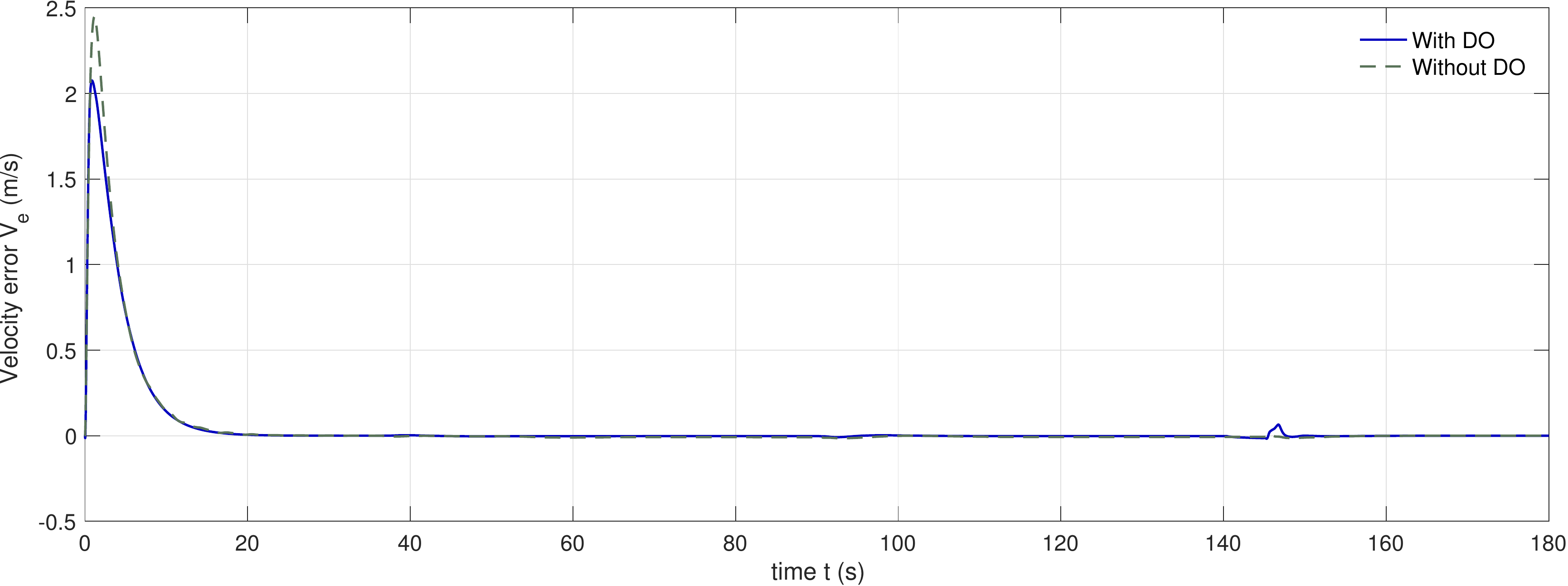}\\ \vspace{-5mm}
\caption{Speed tracking errors, $V_e=V_f-V_r$  (Scenario 1)}\vspace{-3mm}
\label{Fig: ErrV}
\end{figure}
\begin{figure}[htbp]
\centering
\includegraphics[width=0.8\textwidth]{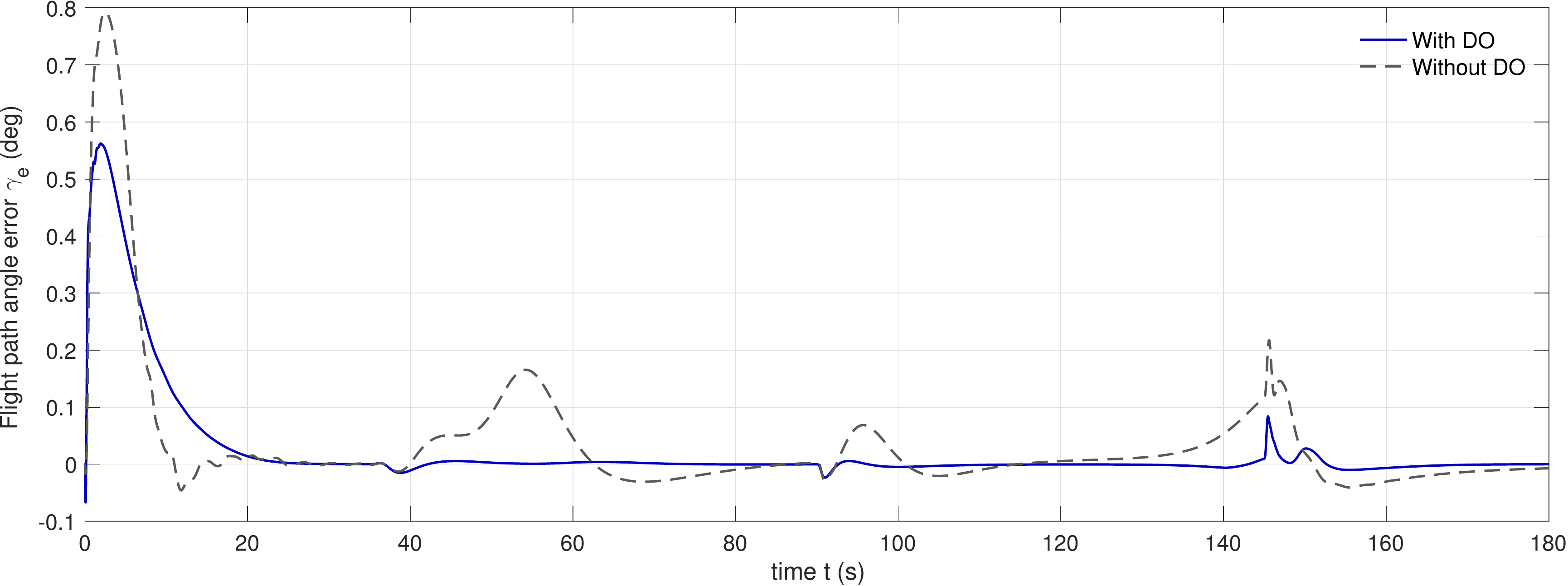}\\ \vspace{-5mm}
\caption{Flight path angle tracking errors, $\gamma_e=\gamma_f-\gamma_r$  (Scenario 1)}
\label{Fig: ErrGamma}\vspace{-1mm}
\end{figure}
\begin{figure}[htbp]
\centering
\includegraphics[width=0.8\textwidth]{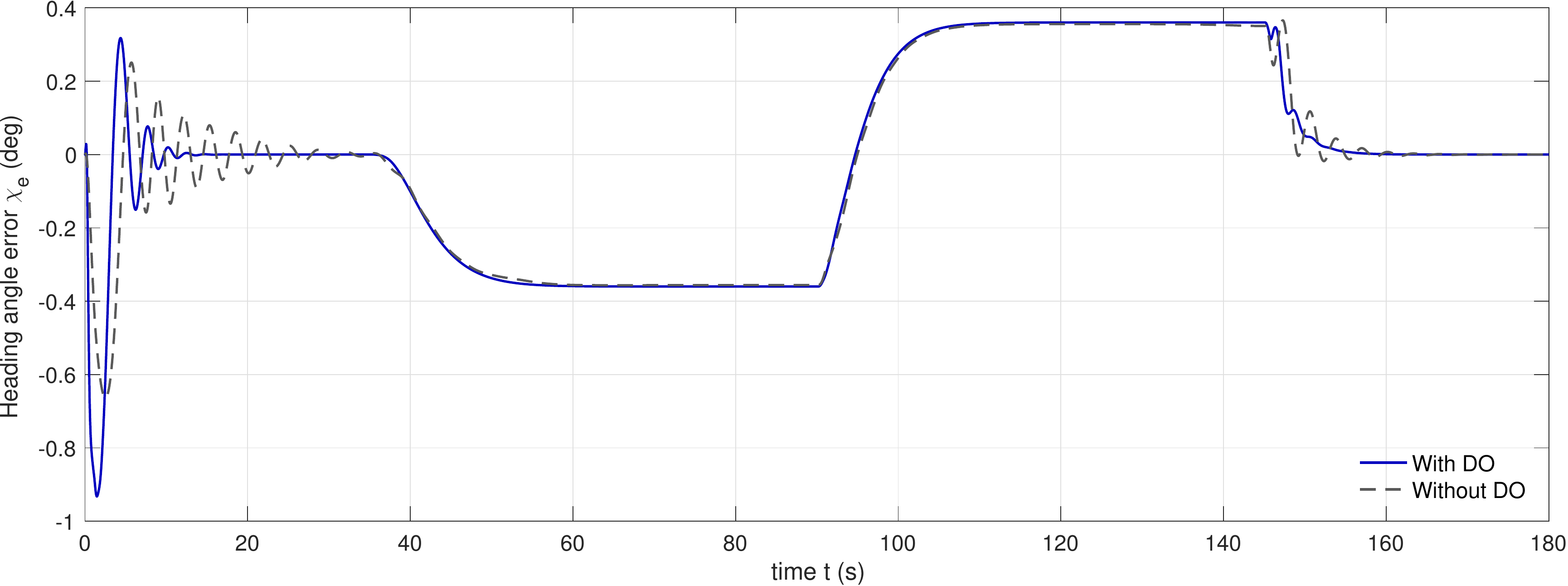}\\ \vspace{-5mm}
\caption{Heading angle tracking errors, $\chi_e=\chi_f-\chi_r$  (Scenario 1)}
\label{Fig: ErrChi}\vspace{-1mm}
\end{figure}

\begin{figure}[h]
\centering
\begin{tabular}{lll}
\includegraphics[width=0.45\textwidth]{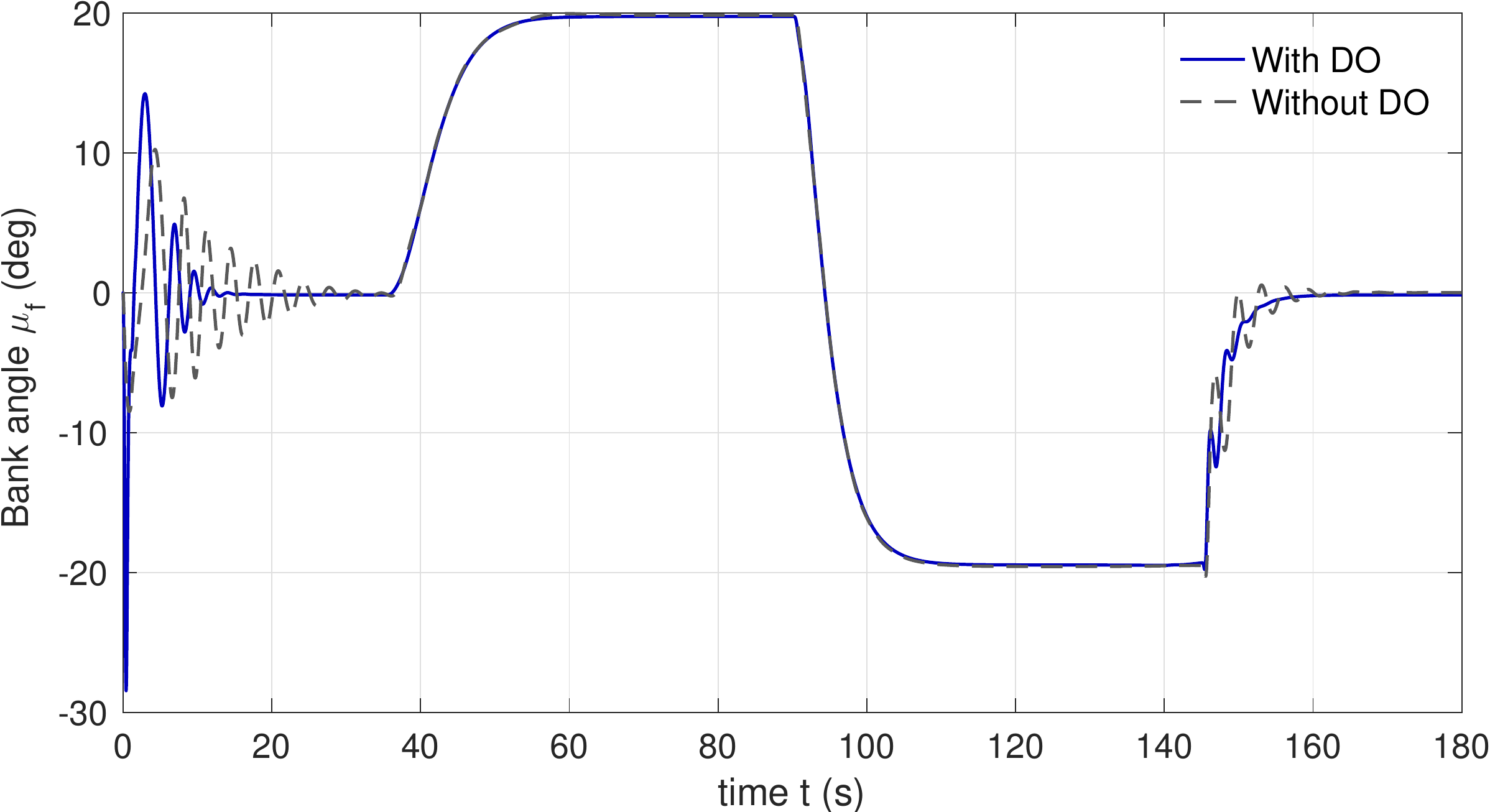}
& &
\includegraphics[width=0.45\textwidth]{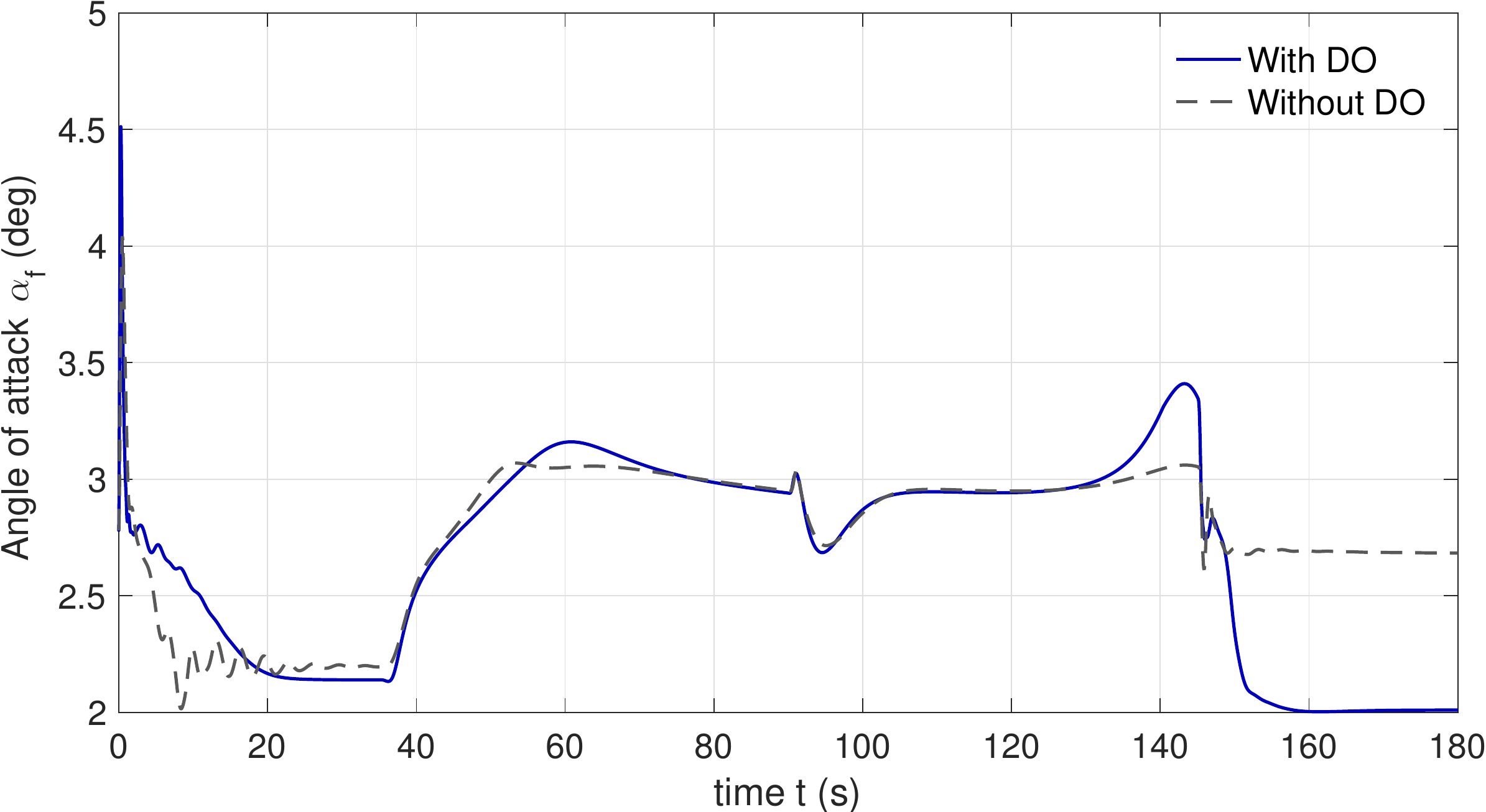} \\
\includegraphics[width=0.45\textwidth]{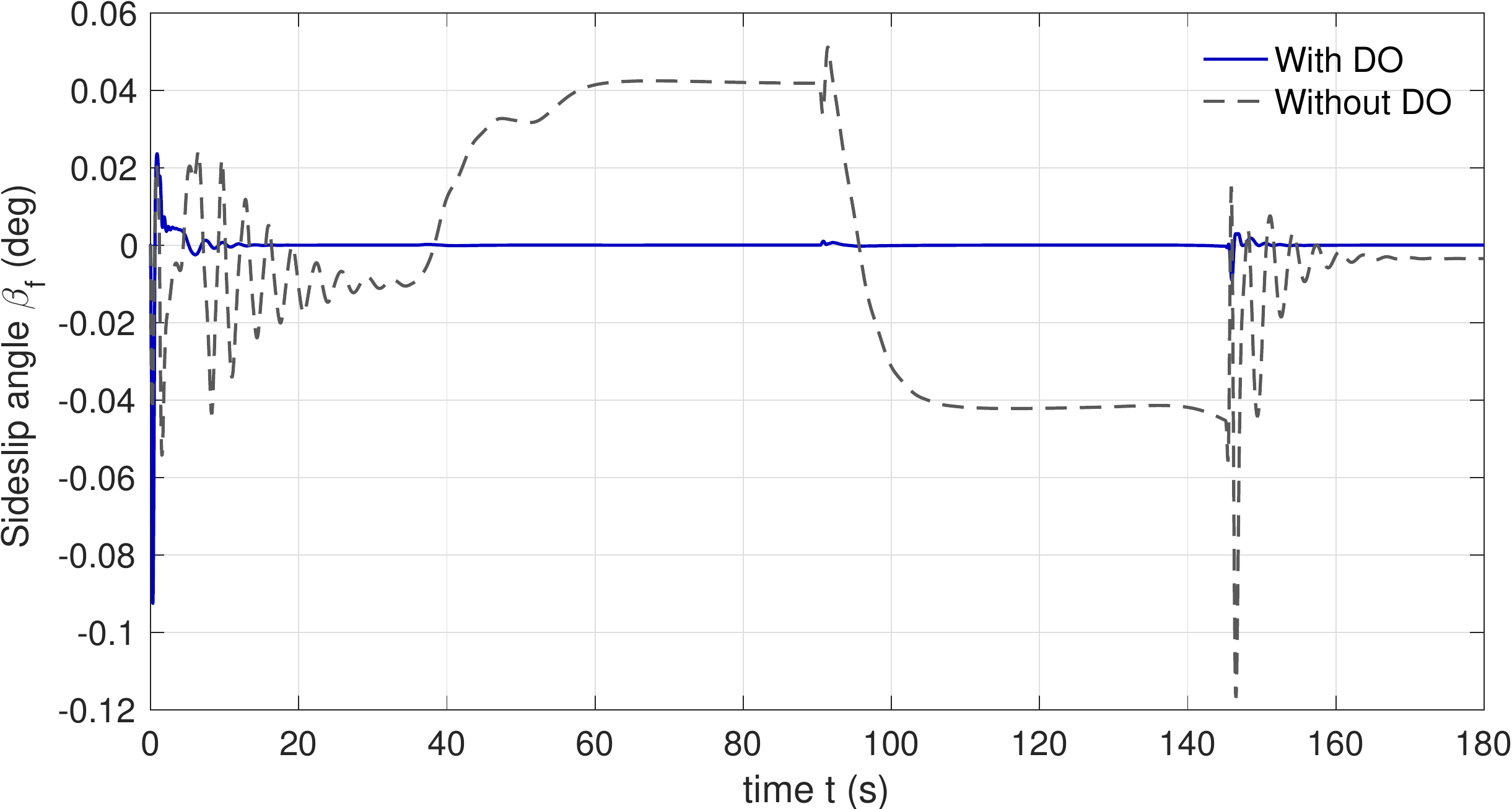}
& &
\includegraphics[width=0.45\textwidth]{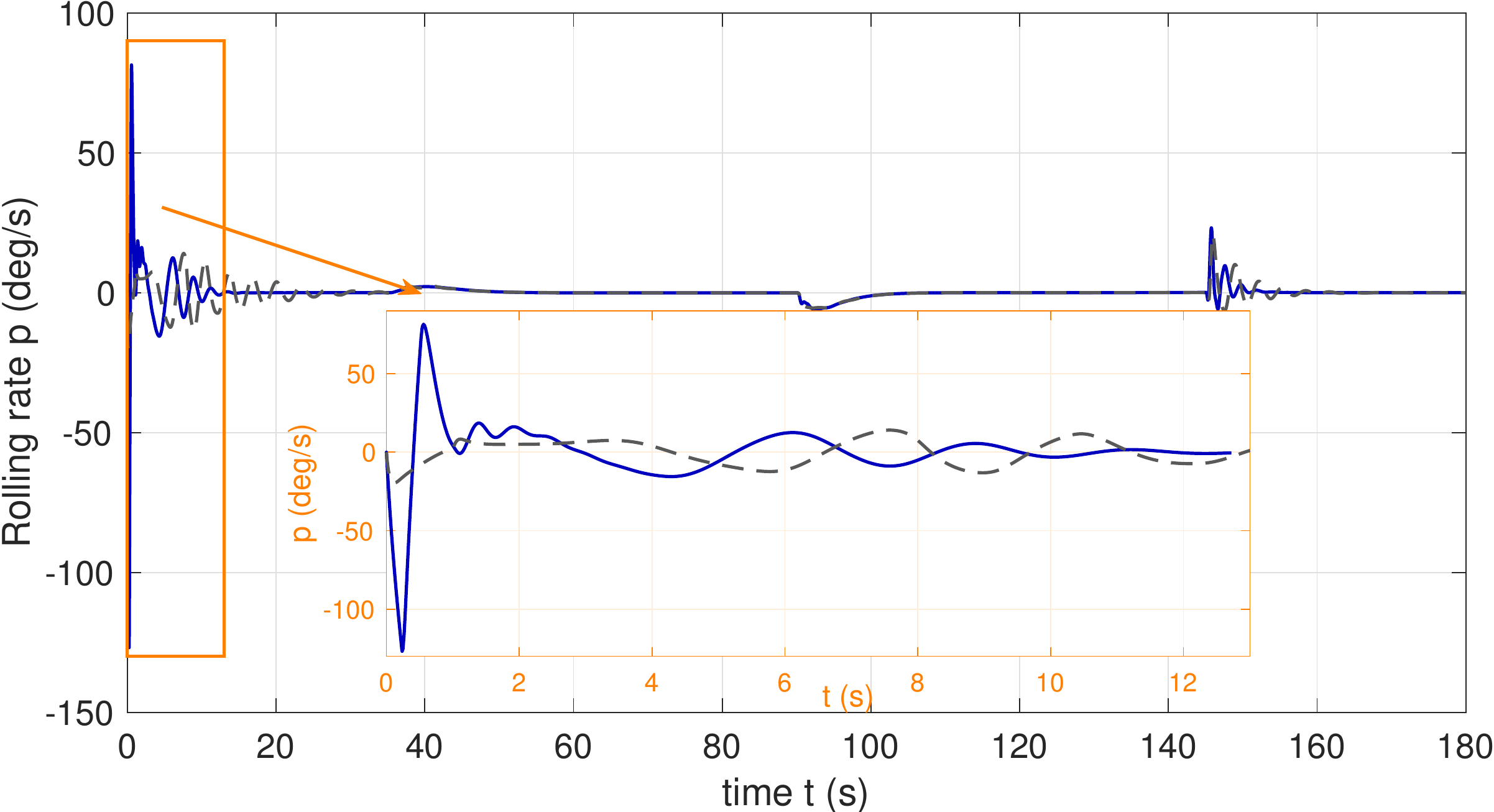} \\
\includegraphics[width=0.45\textwidth]{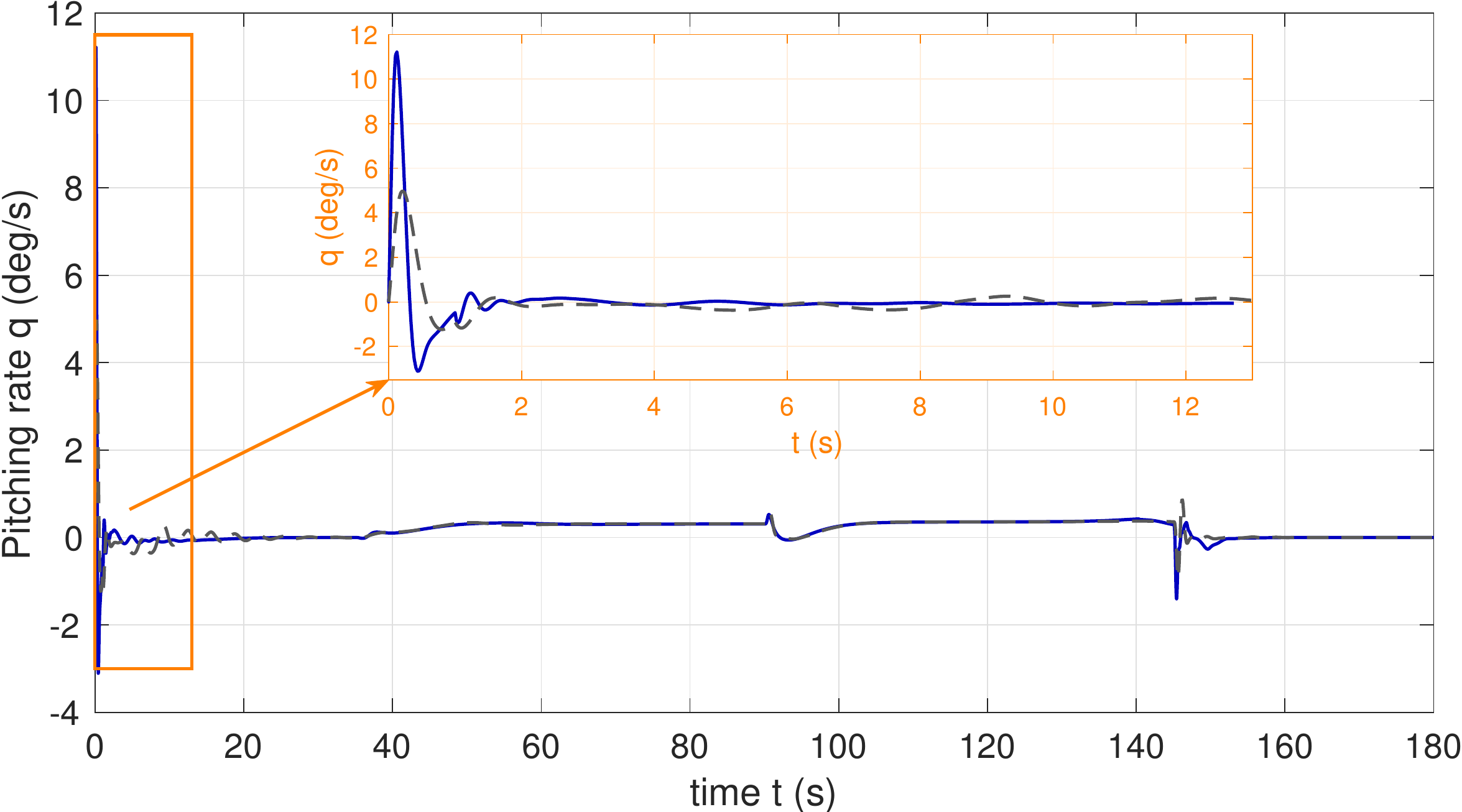}
& &
\includegraphics[width=0.45\textwidth]{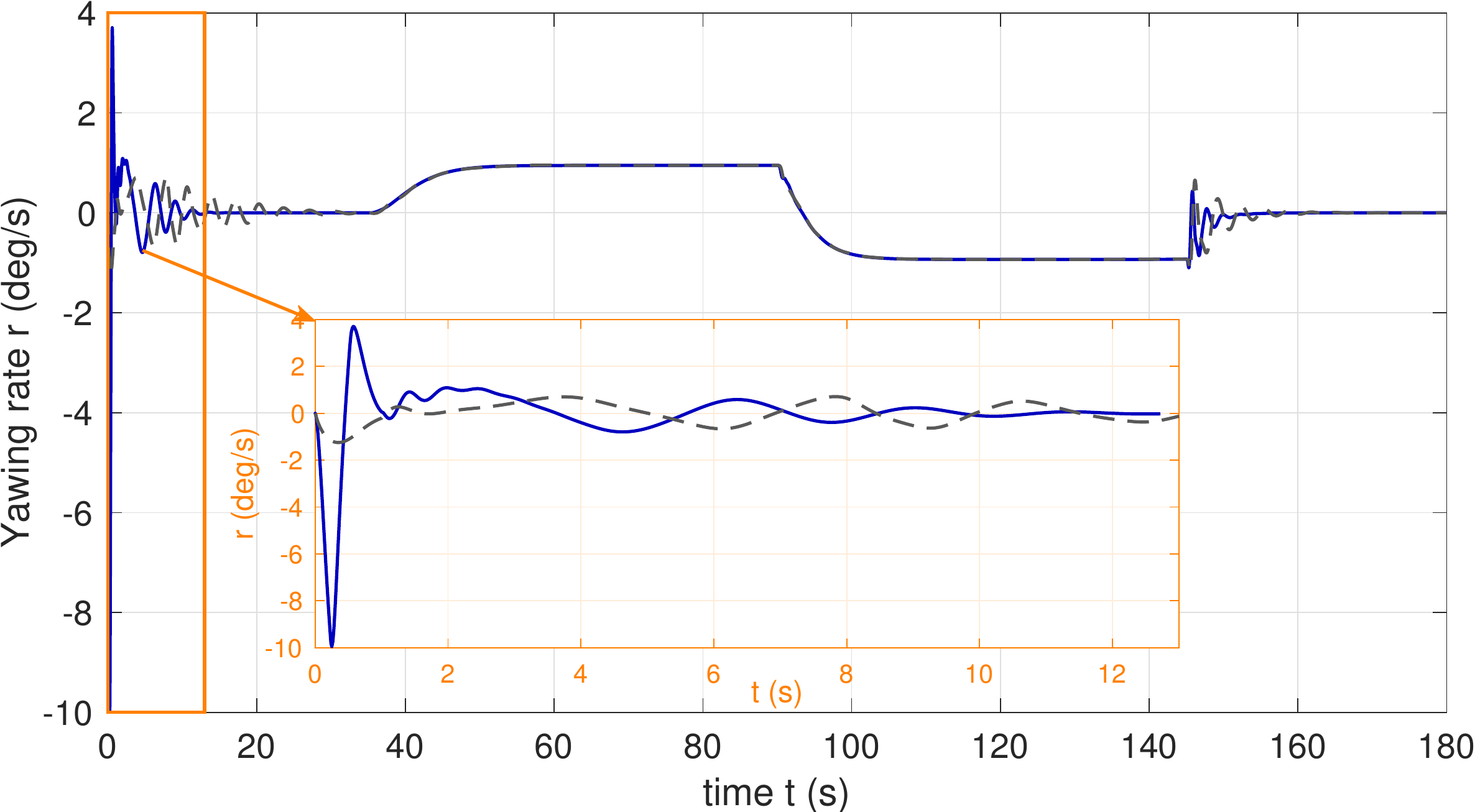}
\end{tabular}\vspace{-5mm}
\caption{Inner-loop state responses  (Scenario 1)}\vspace{-1mm}
\label{Fig: InnerLoopStateResp}
\end{figure}
\begin{figure}[h]
\centering
\begin{tabular}{lll}
\includegraphics[width=0.45\textwidth]{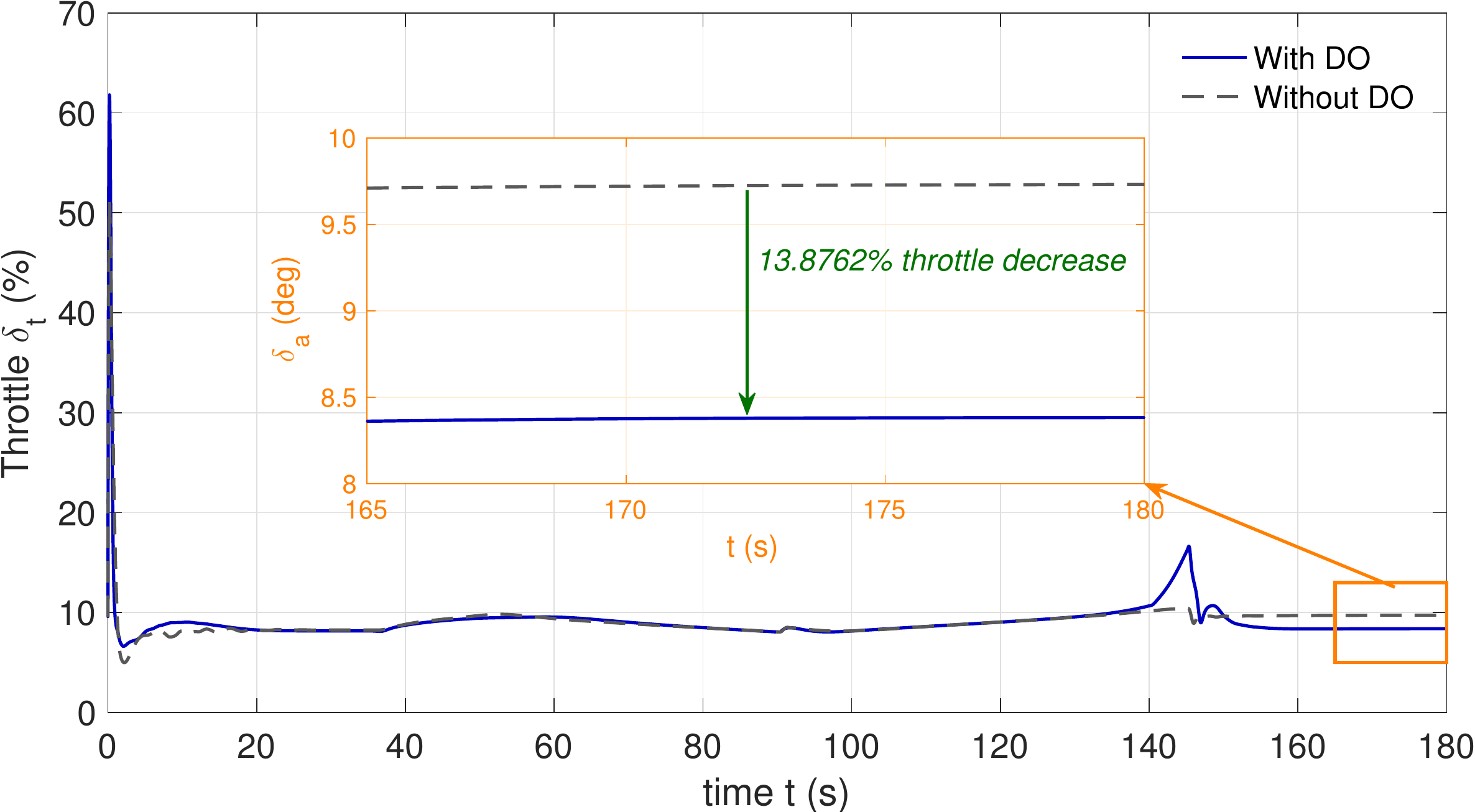}
& &
\includegraphics[width=0.45\textwidth]{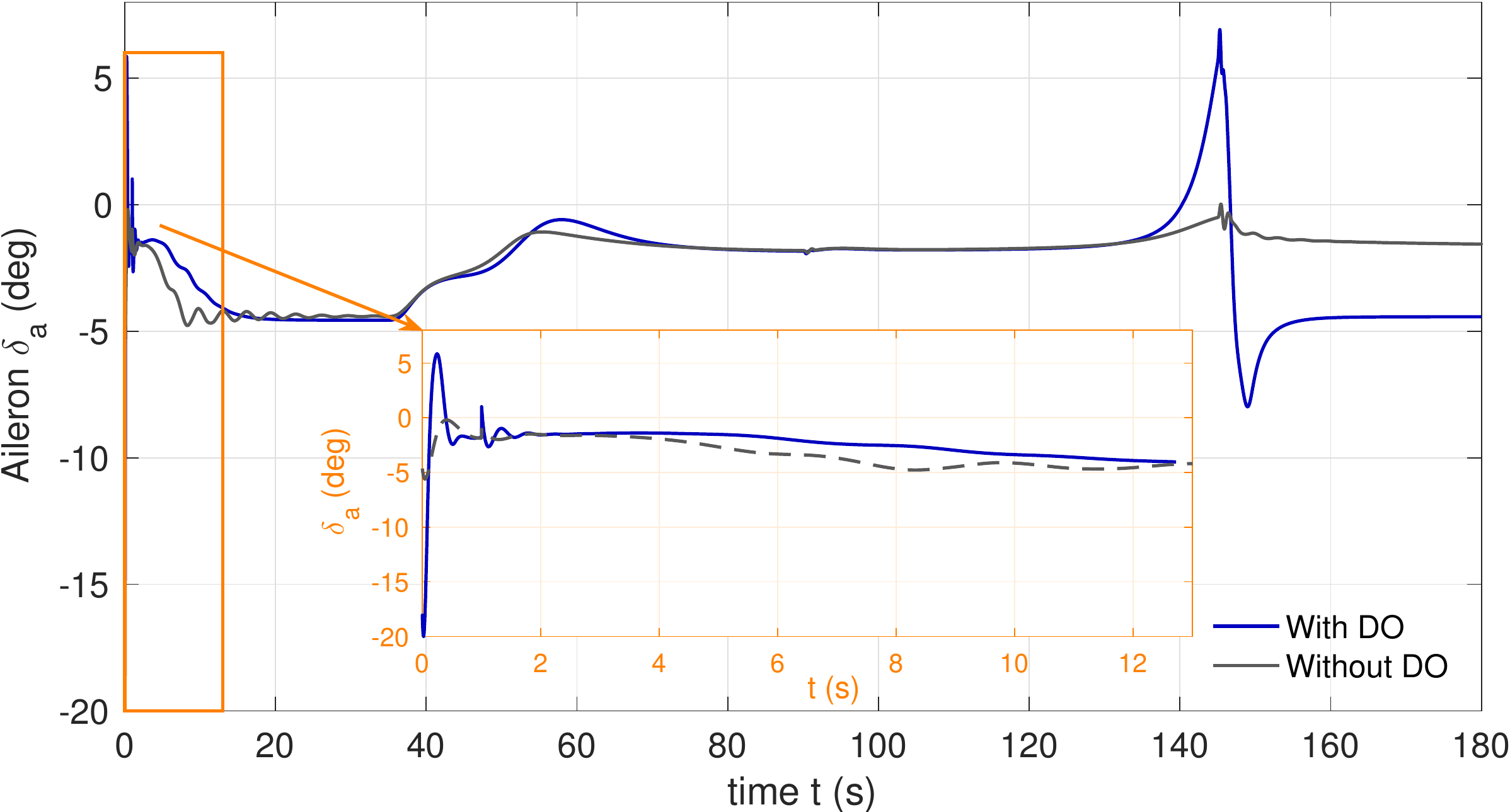} \\
\includegraphics[width=0.45\textwidth]{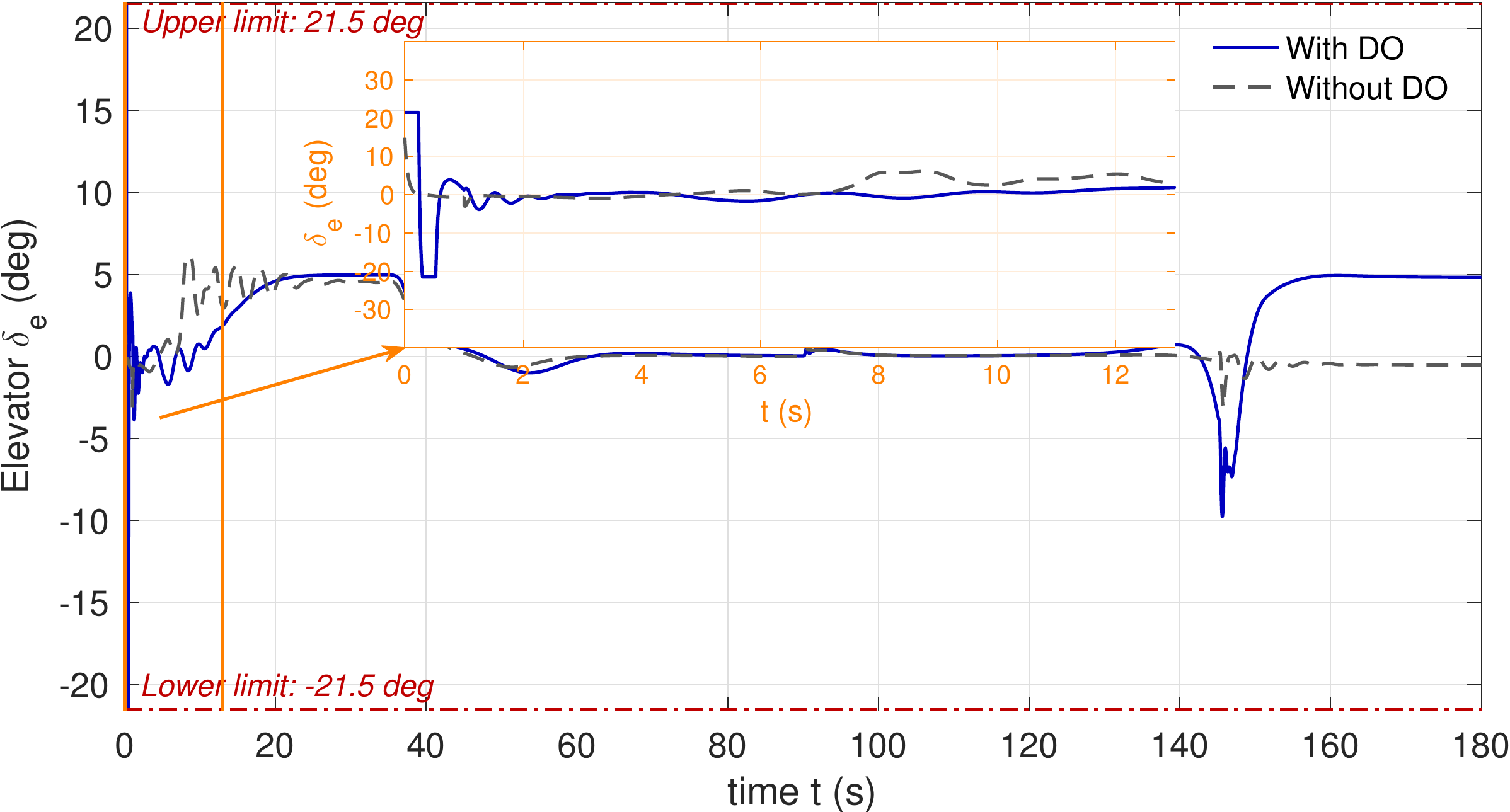}
&&
\includegraphics[width=0.45\textwidth]{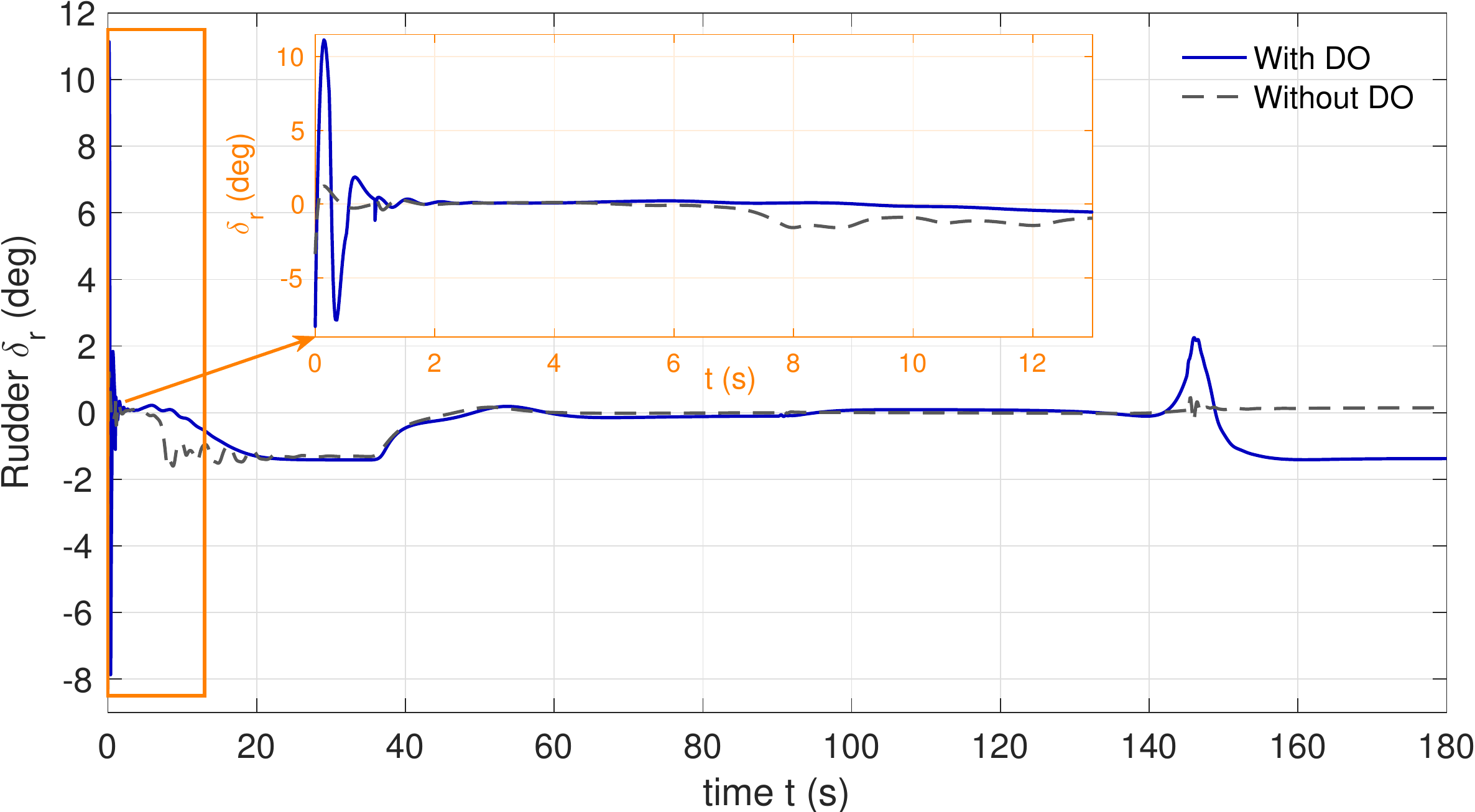} 
\end{tabular}\vspace{-5mm}
\caption{Actuator responses (Scenario 1)}\vspace{-3mm}
\label{Fig: ActuatorResp}
\end{figure}

\vspace{-6mm}

\subsection{Scenario 2: Different flight speeds}
The efficacy of the proposed robust nonlinear controller is further verified by running the close formation flight under different velocities. All the control parameters, the reference trajectories, and the initial conditions will be same as those at the first scenario. Without loss of generality, only position tracking errors and control input responses are given. As compared with a linear controller, nonlinear formation controller can be applied to much wider flight scenarios with stability and performance guaranty.  This advantage of the nonlinear controller is demonstrated in Figure \ref{Fig: ErrX_DiffV}, \ref{Fig: ErrY_DiffV}, and \ref{Fig: ErrZ_DiffV}. It is observed that the proposed robust nonlinear controller is able to ensure almost the same control performance under different velocities. The control input responses under different velocities are illustrated in Figure \ref{Fig: ActuatorResp_DiffV}.  Another interesting observation is that increasing speed will result in better performance in lateral position tracking as shown in Figure \ref{Fig: ErrY_DiffV}. Explaining this observation is difficult from the nonlinear perspective, but we could analyze it from a linear method at a special case. Assume $V_f=V_r=Const$, $\gamma_f=\gamma_r=0$, and $\chi_r=0$.  The nonlinear closed outer-loop lateral dynamics are linearized about its equilibrium, which results in
\begin{equation}
\dot{y}_e=V_fe_\chi\text{, }\quad \dot{e}_\chi =-\frac{K_\chi}{2}{e}_\chi-c_yV_f y_e+d_\chi
\end{equation}
where $d_\chi$ denotes any  uncertainties, disturbances, or inputs.  The transfer function from $d_\chi
$ to ${y}_e$ is
\begin{equation}
G\left(s\right)=\frac{1}{s^2+\frac{K_\chi}{2}s+c_yV_f^2}
\end{equation}
According to the final value theorem, the increase of $V_f$ leads to smaller steady values in ${y}_e$ as illustrated in Figure \ref{Fig: ErrY_DiffV}.
\begin{figure}[htbp]
\centering
\includegraphics[width=0.8\textwidth]{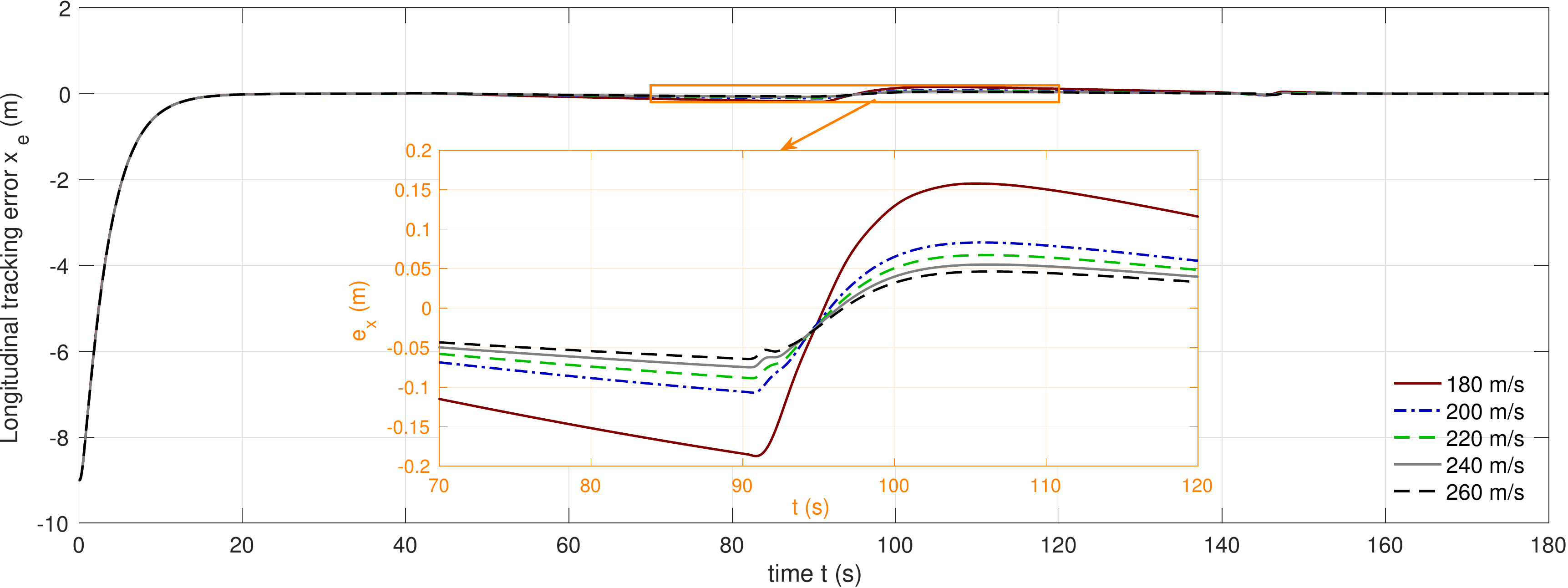}\\ \vspace{-5mm}
\caption{Longitudinal tracking errors, $x_e=x_f-x_r$ (Scenario 2)}\vspace{-3mm}
\label{Fig: ErrX_DiffV}
\end{figure}
\begin{figure}[htbp]
\centering
\includegraphics[width=0.8\textwidth]{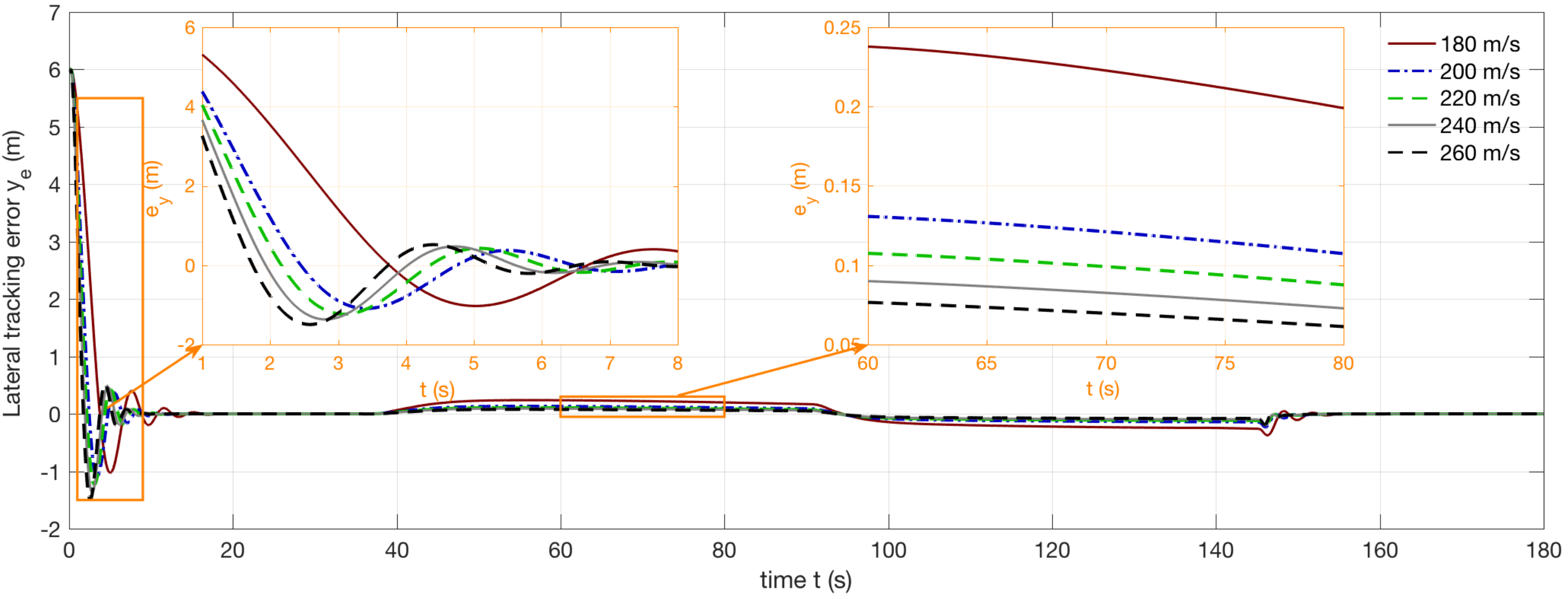}\\ \vspace{-5mm}
\caption{Lateral tracking errors, $y_e=y_f-y_r$ (Scenario 2)}\vspace{-3mm}
\label{Fig: ErrY_DiffV}
\end{figure}
\begin{figure}[htbp]
\centering
\includegraphics[width=0.8\textwidth]{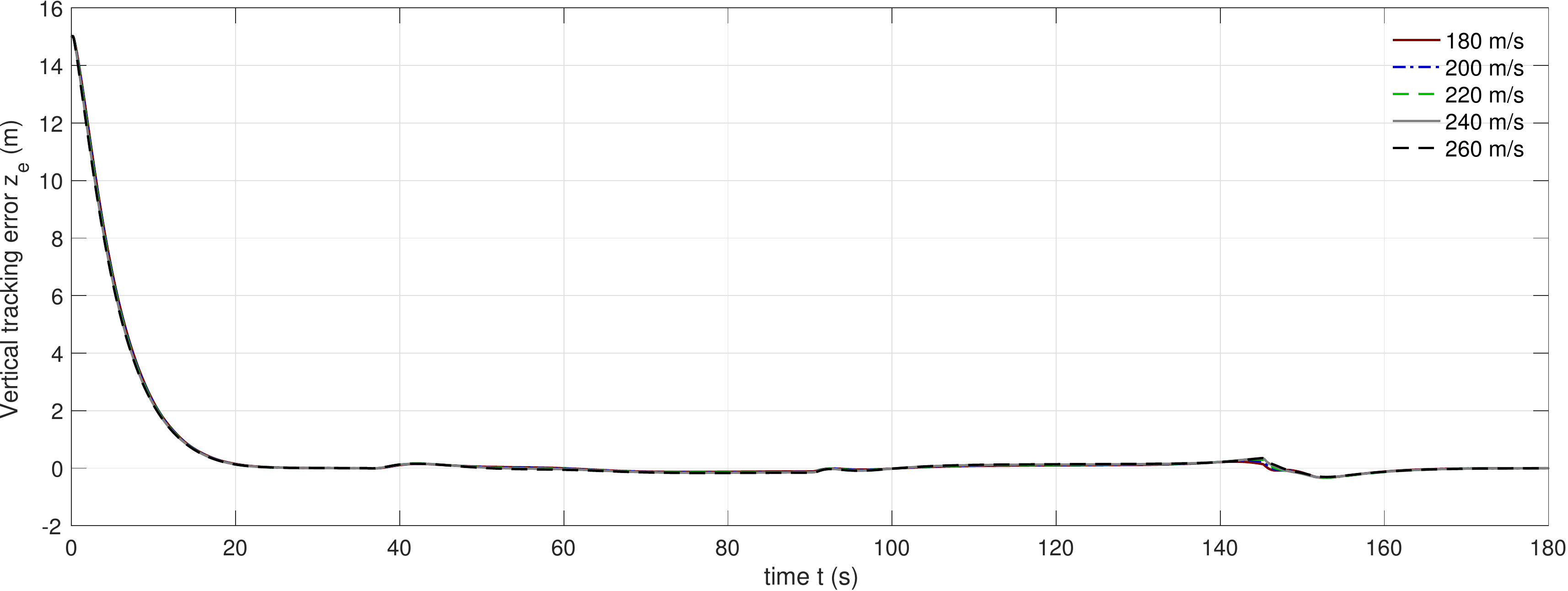}\\ \vspace{-5mm}
\caption{Vertical tracking errors, $z_e=z_f-z_r$ (Scenario 2)}\vspace{-2mm}
\label{Fig: ErrZ_DiffV}
\end{figure}

\begin{figure}[h]
\centering
\begin{tabular}{lll}
\includegraphics[width=0.45\textwidth]{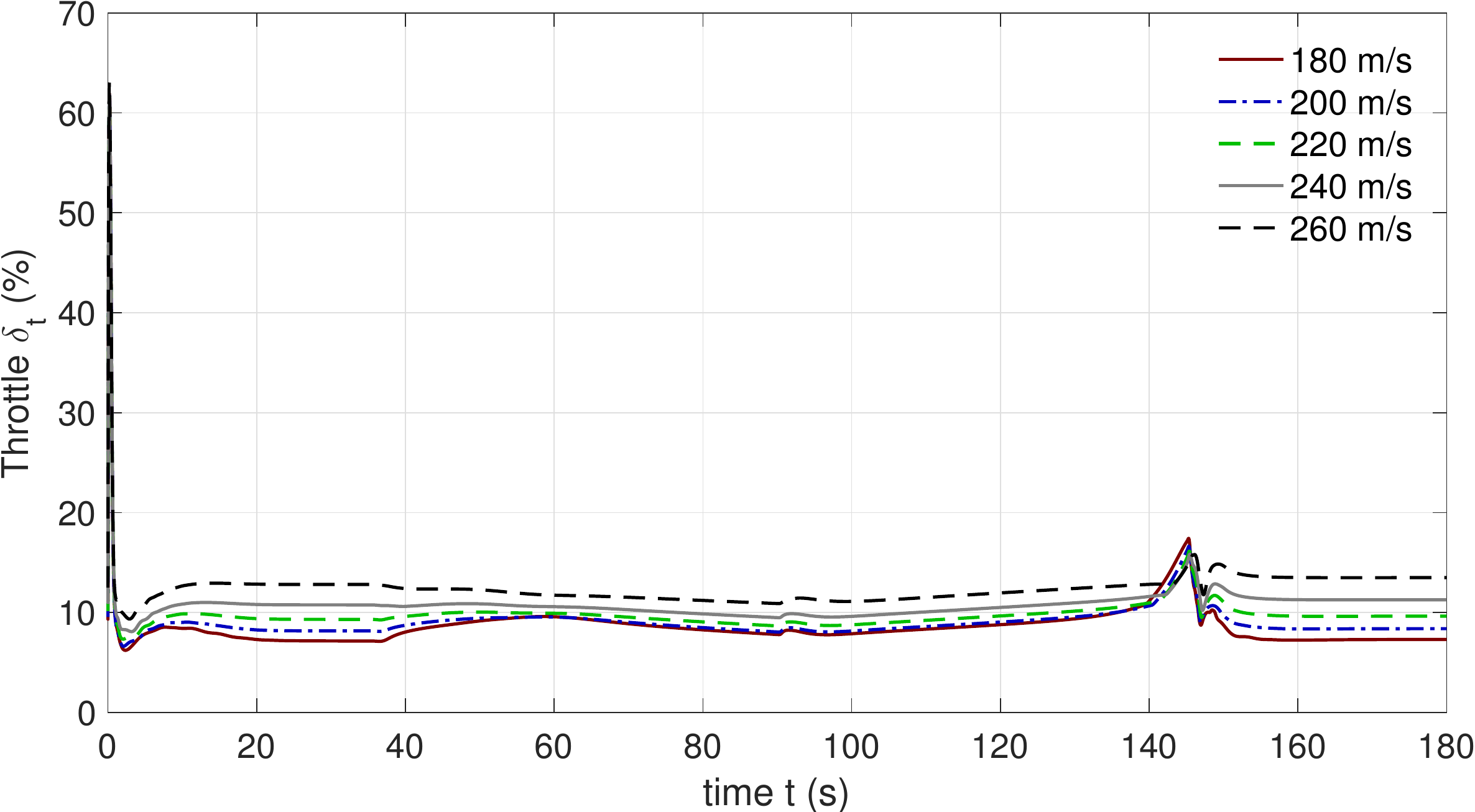}
& &
\includegraphics[width=0.45\textwidth]{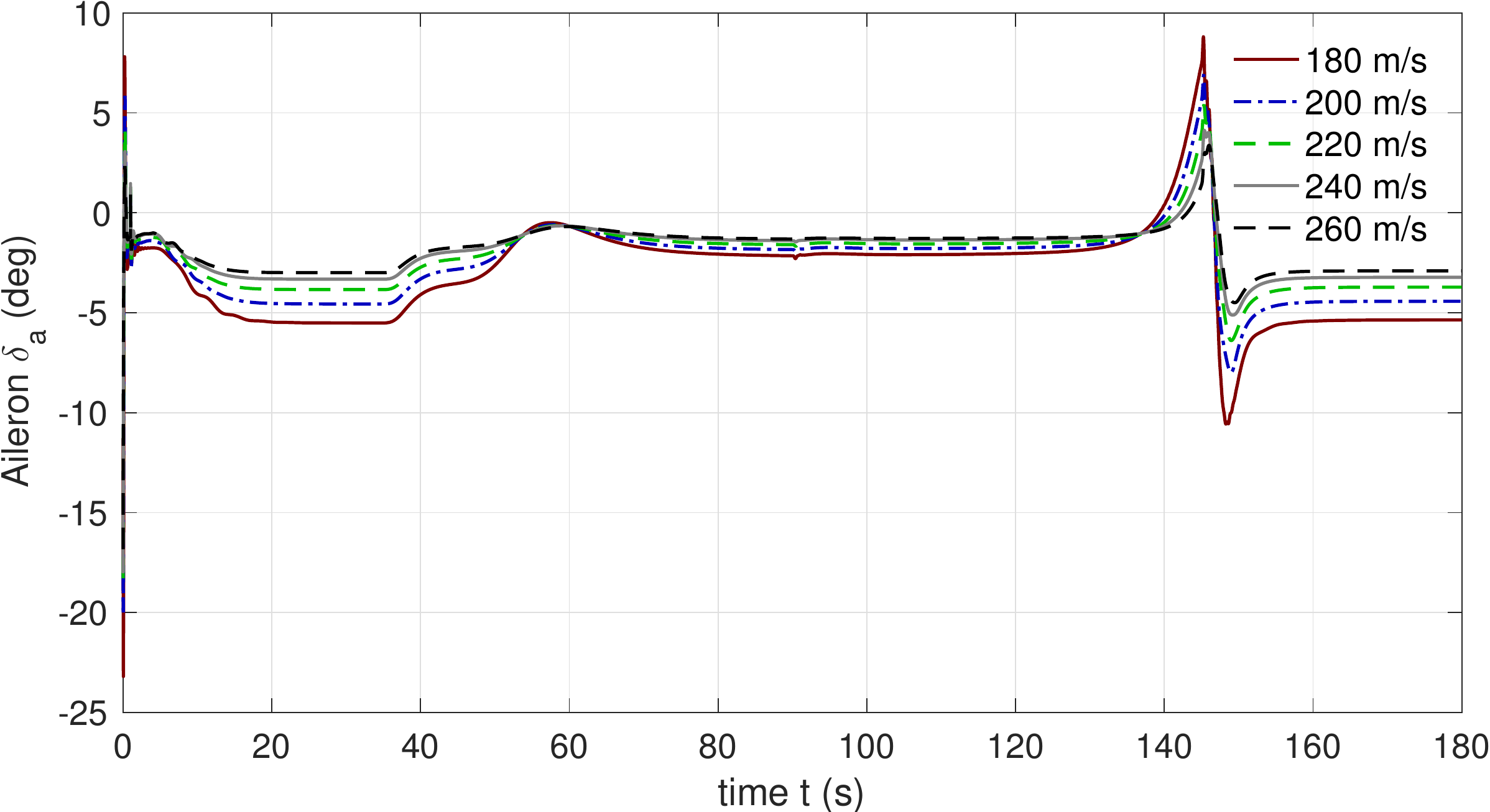} \\
\includegraphics[width=0.45\textwidth]{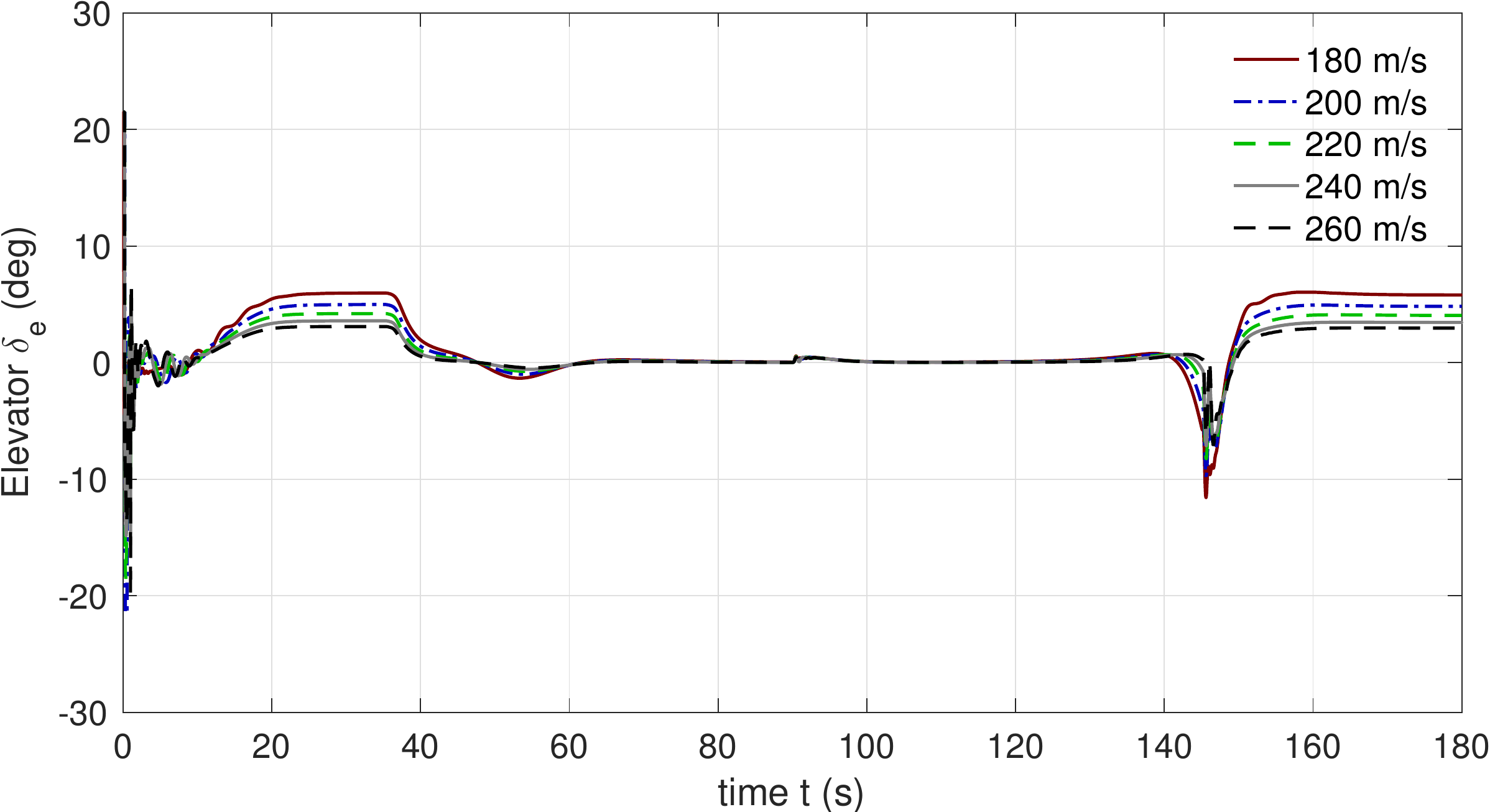}
&&
\includegraphics[width=0.45\textwidth]{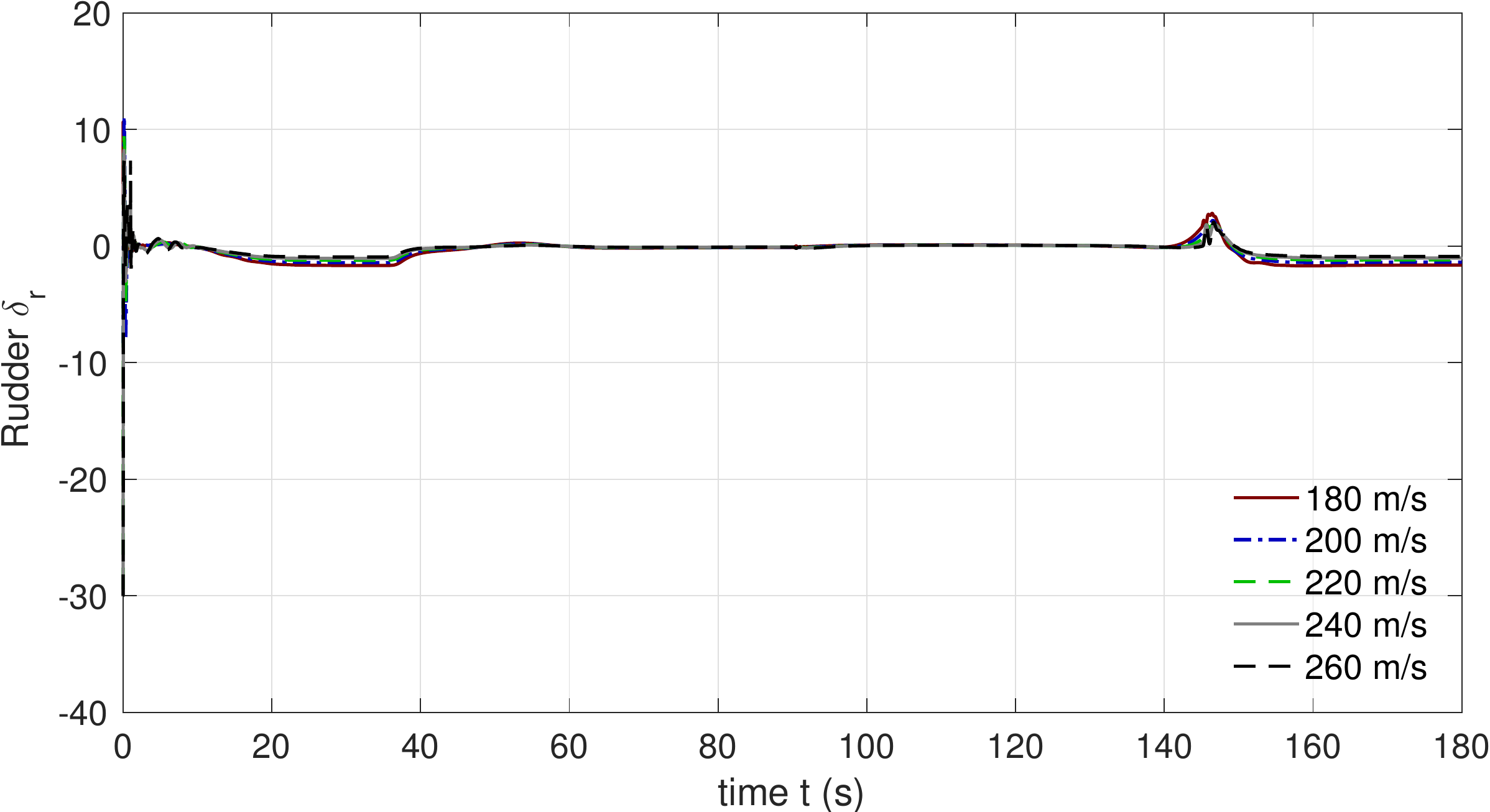} 
\end{tabular}\vspace{-5mm}
\caption{Actuator responses (Scenario 2)}\vspace{-2mm}
\label{Fig: ActuatorResp_DiffV}
\end{figure}

\section{Conclusions} \label{Sec: Conclusion}
The paper presented a robust nonlinear controller for autonomous close formation flight under different flight maneuvers. The proposed controller was developed by combining the command filtered backstepping method and the disturbance observation technique. Both inner-loop and outer-loop controllers were designed in this paper. Based on the proposed design, a follower aircraft is able to track its optimal relative position to a leader aircraft under different flight maneuvers. The proposed design was able to be extended to close formation flight of more than three aircraft, though it was described in the scenario of two-aircraft close formation. Enough robustness and high accuracy could be achieved by the presented design. Different numerical simulations were conducted to demonstrate the efficacy of the presented robust nonlinear close formation controller. %

\section*{Appendix A: Proof of Lemma \ref{Lem: Chap5_DOB_General}}
For (\ref{Eq: Chap5_1st-OrderFilter}),  the error dynamic equation is
\begin{equation}
\dot{\widetilde{d}}=-\frac{1}{\mathcal{T}_d}\widetilde{d}+\dot{d} \label{Eq: Chap5_EstErrDyn_1stOrderFilter}
\end{equation}
By solving the first-order differential equation (\ref{Eq: Chap5_EstErrDyn_1stOrderFilter}), one has $\widetilde{d}\left(t\right) =  e^{-\frac{t}{\mathcal{T}_d}}\widetilde{d}\left(0\right)+\int_{0}^{t}e^{-\frac{t-\tau}{\mathcal{T}_d}}\dot{d}\left(\tau\right)d\tau$. If choosing $\widehat{d}\left(0\right)=0$, we have $\widetilde{d}\left(0\right)=-{d}\left(0\right)$. If $\dot{d}\left(t\right)\in\mathscr{L}_\infty$,
\begin{equation}
\vert\widetilde{d}\left(t\right)\vert  \leq e^{-\frac{t}{\mathcal{T}_d}}\vert{d}\left(0\right)\vert+\int_{0}^{t}e^{-\frac{t-\tau}{\mathcal{T}_d}}d\tau\Vert\dot{d}\Vert_\infty
				                   \leq e^{-\frac{t}{\mathcal{T}_d}}\vert{d}\left(0\right)\vert+\mathcal{T}_d\left(1-e^{-\frac{t}{\mathcal{T}_d}}\right)\Vert\dot{d}\Vert_\infty \label{Eq: EstErrBound}
\end{equation}
Obviously,  $\vert\widetilde{d}\vert\leq \max\left\{\vert{d}\left(0\right)\vert\text{, }\mathcal{T}_d\Vert\dot{d}\Vert_\infty\right\} $ globally exists. 

For any positive constant $\epsilon$, if $\mathcal{T}_d\Vert\dot{d}\Vert_\infty+\epsilon>\vert{d}\left(0\right)\vert$,  there exists $\vert\widetilde{d}\vert<\mathcal{T}_d\Vert\dot{d}\Vert_\infty+\epsilon$ $\forall t_\epsilon>0$ according to the global boundedness of $\vert\widetilde{d}\vert$. If $\mathcal{T}_d\Vert\dot{d}\Vert_\infty+\epsilon<\vert{d}\left(0\right)\vert$, the first term on the right-hand side of (\ref{Eq: EstErrBound}) will decrease, while the second term will increase. At certain time $t_\epsilon$, there exists $e^{-\frac{t}{\mathcal{T}_d}}\vert{d}\left(0\right)\vert+\mathcal{T}_d\left(1-e^{-\frac{t}{\mathcal{T}_d}}\right)\Vert\dot{d}\Vert_\infty=\mathcal{T}_d\Vert\dot{d}\Vert_\infty+\epsilon$. After $t_\epsilon$, the right hand side of (\ref{Eq: EstErrBound}) will be smaller than $\mathcal{T}_d\Vert\dot{d}\Vert_\infty+\epsilon$. By simple mathematical calculation, one has $t_\epsilon=\mathcal{T}_d\ln\left(\frac{\vert\vert{d}\left(0\right)\vert-\mathcal{T}_d\Vert\dot{d}\Vert_\infty\vert}{\epsilon}\right)$, so the second conclusion of Lemma \ref{Lem: Chap5_DOB_General} is obtained. In addition, the estimator error is input-to-state stable with respect to $\dot{d}$ according to (\ref{Eq: Chap5_EstErrDyn_1stOrderFilter}). Hence, it is obvious that $\lim_{t\to\infty} \widetilde{d} =0 $, if $\lim_{t\to\infty} \dot{d} =0 $.

\section*{Appendix B: Proof of Lemma \ref{Lem: 2ndOrderEstProperties}}
 The following error dynamics are easy to obtain. \vspace{-3mm}
\begin{equation}
\dot{\mathbf{e}}_{\mathscr{S}}=\mathbf{A}_{\mathscr{S}}\mathbf{e}_{\mathscr{S}}+\mathbf{B}_{\mathscr{S}} \left(\ddot{\mathscr{S}}+2\zeta_\mathscr{S}\omega_\mathscr{S}\dot{\mathscr{S}}\right)
\label{Eq: ErrDyn_Lem1}
\end{equation} \vspace{-3mm}
where $\mathbf{B}_{\mathscr{S}}=\left[0\text{, }-1\right]^T$. By choosing $\zeta_\mathscr{S}\text{, } \omega_\mathscr{S}>0$, $\mathbf{A}_{\mathscr{S}}$ is Hurwitz, which implies $\mathbf{P}_{\mathscr{S}}\mathbf{A}_{\mathscr{S}}+\mathbf{A}^T_{\mathscr{S}}\mathbf{P}_{\mathscr{S}}=-\mathbf{I}$ with $\mathbf{P}_{\mathscr{S}}>0$. Choose $\mathbb{V}=\mathbf{e}_{\mathscr{S}}^T\mathbf{P}_{\mathscr{S}}\mathbf{e}_{\mathscr{S}}$ as the Lyapunov function for (\ref{Eq: ErrDyn_Lem1}), so
\begin{equation}
\lambda_{min}\left(\mathbf{P}_{\mathscr{S}}\right)\Vert\mathbf{e}_{\mathscr{S}}\Vert_2^2\leq\mathbb{V}\leq\lambda_{max}\left(\mathbf{P}_{\mathscr{S}}\right)\Vert\mathbf{e}_{\mathscr{S}}\Vert_2^2 \label{Eq: Lem-Ineq-Lyap}
\end{equation}
If both $\ddot{\mathscr{S}}$ and $\dot{\mathscr{S}}$ are bounded, differentiating $\mathbb{V}$ with respect to time will yield
\begin{equation*}
\begin{array}{ll}
\dot{\mathbb{V}} 	&  =  \mathbf{e}_{\mathscr{S}}^T\left(\mathbf{P}_{\mathscr{S}}\mathbf{A}_{\mathscr{S}}+\mathbf{A}_{\mathscr{S}}^T\mathbf{P}_{\mathscr{S}}\right)
					    \mathbf{e}_{\mathscr{S}}+2\mathbf{e}_{\mathscr{S}}^T\mathbf{P}_{\mathscr{S}}\mathbf{B}_{\mathscr{S}}\left(\ddot{\mathscr{S}}
					    +2\zeta_\mathscr{S}\omega_\mathscr{S}\dot{\mathscr{S}}\right)\\
				&\leq-\frac{\mathbb{V}}{\lambda_{max}\left(\mathbf{P}_{\mathscr{S}}\right)}+\frac{2\lambda_{max}\left(\mathbf{P}_{\mathscr{S}}\right)}{\sqrt{\lambda_{min}\
					   \left(\mathbf{P}_{\mathscr{S}}\right)}}\Vert\ddot{\mathscr{S}}+2\zeta_\mathscr{S}\omega_\mathscr{S}\dot{\mathscr{S}}\Vert_{\infty} \sqrt{\mathbb{V}}
\end{array}
\end{equation*}
Let $\mathbb{W}\left(t\right)=\sqrt{\mathbb{V}}$, so $\dot{\mathbb{W}}=\frac{\dot{\mathbb{V}}}{2\sqrt{\mathbb{V}}}$, when $\mathbb{V}\neq0$. Hence, $\dot{\mathbb{W}} \leq -{\mathbb{W} }\left/{\left(2\lambda_{max}\left(\mathbf{P}_{\mathscr{S}}\right)\right)}\right.+{\lambda_{max}\left(\mathbf{P}_{\mathscr{S}}\right)}
					\Vert\ddot{\mathscr{S}}+2\zeta_\mathscr{S}\omega_\mathscr{S}\dot{\mathscr{S}}\Vert_{\infty}\left/{\sqrt{\lambda_{min}\left(\mathbf{P}_{\mathscr{S}}\right)}}\right.$.
According to the comparison principle (Page 102, \cite{Khalil2002Book}), one can get
\begin{equation*}
\begin{array}{ll}
\mathbb{W} &\leq  e^{-\frac{t }{2\lambda_{max}\left(\mathbf{P}_{\mathscr{S}}\right)}}\mathbb{W}\left(0\right)+\left(1-e^{-\frac{t}{2\lambda_{max}
		    		\left(\mathbf{P}_{\mathscr{S}}\right)}}\right)\frac{2\lambda^2_{max}\left(\mathbf{P}\right)}{\sqrt{\lambda_{min}\left(\mathbf{P}\right)}}
				\Vert\ddot{\mathscr{S}}+2\zeta_\mathscr{S}\omega_\mathscr{S}\dot{\mathscr{S}}\Vert_{\infty}
\end{array} \label{Eq: W-Ineq-Chap5} 
\end{equation*}
According to (\ref{Eq: Lem-Ineq-Lyap}), $\mathbb{W}\leq\sqrt{\lambda_{max}\left(\mathbf{P}_{\mathscr{S}}\right)}\Vert\mathbf{e}_{\mathscr{S}}\Vert_2$, 
so $\mathbb{W}\left(0\right)\leq\lambda_{max}\left(\mathbf{P}_{\mathscr{S}}\right)\Vert\mathbf{e}_{\mathscr{S}}\left(0\right)\Vert_2$. By setting $\mathscr{S}_c\left(0\right)=\mathscr{S}\left(0\right)$, and $\dot{\mathscr{S}}_c\left(0\right)=0$, one has  $\Vert\mathbf{e}_{\mathscr{S}}\left(0\right)\Vert=\vert \dot{\mathscr{S}}\left(0\right)\vert$. Eventually,
\begin{equation*}\footnotesize
\Vert\mathbf{e}_{\mathscr{S}}\Vert_2\leq \sqrt{\frac{\lambda_{max}\left(\mathbf{P_{\mathscr{S}}}\right)}{\lambda_{min}\left(\mathbf{P_{\mathscr{S}}}\right)}}e^{-\frac{t}{2\lambda_{max}\left(\mathbf{P_{\mathscr{S}}}\right)}}\vert \dot{\mathscr{S}}\left(0\right)\vert+\left(1-e^{-\frac{t}{2\lambda_{max}\left(\mathbf{P_{\mathscr{S}}}\right)}}\right)\frac{2\lambda_{max}^2\left(\mathbf{P}_{\mathscr{S}}\right)}{\lambda_{min}\left(\mathbf{P}_{\mathscr{S}}\right)}\Vert\ddot{\mathscr{S}}+2\zeta_\mathscr{S}\omega_{\mathscr{S}}\dot{\mathscr{S}}\Vert_{\infty}
\end{equation*}
Therefore, $\mathbf{e}_{\mathscr{S}}$ is uniformly and ultimately bounded. With the consideration of $\Vert\tilde{\mathscr{S}}\Vert_2\leq\Vert\mathbf{e}_{\mathscr{S}}\Vert_2$, we are able to conclude that $\tilde{\mathscr{S}}$ is uniformly and ultimately bounded. 

The second conclusion of Lemma \ref{Lem: 2ndOrderEstProperties} could be demonstrated by virtue of the perturbation theory. According to (\ref{Eq: ErrDyn_Lem1}), the following singularly perturbed system is readily obtained.
\begin{equation*}\small
\frac{1}{\omega_\mathscr{S}}\dot{\tilde{\mathscr{S}}}_1=\tilde{ \mathscr{S}}_2\text{, }\qquad
\frac{1}{\omega_\mathscr{S}}\dot{\tilde{ \mathscr{S}}}_2=-2\zeta_\mathscr{S}\tilde{ \mathscr{S}}_2-\tilde{ \mathscr{S}}_1-\frac{1}{\omega_\mathscr{S}}\left(\frac{1}{\omega_\mathscr{S}}\ddot{\mathscr{S}}+2\zeta_\mathscr{S}\dot{\mathscr{S}}\right)
\end{equation*}
where $\tilde{\mathscr{S}}_1=\tilde{\mathscr{S}}$ and $\tilde{\mathscr{S}}_2=\frac{1}{\omega_\mathscr{S}}\dot{\tilde{\mathscr{S}}}$. Obviously, the effect of $\frac{1}{\omega_\mathscr{S}}\ddot{\mathscr{S}}+2\zeta_\mathscr{S}\dot{\mathscr{S}}$ will be diminished by increasing the natural frequency $\omega_\mathscr{S}$. According to the properties of singularly perturbed system, we readily have $ \tilde{\mathscr{S}}_1=\mathscr{S}_c- \mathscr{S}=\mathcal{O}\left(\frac{1}{\omega_ \mathscr{S}}\right)$ and $\dot{\tilde{\mathscr{S}}}_2=\frac{\dot{\tilde{\mathscr{S}}}}{\omega_\mathscr{S}}=\mathcal{O}\left(\frac{1}{\omega_\mathscr{S}}\right)$.

\section*{Appendix C: Proof of Lemma \ref{Lem: Chap5_InnerStab}}
In light of (\ref{Eq: Chap5_Errdyn_Theta2}), (\ref{Eq: Chap5_Omega_ErrDyn}), (\ref{Eq: Chap5_Desired_u_Theta}), (\ref{Eq: Chap5_Omega_ErrDyn2}), (\ref{Eq: Chap5_DesiredTorque}), (\ref{Eq: Chap5_Auxil_Xi_Theta}), and (\ref{Eq: Chap5_D-Theta-Tau_EstErr}), one has
\begin{equation}
\left\{\begin{array}{ll}
\boldsymbol{\dot{\varepsilon}}_\Theta &= -\mathbf{K}_{\Theta}\boldsymbol{\varepsilon}_\Theta +\mathbf{G}\mathbf{e}_\Omega-\mathbf{\widetilde{d}}_\Theta \\
\mathbf{\dot{e}}_\Omega &= -\mathbf{K}_{\Omega}{\mathbf{e}}_\Omega-\mathbf{C}_\Omega\mathbf{G}^T\boldsymbol{\varepsilon}_\Theta-\mathbf{\widetilde{d}_\tau} 
\end{array}\right.
\qquad
\left\{\begin{array}{ll}
\mathbf{\dot{\widetilde{d}}}_\Theta&=-\boldsymbol{\mathcal{T}}_\Theta^{-1}\widetilde{\mathbf{d}}_\Theta  \\
\mathbf{\dot{\widetilde{d}}}_\tau&=-\boldsymbol{\mathcal{T}}_\Omega^{-1}\widetilde{\mathbf{d}}_\tau
\end{array}\right.  \label{Eq: Chap5_ClosedInnerLoop}
\end{equation}

The stability of (\ref{Eq: Chap5_ClosedInnerLoop}) is shown by picking the Lyapunov function $\mathbb{V}$ as below.
\begin{equation}
\mathbb{V} =\frac{\boldsymbol{\varepsilon}_\Theta^T\boldsymbol{\varepsilon}_\Theta}{2}+\frac{\mathbf{e}_\Omega^T\mathbf{C}_\Omega^{-1}\mathbf{e}_\Omega}{2}+\frac{\widetilde{\mathbf{d}}_\Theta^T\widetilde{\mathbf{d}}_\Theta}{2}+\frac{\widetilde{\mathbf{d}}_\tau^T\mathbf{C}_\Omega^{-1}\widetilde{\mathbf{d}}_\tau}{2}
\end{equation}
The derivative of $\mathbb{V}$ is 
\begin{equation*}\small
\begin{array}{ll}
\dot{\mathbb{V}}&= \boldsymbol{\varepsilon}_\Theta^T\left(-\mathbf{K}_{\Theta}\boldsymbol{\varepsilon}_\Theta +\mathbf{G}\mathbf{e}_\Omega-\mathbf{\widetilde{d}}
					_\Theta\right) - \mathbf{e}_\Omega^T\mathbf{C}_\Omega^{-1}\left(\mathbf{K}_{\Omega}{\mathbf{e}}_\Omega+\mathbf{C}_\Omega\mathbf{G}
					^T\boldsymbol{\varepsilon}_\Theta+\mathbf{\widetilde{d}_\tau} \right)-\mathbf{\widetilde{d}}_\Theta^T\boldsymbol{\mathcal{T}}_\Theta^{-1}\mathbf{\widetilde{d}}_\Theta  -\mathbf{\widetilde{d}}_\tau^T\mathbf{C}_\Omega^{-1}\boldsymbol{\mathcal{T}}_\Omega^{-1}\mathbf{\widetilde{d}}
					_\tau\nonumber \\
				&=	- \boldsymbol{\varepsilon}_\Theta^T\mathbf{K}_{\Theta}\boldsymbol{\varepsilon}_\Theta-\boldsymbol{\varepsilon}_\Theta^T\mathbf{\widetilde{d}}
					_\Theta-\mathbf{e}_\Omega^T\mathbf{C}_\Omega^{-1}\mathbf{K}_{\Omega}{\mathbf{e}}_\Omega-\mathbf{e}_\Omega^T\mathbf{C}
					_\tau^{-1}\mathbf{\widetilde{d}_\tau}-\mathbf{\widetilde{d}}_\Theta^T\boldsymbol{\mathcal{T}}_\Theta^{-1}\mathbf{\widetilde{d}}_\Theta-
					\mathbf{\widetilde{d}}_\tau^T\mathbf{C}_\Omega^{-1}\boldsymbol{\mathcal{T}}_\Omega^{-1}\mathbf{\widetilde{d}}_\tau  \nonumber 
\end{array}
\end{equation*}
where $\mathbf{\widetilde{d}}_\Theta$ and $\mathbf{\widetilde{d}}_\tau$ are both vectors. Let $\mathbf{\widetilde{d}}_\Theta=\left[\widetilde{d}_\mu\text{, }\widetilde{d}_\alpha\text{, }\widetilde{d}_\beta\right]^T$ and $\mathbf{\widetilde{d}}_\tau=\left[\widetilde{d}_p\text{, }\widetilde{d}_q\text{, }\widetilde{d}_r\right]^T$, so
\begin{eqnarray}
\dot{\mathbb{V}}&=& -K_\mu\varepsilon_\mu^2-K_\alpha\varepsilon_\alpha^2-K_\beta\varepsilon_\beta^2-\varepsilon_\mu\widetilde{d}_\mu-\varepsilon_\alpha\widetilde{d}
					_\alpha-\varepsilon_\beta\widetilde{d}_\beta-\frac{\widetilde{d}_\mu^2}{\mathcal{T}_\mu}-\frac{\widetilde{d}_\alpha^2}{\mathcal{T}_\alpha}-
					\frac{\widetilde{d}_\beta^2}{\mathcal{T}_\beta}\nonumber \\
				& & -\frac{K_pe_p^2}{c_p}-\frac{K_qe_q^2}{c_q}-\frac{K_re_r^2 }{c_r}-\frac{e_p\widetilde{d}_p}{c_p}-\frac{e_q\widetilde{d}_q}{c_q}
					-\frac{e_r\widetilde{d}_r}{c_r}-\frac{\widetilde{d}_p^2}{c_p\mathcal{T}_p}-\frac{\widetilde{d}_q^2}{c_q\mathcal{T}_q}-\frac{\widetilde{d}_r^2}{c_r\mathcal{T}_r}
\end{eqnarray}
It is easy to show that $\dot{\mathbb{V}}_2$ is negative definite, if all control parameters are chosen to be positive. Hence, $\boldsymbol{\varepsilon}_\Theta$ and ${\mathbf{e}}_\Omega$ will exponentially converge to zero. To show the second conclusion in Lemma \ref{Lem: Chap5_InnerStab}, a singular perturbation model is established based on (\ref{Eq: Chap5_Cmd_PQR}) and (\ref{Eq: Chap5_Auxil_Xi_Theta}). To simplify the analysis complexity, assume that $\omega_p=\omega_q=\omega_r=\omega_\Omega$ and $\zeta_p=\zeta_q=\zeta_r=\zeta_\Omega$. Define a new variable $\boldsymbol{\bar{\Omega}}_{cdot}= \frac{\boldsymbol{\dot{\Omega}}_c}{\omega_\Omega}$. Therefore, the singular perturbation model is \vspace{-5mm}
\begin{subequations} \label{Eq: Chap5_SingularPerturb_Theta}
\begin{align}
\boldsymbol{\dot{\xi}}_\Theta&=-\mathbf{K}_{\Theta}\boldsymbol{{\xi}}_{{\Theta}}+\mathbf{G}\left(\boldsymbol{\Omega}_c-\boldsymbol{\Omega}_d\right)   \label{Eq: Chap5_SingularPerturb_Theta1}\\
\frac{1}{\omega_\Omega}\boldsymbol{\dot{\Omega}}_c&=\boldsymbol{\bar{\Omega}}_{cdot}  \label{Eq: Chap5_SingularPerturb_Theta2}\\
\frac{1}{\omega_\Omega}\boldsymbol{\dot{\bar{\Omega}}}_{cdot} &= -\frac{2\zeta_\Omega}{\omega_\Omega}\boldsymbol{\dot{\bar{\Omega}}}_{cdot}-\left(\boldsymbol{\Omega}_c-\boldsymbol{\Omega}_d\right)  \label{Eq: Chap5_SingularPerturb_Theta3}
\end{align}
\end{subequations}
Obviously, the command filter system composed of (\ref{Eq: Chap5_SingularPerturb_Theta2}) and (\ref{Eq: Chap5_SingularPerturb_Theta3}) has much faster dynamics than the auxiliary system (\ref{Eq: Chap5_SingularPerturb_Theta1}), if $\omega_\Omega$ is chosen to be sufficiently large. The reduced system of (\ref{Eq: Chap5_SingularPerturb_Theta}) is given by $\boldsymbol{\dot{\bar{\xi}}}_\Theta=-\mathbf{K}_{\Theta}\boldsymbol{\bar{\xi}}_{\Theta}$ by setting $\omega_\Omega$ to be infinity, where $\boldsymbol{\bar{\xi}}_{\Theta}$ is the reduced system state vector. If $\boldsymbol{\xi}_{\Theta}\left(0\right)=\boldsymbol{\bar{\xi}}_{\Theta}\left(0\right)=0$, $\Vert\boldsymbol{\xi}_{\Theta}\Vert=\mathcal{O}\left(\frac{1}{\omega_\Omega}\right)$ will hold uniformly according to the Tikhonov's theorem, which implies ${\xi}_\sigma=\mathcal{O}\left(\frac{1}{\omega_\sigma}\right)$.  It is easy to get $\sigma_c-\sigma_d=\mathcal{O}\left(\frac{1}{\omega_\sigma}\right)$ according to Lemma \ref{Lem: 2ndOrderEstProperties}.

\section*{Appendix D: Proof of Proposition \ref{Prop: Chap5_InnerStab_NoZeroDer}}
When $\mathbf{\dot{d}}_\Theta \neq 0$ and $\boldsymbol{\dot{D}}_\tau \neq 0$, the closed-loop error dynamics (\ref{Eq: Chap5_ClosedInnerLoop}) will be rewritten as 
\begin{equation}
\left\{\begin{array}{ll}
\boldsymbol{\dot{\varepsilon}}_\Theta &= -\mathbf{K}_{\Theta}\boldsymbol{\varepsilon}_\Theta +\mathbf{G}\mathbf{e}_\Omega-\mathbf{\widetilde{d}}_\Theta \\
\mathbf{\dot{e}}_\Omega &= -\mathbf{K}_{\Omega}{\mathbf{e}}_\Omega-\mathbf{C}_\Omega\mathbf{G}^T\boldsymbol{\varepsilon}_\Theta-\mathbf{\widetilde{d}_\tau}
\end{array}\right.
\qquad
\left\{\begin{array}{ll}
\boldsymbol{\mathcal{T}}_\Theta\mathbf{\dot{\widetilde{d}}}_\Theta&=-\mathbf{\widetilde{d}}_\Theta -\boldsymbol{\mathcal{T}}_\Theta\mathbf{\dot{d}}_\Theta \\
\boldsymbol{\mathcal{T}}_\Omega\mathbf{\dot{\widetilde{d}}}_\tau&=-\mathbf{\widetilde{d}_\tau}-\boldsymbol{\mathcal{T}}_\Omega\mathbf{\dot{d}}_\tau
\end{array}\right.  \label{Eq: Chap5_ClosedInnerLoop1}
\end{equation}
If the time constants $\boldsymbol{\mathcal{T}}_\Theta$ and $\boldsymbol{\mathcal{T}}_\Omega$ are chosen to be sufficiently small,  (\ref{Eq: Chap5_ClosedInnerLoop1}) will be a standard perturbation model whose reduced system is given in (\ref{Eq: Chap5_ClosedInnerLoopStandard}). The reduced system  (\ref{Eq: Chap5_ClosedInnerLoopStandard}) is apparently exponentially stable. Since both $\mathbf{\dot{d}}_\Theta$ and $\boldsymbol{\dot{D}}_\tau$ are bounded, their impact will be diminished by reducing $\boldsymbol{\mathcal{T}}_\Theta$ and $\boldsymbol{\mathcal{T}}_\Omega$, respectively. Since $\boldsymbol{\Omega}_c\left(0\right)=\boldsymbol{\Omega}_c\left(0\right)$ and $\boldsymbol{\xi}_\Theta\left(0\right)=0$, one has $\boldsymbol{{\varepsilon}}_\Theta\left(0\right)=\boldsymbol{\bar{e}}_\Theta\left(0\right)$ and $\boldsymbol{{e}}_\Omega\left(0\right)=\boldsymbol{\bar{e}}_\Omega\left(0\right)$. According to the Tikhonov's theorem for a standard perturbation model, one is able to conclude that $\Vert{\boldsymbol{\varepsilon}}_\Theta-\boldsymbol{\bar{e}}_\Theta\Vert= \mathcal{O}\left(\epsilon_1\right)$ and $\Vert\boldsymbol{e}_\Omega-\boldsymbol{\bar{e}}_\Omega\Vert= \mathcal{O}\left(\epsilon_1\right)$ will uniformly hold. According to the definition of the oder of magnitude, it is easy to find that $\epsilon_1=\max\left\{\mathcal{T}_\mu\text{, }\mathcal{T}_\alpha\text{, }\mathcal{T}_\beta\text{, }\mathcal{T}_p\text{, }\mathcal{T}_q\text{, }\mathcal{T}_r\right\}$.

Furthermore, with the consideration of (\ref{Eq: Chap5_SingularPerturb_Theta}) and (\ref{Eq: Chap5_ClosedInnerLoop1}), we have
\begin{equation}
\left\{\begin{array}{ll}
\boldsymbol{\dot{e}}_\Theta &= -\mathbf{K}_{\Theta}\boldsymbol{e}_\Theta +\mathbf{G}\mathbf{e}_\Omega+\mathbf{G}\left(\boldsymbol{\Omega}_c-\boldsymbol{\Omega}_d\right) -\mathbf{\widetilde{d}}_\Theta  \\
\mathbf{\dot{e}}_\Omega &= -\mathbf{K}_{\Omega}{\mathbf{e}}_\Omega-\mathbf{C}_\Omega\mathbf{G}^T\left(\boldsymbol{e}_\Theta-\boldsymbol{\xi}_{\Theta}\right)-\mathbf{\widetilde{d}_\tau}\\
\boldsymbol{\dot{\xi}}_\Theta&=-\mathbf{K}_{\Theta}\boldsymbol{\xi}_{\Theta}+\mathbf{G}\left(\boldsymbol{\Omega}_c-\boldsymbol{\Omega}_d\right)   
\end{array}\right. \label{Eq: Chap5_InnerPerturb_Full1}
\end{equation}
\begin{equation}
\left\{\begin{array}{ll}
\frac{1}{\omega_\Omega}\boldsymbol{\dot{\Omega}}_c&=\boldsymbol{\bar{\Omega}}_{cdot} \\
\frac{1}{\omega_\Omega}\boldsymbol{\dot{\bar{\Omega}}}_{cdot} &= -\frac{2\zeta_\Omega}{\omega_\Omega}\boldsymbol{\dot{\bar{\Omega}}}_{cdot}-\left(\boldsymbol{\Omega}_c-\boldsymbol{\Omega}_d\right) \\
\end{array}\right.
\qquad
\left\{\begin{array}{ll}
\boldsymbol{\mathcal{T}}_\Theta\mathbf{\dot{\widetilde{d}}}_\Theta&=-\mathbf{\widetilde{d}}_\Theta -\boldsymbol{\mathcal{T}}_\Theta\mathbf{\dot{d}}_\Theta\\
\boldsymbol{\mathcal{T}}_\Omega\mathbf{\dot{\widetilde{d}}}_\tau&=-\mathbf{\widetilde{d}_\tau}-\boldsymbol{\mathcal{T}}_\Omega\mathbf{\dot{d}}_\tau
\end{array}{ll}\right. \label{Eq: Chap5_InnerPerturb_Full2}
\end{equation}
Notice that Eq. (\ref{Eq: Chap5_InnerPerturb_Full2}) will perform as fast dynamics, if $\omega_\Omega$ is chosen sufficiently large and $\boldsymbol{\mathcal{T}}_\Theta$ and $\boldsymbol{\mathcal{T}}_\Omega$ are chosen sufficiently small. The reduced system composed of (\ref{Eq: Chap5_InnerPerturb_Full1}) and (\ref{Eq: Chap5_InnerPerturb_Full2}) is still (\ref{Eq: Chap5_ClosedInnerLoopStandard}). By picking $\boldsymbol{\xi}_{\Theta}\left(0\right)=0$, one has $\boldsymbol{e}_\Theta\left(0\right)=\bar{e}_\Theta$, so $\Vert{\boldsymbol{e}}_\Theta-\boldsymbol{\bar{e}}_\Theta\Vert= \mathcal{O}\left(\epsilon_2\right)$ will uniformly hold, where $\epsilon_2=\max\left\{\epsilon_1\text{, }\frac{1}{\omega_p}\text{, }\frac{1}{\omega_q}\text{, }\frac{1}{\omega_r}\right\}$. In addition, $\lim_{t\to\infty}\boldsymbol{\bar{e}}_\Theta\to0$, so $\boldsymbol{e}_\Theta$ will be ultimately bounded.

\section*{Acknowledgments}
This research work presented in this paper is supported by Natural Science and Engineering Research Council of Canada (NSERC) Discovery Grant 227674.

\section*{References} \label{Sec: Reference}
\bibliographystyle{aiaa}
\bibliography{./MyRefs}

\end{document}